\documentclass[11pt]{article}

\usepackage{fullpage}

\usepackage{microtype}
\usepackage{graphicx}
\usepackage{subcaption}
\usepackage{booktabs} 

\usepackage{hyperref}
\usepackage{enumitem}

\usepackage{tikz}
\usetikzlibrary{arrows.meta,calc}

\usepackage{amsmath}
\usepackage{amssymb}
\usepackage{mathtools}
\usepackage{amsthm}

\usepackage{caption}
\usepackage{subcaption}

\usepackage[capitalize,noabbrev]{cleveref}

\theoremstyle{plain}
\newtheorem{theorem}{Theorem}[section]

\newtheorem{lemma}[theorem]{Lemma}

\theoremstyle{definition}

\theoremstyle{remark}
\newtheorem{remark}[theorem]{Remark}

\usepackage[textsize=tiny]{todonotes}

\newcommand{\R}{{\mathbb{R}}}
\newcommand{\eat}[1]{}

\title{\bf Path Dependence under Adaptive AI Delegation}

\author{
	Lingxiao Huang\\
	Nanjing University
	\and
	Nisheeth K. Vishnoi \\
	Yale University
}

\date{}

\begin{document}
	
	\maketitle
	
	\begin{abstract}
		Repeated AI assistance can improve immediate task performance while reducing
		the skill available for future independent work. We develop a mathematical framework for
		this long-run tradeoff. The model tracks two state variables: a latent human
		skill level governing expected independent performance, and a delegation level
		representing the learner's evolving tendency to rely on AI. Skill changes
		through error-driven learning under practice and decay under delegation;
		delegation responds to observed performance, increasing when AI-assisted work
		appears to outperform independent work.
		We analyze the resulting dynamics and contrast them with fixed delegation.
		With fixed delegation, skill follows a one-dimensional learning-decay
		process with a single stable equilibrium. With adaptive delegation, the
		coupled system has two attracting equilibria separated by the stable manifold
		of an interior saddle. The existence and geometry of this separatrix require
		a global phase-plane analysis of the coupled dynamics. The system is
		path-dependent: small differences in initial skill or reliance can lead to
		different long-run outcomes.
		We use this characterization to show that AI assistance can improve
		short-run performance while producing worse long-run performance than a
		no-AI baseline. Increasing AI capability can enlarge the basin of attraction
		of the low-skill equilibrium, making delegation appear beneficial for longer
		while increasing the risk of eventual skill loss. 
		The qualitative picture is observed to persist across alternative specifications. Together, these results show that the risk is not AI assistance
		itself, but the coupling between performance-driven reliance and use-dependent skill change.
	\end{abstract}
	
	\newpage
	\tableofcontents
	\newpage
	
	\section{Introduction}
	\label{sec:intro}

	AI systems are increasingly used to support cognitive tasks such as writing,
	coding, tutoring, and problem solving
	\cite{Bai2023ChatGPTTC, microsoft, Bastani2025GenerativeAW,
		lehmann2025aimeetsclassroomlarge}. In many of these settings, AI assistance
	improves task performance, output quality, or throughput over a short horizon
	\cite{Bastani2025GenerativeAW, lehmann2025aimeetsclassroomlarge, microsoft}.
	The short-run view treats AI use as a performance-enhancing intervention.

	Repeated use changes the nature of the interaction. A learner who relies on
	AI across many related tasks is not only producing outputs; they are also
	changing the skill with which future outputs will be produced. Delegation
	today therefore changes both the immediate output and the human capacity that
	produces future outputs. The two effects can pull in opposite directions.
	Recent empirical work suggests that this distinction matters: studies of
	AI-assisted writing, programming, and mathematics report immediate gains in
	performance or productivity alongside reduced retention, weaker transfer, or
	worse post-test performance once assistance is removed
	\cite{Kosmyna2025YourBO, Bastani2025GenerativeAW,
		lehmann2025aimeetsclassroomlarge, Liu2026AIAR, Shen2026HowAI, Ejaz2025AIAC}.
	The effects are also heterogeneous across users
	\cite{Shin_2025, lehmann2025aimeetsclassroomlarge, Zhai2024TheEO}.
	
	These findings raise a structural question:
	\emph{how does adaptive AI reliance shape the long-run dynamics of human skill,
		and when do differences in initial skill and reliance lead to different
		long-run outcomes?}
	
	\paragraph{Related work.}
	The empirical studies above are often interpreted through cognitive
	offloading, automation bias, task substitution, or reduced cognitive
	engagement
	\cite{Parasuraman, Norman1994, RiskoGilbert2016, Ejaz2025AIAC}.
	These perspectives help explain why users may rely on AI in the moment, but
	they do not provide a dynamical account of how reliance and skill co-evolve
	over repeated use.
	
	A complementary theoretical direction studies AI delegation as a one-shot
	decision problem. A recent model of delegation and verification under AI
	\cite{huang2026delegation} considers an agent who optimizes how much to
	delegate and verify under task-level performance, effort, and verification
	tradeoffs. That framework explains workflow choice and institutional
	evaluation in a single decision environment. The present paper studies a
	different regime: delegation is an evolving reliance tendency formed across
	repeated tasks. The central object is therefore not within-task delegation, but
	the long-run feedback loop between performance-driven reliance and human skill
	formation.
	
	Mathematical models of learning and forgetting capture error-driven
	improvement and use-dependent decay under fixed task execution
	\cite{Ebbinghaus1913, NewellRosenbloom1981, Heathcote2000,
		Griffiths2008, RescorlaWagner1972, Pemantle2007}. Classical stochastic
	approximation connects such rules to limiting ODEs whose global behavior
	governs convergence \cite{Benaim}. These models typically treat practice as
	exogenous or fixed. They do not capture settings in which reliance on AI
	responds to observed performance, thereby endogenously changing the practice
	that drives future skill. Related work on performative prediction studies how
	deployed models reshape data distributions \cite{perdomo20a, Milli}; here
	the feedback loop operates through human skill and reliance rather than model
	retraining. Thus, the novelty is not learning decay or adaptive reliance in isolation, but
	the basin structure created when reliance adapts to short-run performance while
	skill evolves through use and non-use.
	
	\paragraph{Our contributions.}
	Answering the question above requires a model in which reliance and skill
	evolve together. We develop such a framework for repeated AI-assisted
	learning, with two coupled state variables: a latent skill level \(\theta(t)\)
	governing expected independent performance, and a delegation level \(p(t)\)
	representing the learner's tendency to rely on AI. Skill evolves through
	error-driven learning under practice and decay under delegation, while
	delegation responds to the short-run performance comparison between
	independent and AI-assisted work. This separation is central: reliance adapts
	to observed performance, while skill evolves through the consequences of use
	and non-use.
	Our main results are as follows.
	\begin{itemize}
		
		\item We characterize the equilibrium structure of the resulting dynamical
		system. With fixed delegation, the dynamics reduce to a one-dimensional
		learning-decay process with a single stable equilibrium. With adaptive
		delegation, the system has two attracting equilibria: a high-skill,
		low-delegation equilibrium and a low-skill, high-delegation equilibrium,
		together with an interior saddle (Theorem~\ref{thm:convergence}).
		
		\item We characterize which initial conditions lead to which long-run outcome.
		The stable manifold of the interior saddle is a global basin boundary in the
		\((\theta,p)\) state space, separating initial skill-delegation pairs that
		converge to sustained learning from those that converge to persistent
		delegation and low skill (Theorem~\ref{thm:saddle}). This yields path
		dependence: small differences in initial skill or reliance can lead to
		different long-run outcomes.
		
		\item We show that short-run performance gains can mask long-run losses. AI
		assistance can strictly improve performance initially while producing worse
		long-run performance than the no-AI baseline
		(Theorem~\ref{thm:performance}). This effect arises because
		performance-driven reliance reshapes the practice that drives future skill,
		not from an explicit incentive misalignment.
		
		\item We analyze how the basin structure changes under AI capability and model
		variations. Increasing AI capability reshapes the basin boundary
		(Theorem~\ref{thm:monotonicity}), enlarging the set of initial conditions that
		converge to the low-skill equilibrium and extending the duration of apparent
		performance gains. The same qualitative structure is observed under jagged or
		misperceived AI performance, alternative objectives, multiple skills, learning
		from AI during delegation, and richer learner-AI dynamics
		(Sections~\ref{sec:extension} and~\ref{sec:details_extension}); the geometric
		view also suggests interventions that shift the basin boundary toward
		sustained-skill outcomes (Remark~\ref{remark:intervention} and
		Section~\ref{sec:intervention}).
		
	\end{itemize}
	Together, the framework and these results identify the possible long-run
	regimes, characterize which initial skill and reliance levels lead to each,
	and show how AI capability reshapes this partition. The source of long-run
	risk is not AI assistance itself, but the dynamic coupling between
	performance-driven reliance and use-dependent skill change.

	\section{Model}
	\label{sec:model}
	
	\begin{figure*}[t]
		\centering
		\includegraphics[width=\linewidth]{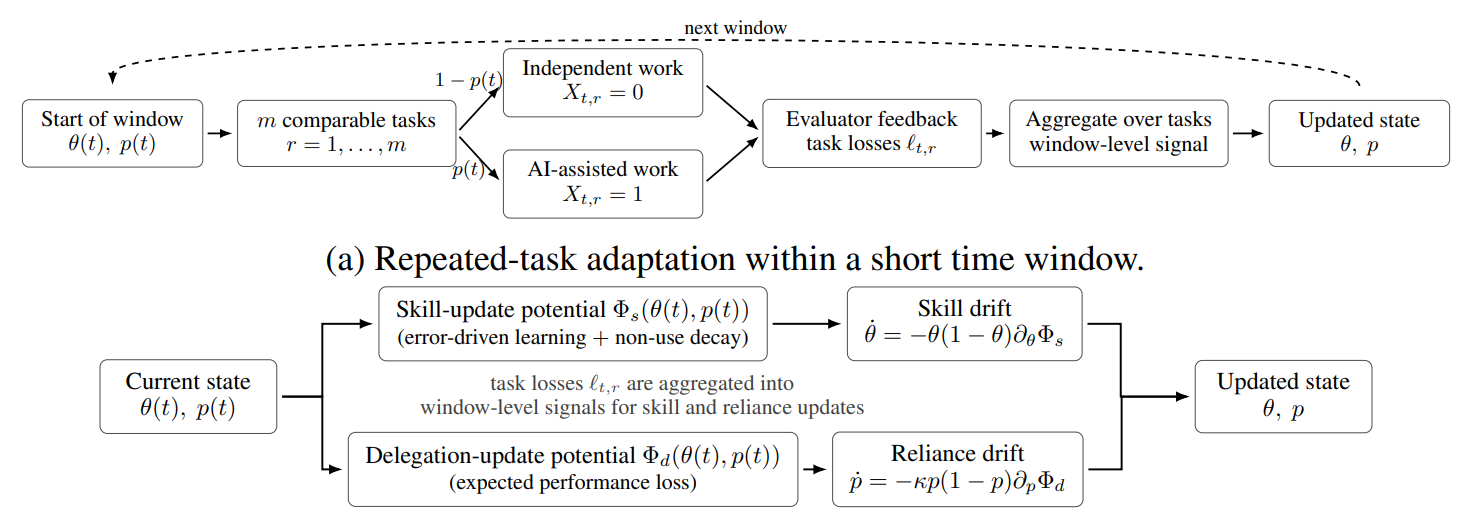}
		\caption{Overview of the repeated-task adaptation model.
			In each short time window \(t\), the learner faces many comparable tasks and
			delegates a fraction \(p(t)\) of them. Evaluator feedback is aggregated
			across tasks and enters the skill-update and delegation-update mechanisms. The
			former combines error-driven learning under practice with non-use decay, while
			the latter updates reliance through relative performance feedback.}
		\label{fig:overview}
		
	\end{figure*}
	
	We study a repeated task setting in which a learner faces many comparable
	tasks over time. In each short time window, the learner completes a fraction
	of tasks independently, delegates the remaining fraction to AI, and receives
	feedback on the resulting outputs. The window is chosen short enough that the
	learner's skill \(\theta(t)\) and delegation level \(p(t)\) can be treated as
	fixed within the window, but long enough to average over many task instances.
	The learner's state is therefore \((\theta(t),p(t))\). We focus on \emph{error-driven}
	learners \cite{RescorlaWagner1972,erev}, whose skill and reliance updates are
	driven by observed performance feedback.
	
	\paragraph{Skill representation.}
	We represent task-specific skill by a latent parameter
	\(\theta\in[0,1]\) that governs expected independent performance on a repeated
	class of comparable tasks. The variable \(\theta\) is not the realized score on
	a particular task; realized performance may vary with task difficulty,
	attention, or execution noise.
	This abstraction is standard in latent-trait models of ability. In item
	response theory, ability is modeled as a continuous latent trait that combines
	with item difficulty to determine response probabilities
	\cite{Lord1980ApplicationsOI,EmbretsonReise2000}. Our use of
	\(\theta\) is analogous, normalized so that \(\theta=1\) denotes optimal
	expected independent performance.
	The scalar representation is an analytical simplification, not a claim that
	ability is intrinsically one-dimensional; Section~\ref{sec:multiple_skill}
	analyzes a multiple-skill extension.
	
	Let \(\theta(t)\) denote the learner's skill during window \(t\). Let
	\(\theta_a\in[0,1]\) denote the AI's \emph{effective skill}, defined through
	its expected loss under the same task objective; under the squared-loss
	baseline, random AI performance \(S_a\) is represented by
	\((1-\theta_a)^2=\mathbb{E}[(1-S_a)^2]\).
	
	\paragraph{Delegation representation.}
	Let \(p(t)\in[0,1]\) denote the delegation level, i.e., the fraction of tasks
	delegated to AI in window \(t\). For a task sampled from this window, let
	\(X\sim\mathrm{Bern}(p(t))\), where \(X=1\) indicates delegation to AI and
	\(X=0\) indicates independent work. 
	Thus \(p(t)\) is an aggregate reliance tendency shaped by feedback. The
	Bernoulli notation represents within-window task fractions, not a one-shot
	delegation decision.

	\paragraph{Performance loss.}
	Task performance is evaluated against a fixed target by an external evaluator,
	which may represent a teacher, benchmark, or grader. Let \(\ell(\theta)\)
	denote the expected performance loss of an output produced with latent skill
	\(\theta\in[0,1]\). We assume that \(\ell(\theta)\) is monotonically
	decreasing, with \(\ell(1)=0\). A canonical example is the power loss
	\(\ell(\theta)=(1-\theta)^z\) for \(z\ge 1\), with the squared loss \(z=2\)
	used as the baseline.
	The instantaneous performance loss for a task with delegation indicator \(X\)
	is
	$\ell(\theta(t),X) := (1-X)\ell(\theta(t)) + X\ell(\theta_a)$.
	Averaging over \(X\sim\mathrm{Bern}(p(t))\) gives
	\begin{align}
		\label{eq:metric}
		\ell(\theta(t),p(t))
		=
		(1-p(t))\ell(\theta(t)) + p(t)\ell(\theta_a).
	\end{align}
	When \(\ell(\theta(t))\ge \ell(\theta_a)\), increasing delegation improves
	short-run performance.\footnote{The loss \(\ell\) may implicitly reflect
		accuracy, time, or effort: any factor affecting observed task quality enters
		through the feedback signal.}

	\paragraph{Skill update.}
	After receiving feedback, the learner's skill changes depending on how the
	task was completed. If the learner works independently (\(X=0\)), feedback
	generates an error-driven learning signal. If the task is delegated to AI
	(\(X=1\)), the learner does not directly practice the underlying task. We model
	non-use as drift toward a default skill level \(\theta_d\in[0,1]\), consistent
	with learning and forgetting models in which performance improves through
	practice and weakens without continued practice or reinforcement
	\cite{Ebbinghaus1913,NewellRosenbloom1981,Heathcote2000,pavlik}.
	
	We encode these effects by the skill-update potential
	\[
	\Phi_s(\theta(t);X)
	=
	(1-X)\ell(\theta(t))
	+
	\Delta X g(\theta(t),\theta_d),
	\]
	where \(\Delta>0\) controls the strength of decay and
	\(g(\theta,\theta_d)\ge 0\) measures distance from the default skill. This
	potential is distinct from the performance loss: its first term captures
	learning from practice, while its second term captures decay from non-use.
	The function \(g\) is left general at this stage; the squared distance used in
	the main analysis is chosen for tractability, and
	Section~\ref{sec:extension} shows that the qualitative phase structure is not
	specific to this choice.
	
	Averaging over \(X\sim\mathrm{Bern}(p(t))\) within window \(t\) gives
	\begin{align}
		\label{eq:skill_update_potential}
		\Phi_s(\theta(t);p(t))
		=
		(1-p(t))\ell(\theta(t))
		+
		\Delta p(t)g(\theta(t),\theta_d).
	\end{align}
	To keep the latent skill in \([0,1]\), we use the one-dimensional information
	geometry \(F(\theta)=1/(\theta(1-\theta))\), the Fisher metric of a Bernoulli
	parameter. The induced natural-gradient flow preconditions ordinary gradient
	descent by \(F(\theta)^{-1}=\theta(1-\theta)\), yielding multiplicative or
	replicator-style dynamics
	\cite{Amari1998,HofbauerSigmund1998}. This is a
	geometric modeling choice for bounded latent skill, not an empirical claim
	about the exact form of human learning.
	Applying this flow to \(\Phi_s\) gives
	\begin{align}
		\label{eq:general_skill_update}
		\dot{\theta}
		=
		-F(\theta)^{-1}\partial_\theta \Phi_s
		=
		-\theta(t)(1-\theta(t))
		\left[
		(1-p(t))\partial_\theta\ell(\theta(t))
		+
		\Delta p(t)\partial_\theta g(\theta(t),\theta_d)
		\right].
	\end{align}
	Since \(\ell(\theta)\) decreases with skill, the first term improves skill
	under independent work; the second term pushes skill toward \(\theta_d\) under
	delegation.

	\paragraph{Delegation update.}
	The learner adjusts reliance \(p(t)\) based on observed performance across
	tasks. If AI-assisted outputs perform better than independent work, reliance
	increases; otherwise, it decreases. We encode this adjustment through the
	delegation-update potential
	\begin{align}
		\label{eq:delegation_update_potential}
		\Phi_d(\theta(t),p(t))
		:=
		\ell(\theta(t),p(t)).
	\end{align}
	Unlike the skill-update potential, which separates learning from practice and
	decay from non-use, \(\Phi_d\) measures which mode of task completion gives
	lower short-run loss.
	
	Since \(p(t)\) is a bounded reliance variable, we use the multiplicative
	geometry \(F(p)=1/(p(1-p))\). The induced natural-gradient flow is
	\begin{align}
		\label{eq:general_delegation_update}
		\dot p
		=
		-\kappa F(p)^{-1}\partial_p\Phi_d(\theta(t),p(t))
		=
		\kappa p(t)(1-p(t))
		\bigl(\ell(\theta(t))-\ell(\theta_a)\bigr),
	\end{align}
	where \(\kappa>0\) controls the relative speed of delegation adaptation.
	Thus reliance increases when AI has lower expected loss than independent work,
	and decreases when independent work has lower expected loss.

	\paragraph{Dynamics.}
	We now specialize the general skill and delegation updates to the main
	specification analyzed in the paper:
	\[
	\ell(\theta)=(1-\theta)^2,
	\qquad
	g(\theta,\theta_d)=(\theta-\theta_d)^2.
	\]
	This yields the coupled dynamics
	\begin{align}
		\label{eq:ODE}
		\dot{\theta}
		= 2\theta(1-\theta)\bigl((1-p)(1-\theta)+\Delta p(\theta_d-\theta)\bigr), \ \
		\dot{p}
		= \kappa p(1-p)\bigl((1-\theta)^2-(1-\theta_a)^2\bigr).
	\end{align}
	Throughout the analysis, we compare the long-run behavior of~\eqref{eq:ODE}
	across different initial conditions \((\theta_0,p_0)\) and against the no-AI
	case. The initial delegation level \(p_0\) reflects the learner's prior belief
	about the relative effectiveness of AI-assisted and independent work, formed
	through prior experience or external information.
	
	Unless otherwise stated, we consider the regime
	\(\theta_a,\theta_0\in[\theta_d,1]\). For the main analysis, we normalize
	\(\theta_d=0\). We also rescale time to remove the factor \(2\) in the skill
	equation and absorb the resulting constant into the free relative-speed
	parameter \(\kappa\). This gives
	\begin{align}
		\label{eq:ODE_simplified}
		\dot{\theta}
		= \theta(1-\theta)\left((1-p)(1-\theta)-\Delta p\theta\right),\ \
		\dot{p}
		= \kappa p(1-p)\left((1-\theta)^2-(1-\theta_a)^2\right).
	\end{align}
	The general ODE~\eqref{eq:ODE} is derived and analyzed in
	Section~\ref{sec:full}; stochastic and biased delegation updates are discussed
	in Section~\ref{sec:SDE}. The model is intended to characterize
	qualitative dynamics rather than to serve as a calibrated structural model of
	human behavior. 
	The functional forms above make the phase-plane analysis tractable; the same
	phase structure is observed to persist under the alternative specifications studied in
	Section~\ref{sec:extension}, including alternative losses and decay functions,
	jagged AI performance, noisy and biased delegation updates, multiple skills,
	and learning from AI during delegation.
	
	\paragraph{Fixed delegation and no-AI case.}
	When the delegation level is held fixed at \(p(t)\equiv p_0\), the system
	reduces to the one-dimensional skill dynamics
	\begin{align}
		\label{eq:fixed_delegation}
		\dot{\theta}
		=
		\theta(1-\theta)\bigl((1-p_0)(1-\theta)-\Delta p_0\theta\bigr).
	\end{align}
	The no-AI case corresponds to \(p_0=0\). For any fixed
	\(p_0\in[0,1]\), \eqref{eq:fixed_delegation} has the unique stable equilibrium
	$
	\theta^\infty(p_0)=\frac{1-p_0}{1-p_0+\Delta p_0}$.
	Thus every trajectory with \(\theta_0\in(0,1)\) converges to
	\(\theta^\infty(p_0)\), and \(\theta^\infty(0)=1\).
	This fixed-delegation case separates ordinary non-use decay from the adaptive
	effect studied below. With fixed delegation, the model is one-dimensional and
	has a single stable equilibrium. Path dependence and multiple basins arise only
	when delegation adapts jointly with skill. Figure~\ref{fig:phase_space}
	contrasts the no-AI and adaptive-delegation phase portraits.

	\section{Theoretical results}
	\label{sec:result}

	\begin{figure*}[t]
		\centering
		
		\begin{subfigure}[t]{0.31\textwidth}
			\centering
			\includegraphics[width=\linewidth]{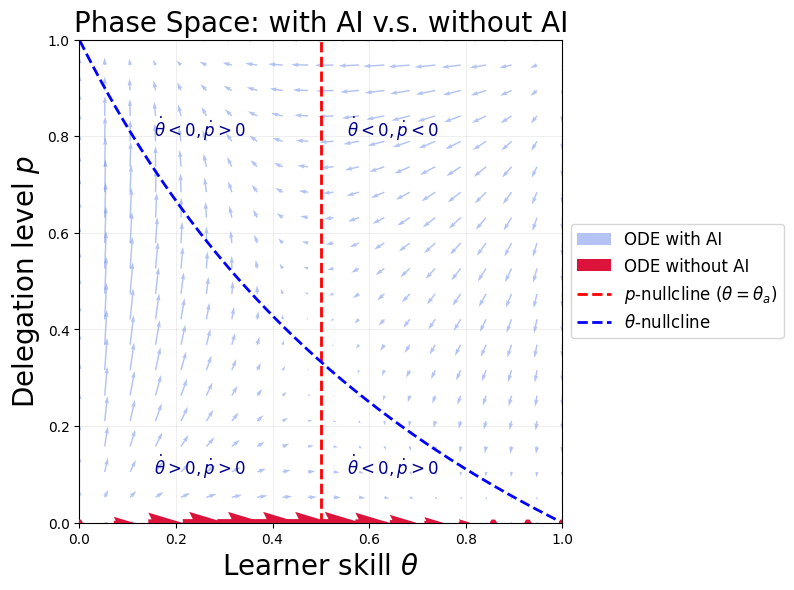}
			\caption{ Fixed versus adaptive delegation}
			\label{fig:phase_space}
		\end{subfigure}
		\hfill
		\begin{subfigure}[t]{0.3\textwidth}
			\centering
			\includegraphics[width=\linewidth]{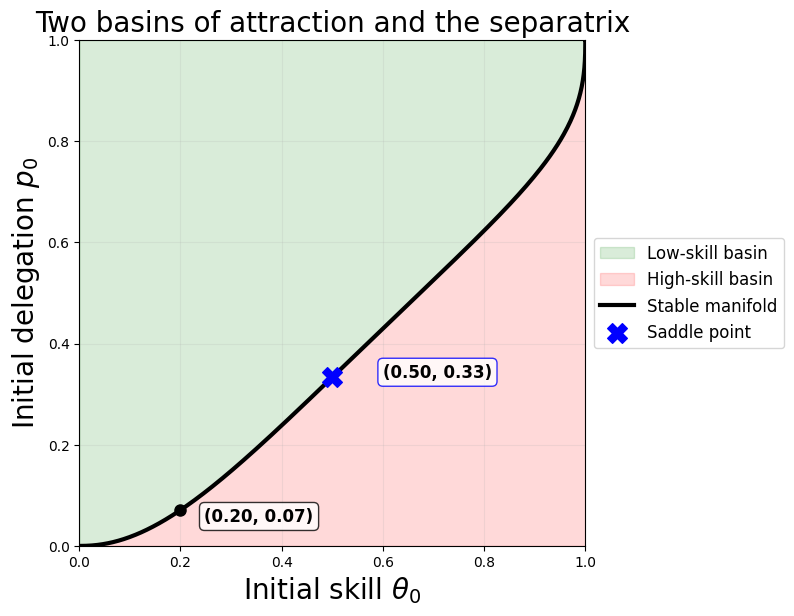}
			\caption{ Stable manifold and basin partition}
			\label{fig:manifold}
		\end{subfigure}
		\hfill
		\begin{subfigure}[t]{0.3\textwidth}
			\centering
			\includegraphics[width=\linewidth]{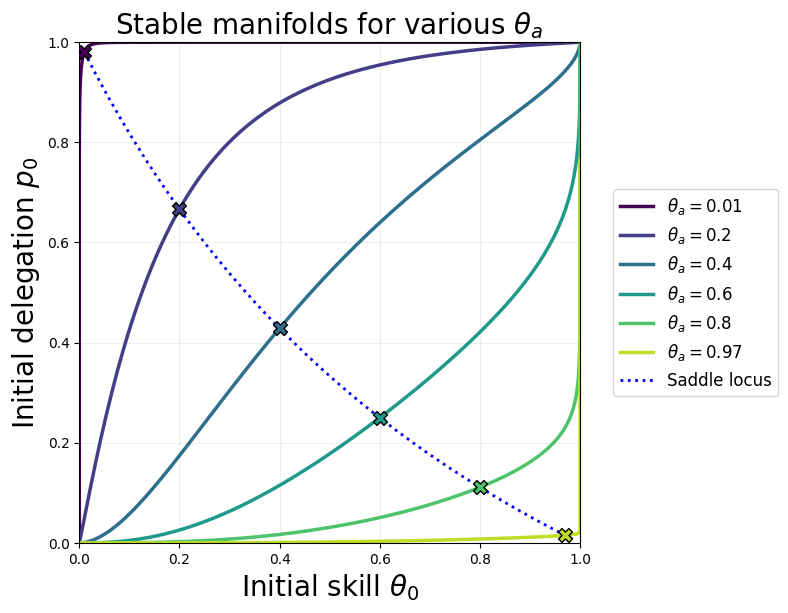}
			\caption{ AI skill reshapes the basin boundary}
			\label{fig:stable_theta_a}
		\end{subfigure}
		
		\caption{ Phase portraits for the adaptive-delegation dynamics.
			(a) Fixed delegation has a single attracting skill level, while adaptive
			delegation creates competing attracting regimes.
			(b) The stable manifold of the interior saddle separates the high-skill and
			low-skill basins.
			(c) Increasing AI skill shifts the basin boundary downward, expanding the
			low-skill basin. Default parameters are
			\((\theta_a,\kappa,\Delta)=(0.5,3,2)\) unless varied.}
		\label{fig:ODE}
	\end{figure*}
	
	We analyze the long-run behavior of the adaptive-delegation
	dynamics~\eqref{eq:ODE_simplified} on the interior \((0,1)^2\); basic
	well-posedness facts are proved in Section~\ref{sec:ODE}. We begin with the
	fixed-point structure.
	Fixed delegation already captures the direct effect of reduced practice on
	long-run skill. The results below show what changes when practice is
	endogenous: the two-way coupling between skill and reliance creates an
	interior saddle whose stable manifold separates two long-run regimes.

	\begin{theorem}[\bf{Equilibria of ODE~\eqref{eq:ODE_simplified}}]
		\label{thm:convergence}
		Let $\theta_a, \kappa, \Delta$ be the parameters of ODE~\eqref{eq:ODE_simplified}.
		When \(\theta_a\in(0,1)\), the system has two corner equilibria,
		\((1,0)\) and \((0,1)\), that are locally asymptotically stable relative to
		\([0,1]^2\); two repelling corner equilibria, \((0,0)\) and \((1,1)\); and one
		hyperbolic interior saddle
		$
		(\theta^\dagger,p^\dagger)
		=
		\left(\theta_a,\frac{1-\theta_a}{1-(1-\Delta)\theta_a}\right).
		$
		The stable corner equilibria are non-hyperbolic because of the boundary
		factors in the vector field.
	\end{theorem}
	
	\noindent
	The key qualitative change is the coexistence of high-skill and low-skill
	attracting regimes. The low-skill equilibrium \((0,1)\) has no analog in the
	no-AI baseline or the fixed-delegation dynamics of Section~\ref{sec:model}.
	The remaining question is which initial conditions converge to which
	attractor. Theorem~\ref{thm:saddle} identifies the basin boundary as the
	stable manifold of the interior saddle.

	\begin{theorem}[\bf{Basin partition by the stable manifold}]
		\label{thm:saddle}
		Let $\theta_a, \kappa, \Delta$ be the parameters of ODE \eqref{eq:ODE_simplified}.
		Let $(\theta^\dagger,p^\dagger)$ be the interior saddle point guaranteed by Theorem \ref{thm:convergence}.
		There exists a nondecreasing differentiable function $\psi:(0,1)\to(0,1)$
		that extends continuously to $[0,1]$ with $\psi(0)=0$ and $\psi(1)=1$,
		whose graph coincides with the one-dimensional stable manifold of $(\theta^\dagger,p^\dagger)$
		within $(0,1)^2$.
		The curve $p = \psi(\theta)$ partitions the state space $[0,1]^2$ into two distinct basins of attraction:
		Initial states with $p_0>\psi(\theta_0)$ converge to $(0,1)$, while those
		with $p_0<\psi(\theta_0)$ converge to $(1,0)$.
	\end{theorem}
	
	\noindent
	Theorem~\ref{thm:saddle} gives the path-dependence statement: for the same
	parameters, small changes in the initial state can place trajectories on
	opposite sides of the separatrix and lead to different long-run outcomes.
	Figure~\ref{fig:manifold} visualizes the stable manifold and the two basins.

	The monotonicity of \(\psi\) also shows why early-stage learners are more
	fragile: when \(\theta_0\) is small, even modest delegation can place the
	trajectory in the low-skill basin. Since \(\psi(\theta)<1\) for every
	\(\theta<1\), sufficiently heavy reliance can place even high-skill learners
	in the low-skill basin.
	The same basin structure extends to multidimensional skill: in a coordinate-wise extension, the basin boundary becomes a surface rather than a curve, but bistability and path dependence persist (Section \ref{sec:multiple_skill}).
	A piecewise-polynomial approximation \(\widetilde{\psi}(\cdot)\) of the basin
	boundary appears in Section~\ref{sec:basin_approximation}, Eq.~\eqref{eq:basin}.
	
	\paragraph{Effect of AI quality on the basin boundary.}
	We next characterize how the basin boundary changes with AI skill.
	
	\begin{theorem}[\bf AI skill monotonically sweeps the basin boundary]
		\label{thm:monotonicity}
		For every $\theta\in(0,1)$, the boundary value
		$\psi_{\theta_a}(\theta)$ is nonincreasing and continuously differentiable in $\theta_a$. 
		Moreover, the family $\{\psi_{\theta_a}\}_{\theta_a\in(0,1)}$ sweeps the state space in the following sense: for every $(\theta_0,p_0)\in(0,1)^2$, there exists
		$\theta_a\in(0,1)$ such that
		$
		p_0=\psi_{\theta_a}(\theta_0).
		$
	\end{theorem}
	
	\noindent
	Thus, increasing AI skill shifts the basin boundary downward and enlarges the
	set of initial conditions that converge to the low-skill equilibrium. More
	capable AI can therefore make delegation attractive for more learners while
	raising the threshold for sustained independent skill. Figure~\ref{fig:stable_theta_a}
	illustrates this deformation; variation with respect to \(\kappa\) and
	\(\Delta\) is studied in Section~\ref{sec:variation_other}.
	The same comparative statics imply that timing matters: early independent
	practice can move a learner toward the high-skill side of the boundary, while
	early delegation can move the learner toward the low-skill side.
	
	\begin{remark}[\bf Interventions as basin-boundary shifts]
		\label{remark:intervention}
		The basin characterization gives a geometric view of interventions:
		mechanisms that reduce early reliance, increase practice, or change delegation
		feedback can move the separatrix \(p=\psi(\theta)\), thereby changing which
		initial conditions converge to sustained skill. Representative mechanisms are
		studied in Section~\ref{sec:intervention}.
	\end{remark}

	\subsection{Analyzing performance loss across time}
	\label{sec:performance}
	
	\begin{figure*}[t]
		\centering
		
		\begin{subfigure}[t]{0.32\textwidth}
			\centering
			\includegraphics[width=\linewidth]{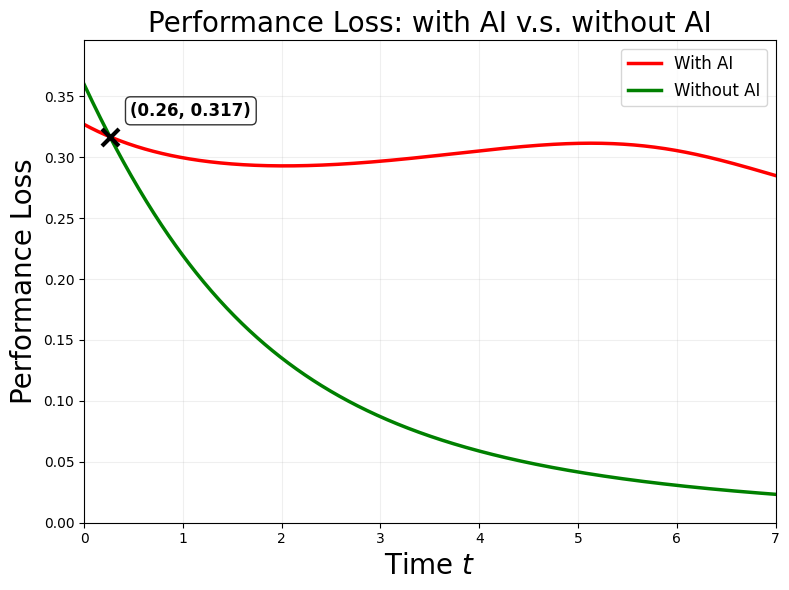}
			\caption{ Short-run gains and long-run losses}
			\label{fig:performance}
		\end{subfigure}
		\hfill
		\begin{subfigure}[t]{0.32\textwidth}
			\centering
			\includegraphics[width=\linewidth]{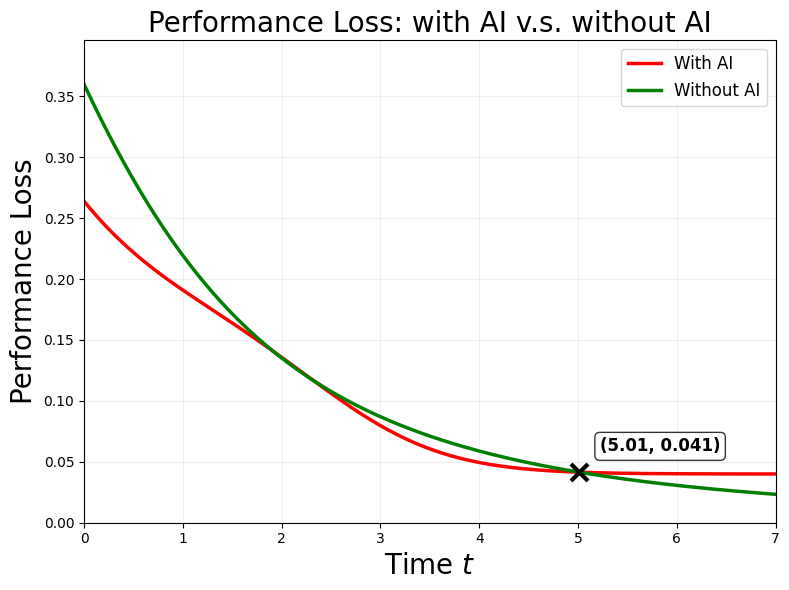}
			\caption{ Multiple crossings at high AI skill}
			\label{fig:performance_0.8}
		\end{subfigure}
		\hfill
		\begin{subfigure}[t]{0.3\textwidth}
			\centering
			\includegraphics[width=\linewidth]{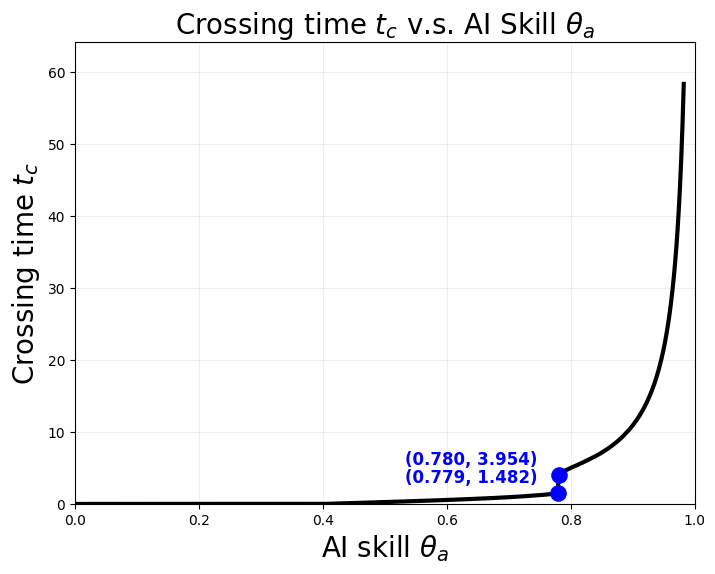}
			\caption{ Crossing time versus AI skill}
			\label{fig:crossing_AI}
		\end{subfigure}
		
		\caption{ Plots illustrating the instantaneous performance loss over time with and without AI, and how the crossing point varies with respect to AI skill $\theta_a$, with default settings $(\theta_a, \kappa, \Delta, \theta_0, p_0) = (0.5, 3, 2, 0.4, 0.3)$.}
		\label{fig:crossing}
	\end{figure*}
	
	Adaptive delegation may reduce instantaneous loss in the short run while
	lowering long-run skill, a pattern observed in recent experiments
	\cite{Kosmyna2025YourBO,Bastani2025GenerativeAW}. We formalize this reversal
	by comparing the same learner with and without AI access, holding initial
	skill fixed.
	
	For an initial state
	$(\theta_0,p_0)\in(0,1)^2$, let $\theta(t;\theta_0,p_0)$ and
	$p(t;\theta_0,p_0)$ denote the learner's skill and delegation level at time
	$t$, respectively, under adaptive delegation. In the no-AI baseline,
	$\theta(t;\theta_0,0)$ denotes the learner's skill trajectory starting from the
	same initial skill, yielding performance loss
	$(1-\theta(t;\theta_0,0))^2$. We define the \emph{performance gap} at time
	$t$ as
	\[
	G_\ell(t)
	:= \ell\left(\theta(t; \theta_0, p_0), p(t; \theta_0, p_0)\right)
	- (1 - \theta(t; \theta_0, 0))^2,
	\]
	so \(G_\ell(t)<0\) means that delegation yields lower instantaneous loss than
	the no-AI baseline. 
	If \(\theta_0\ge \theta_a\), delegation provides no initial performance
	advantage, and Lemma~\ref{lm:coupling} implies that the no-AI trajectory has
	weakly higher skill and lower loss for all \(t>0\). We therefore focus on
	\(\theta_0<\theta_a\).

	\begin{theorem}[\bf Short-term gains, long-term losses]
		\label{thm:performance}
		Let \(\theta_a,\kappa,\Delta\) be the parameters of
		ODE~\eqref{eq:ODE_simplified}. Fix an initial state
		\((\theta_0,p_0)\in(0,1)^2\) with \(\theta_0<\theta_a\). Let $
		t_c := \inf\left\{t\ge 0:\ \forall t'>t,\ G_\ell(t')>0\right\}$
		denote the final crossing time, after which adaptive delegation incurs higher
		loss than the no-AI baseline. Then \(t_c<\infty\) and
		\(\theta(t_c;\theta_0,0)\le \theta_a\).
	\end{theorem}

	\noindent
	The final crossing time \(t_c\) marks the point after which any apparent
	performance advantage from delegation has disappeared. 
	The post-crossing
	behavior depends on the basin of the adaptive-delegation trajectory. In the
	low-skill basin, the trajectory converges to persistent reliance and low skill,
	so the performance gap remains positive in the limit. In the high-skill basin,
	both trajectories converge to high skill, but the no-AI trajectory does so more
	directly; \(G_\ell(t)\) remains positive after \(t_c\) and vanishes
	asymptotically. This formalizes a mechanism consistent with cognitive
	debt \cite{Kosmyna2025YourBO}: short-term gains can be outweighed by
	structurally foregone practice.

	\paragraph{Effect of AI quality on the crossing time.}
	Figure~\ref{fig:crossing_AI} shows that \(t_c\) can grow sharply with
	\(\theta_a\): the no-AI learner may need longer to surpass a highly capable AI.
	Thus higher AI capability can extend the window of apparent gains while also
	expanding the low-skill basin (Theorem~\ref{thm:monotonicity}).
	
	\begin{remark}[\bf Empirical predictions]
		\label{remark:empirical}
		Although the analysis is theoretical, the state variables and parameters have
		direct counterparts in repeated-task data. Given task-level delegation
		indicators \(X_{t,r}\), realized losses \(\ell_{t,r}\), and AI-only benchmark
		losses \(\ell^a_{t,r}\), the delegation level \(p(t)\) can be estimated by
		the fraction of delegated tasks, while \(\theta(t)\) and \(\theta_a\) can be
		estimated from independent or randomized probe-task losses and AI-only
		benchmark losses. The decay rate \(\Delta\) and adaptation rate \(\kappa\)
		can then be estimated from finite differences on the resulting time series.
		The fitted system yields an estimated separatrix
		\(p=\widehat{\psi}(\theta)\). The model predicts that learners initialized on
		opposite sides of this boundary diverge in long-run reliance and independent
		skill, while learners initialized on the same side converge to the same
		long-run regime. Section~\ref{sec:usability} gives the full calibration
		procedure, a synthetic illustration, and falsifiable predictions.
	\end{remark}

	\subsection{Overview of the proofs}
	\label{sec:proof_idea}
	
	We present the main mathematical ingredients; full proofs appear in
	Section~\ref{sec:result_details}.

	\paragraph{Proof overview of Theorem~\ref{thm:convergence}.}
	Write the vector field in~\eqref{eq:ODE_simplified} as
	\[ 
	f(\theta,p)
	=
	\theta(1-\theta)\bigl((1-p)(1-\theta)-\Delta p\theta\bigr),
	\qquad
	h(\theta,p)
	=
	\kappa p(1-p)\bigl((1-\theta)^2-(1-\theta_a)^2\bigr).
	\]
	The equations \(f=h=0\) give the four corner equilibria through the boundary
	factors \(\theta(1-\theta)\) and \(p(1-p)\). In the interior, the nullclines
	are
	\[ 
	(1-p)(1-\theta)-\Delta p\theta=0,
	\
	(1-\theta)^2-(1-\theta_a)^2=0.
	\]
	Since \(\theta,\theta_a\in[0,1]\), the second nullcline gives
	\(\theta=\theta_a\), and the first then gives
	\[
	(\theta^\dagger,p^\dagger)
	=
	\left(
	\theta_a,\frac{1-\theta_a}{1-(1-\Delta)\theta_a}
	\right).
	\]
	Stability is determined from the Jacobian
	$
	J(\theta,p)
	=
	\begin{pmatrix}
		\partial_\theta f & \partial_p f\\
		\partial_\theta h & \partial_p h
	\end{pmatrix}.
	$
	At the corners, the off-diagonal entries vanish, and the stability within the square follows from the diagonal terms together with the one-dimensional boundary drifts. 
	This gives two locally asymptotically stable corners relative to the square,
	\((1,0)\) and \((0,1)\), and two repelling corners, \((0,0)\) and \((1,1)\).
	The stable corners are non-hyperbolic because one direction is tangent to the boundary.
	At the interior equilibrium, the nullcline conditions simplify the Jacobian.
	In particular,
	\[
	\partial_p h(\theta^\dagger,p^\dagger)=0,
	\partial_p f(\theta^\dagger,p^\dagger)<0,
	\partial_\theta h(\theta^\dagger,p^\dagger)<0.
	\]
	Therefore,
	\[
	\det J(\theta^\dagger,p^\dagger)
	=
	-\partial_p f(\theta^\dagger,p^\dagger)
	\partial_\theta h(\theta^\dagger,p^\dagger)
	<0.
	\]
	Thus, the interior equilibrium has one positive and one negative eigenvalue and is a saddle. 
	This negative determinant is the local signature of the
	coupling: higher delegation lowers the skill drift, while lower skill raises the delegation drift. 

	\paragraph{Proof overview of Theorem~\ref{thm:saddle}.}
	The main difficulty is global geometry: the stable manifold of the saddle point must form a single separatrix rather than fold, cycle or create disconnected basin boundaries.
	The proof has three steps: extend the local stable manifold globally, identify it as the basin boundary, and show that it is a monotone graph.

	By the stable manifold theorem~\cite{Perko2013}, the interior saddle has a one-dimensional local stable manifold \(W^s(\theta^\dagger,p^\dagger)\). 
	Let
	\(\Gamma\) be one branch and continue it backward under the flow:
	\[
	\Gamma^-=\{\varphi_{-t}(x):x\in\Gamma,\ t\ge0\}.
	\]
	The boundary of \([0,1]^2\) is invariant and solutions are unique, so an
	interior branch cannot hit the boundary in finite time. Suppose, for
	contradiction, that its backward continuation remains in the interior.
	Then its alpha-limit set is a compact invariant set in \((0,1)^2\).
	With the Dulac function
	\[
	D(\theta,p)=\frac{1}{\theta(1-\theta)p(1-p)},
	\]
	one computes
	\[
	\partial_\theta(Df)+\partial_p(Dh)
	=
	-\frac{(1-p)+\Delta p}{p(1-p)}<0
	\qquad\text{on }(0,1)^2.
	\]
	Hence, the Bendixson-Dulac criterion \cite{Perko2013} rules out periodic or homoclinic orbits in the interior. 
	Since the only interior equilibrium is the saddle, the closure of the backward branch intersects $\partial[0,1]^2$. 
	Applying this to both branches extends \(W^s\) globally.
	
	The global stable manifold is invariant, so trajectories cannot cross it.
	Since every interior trajectory converges to an equilibrium and the only
	attracting equilibria are \((1,0)\) and \((0,1)\), the two connected
	components of \((0,1)^2\setminus W^s\) are contained in the two attracting basins. 
	Conversely, any point on the interior boundary of either attracting
	basin cannot converge to an attracting or repelling corner, and therefore must converge to the saddle. 
	Thus, the resulting invariant stable manifold
	\(W^s(\theta^\dagger,p^\dagger)\) acts as a separatrix in the planar
	phase portrait \cite{Perko2013}.
	It remains to represent this boundary as \(p=\psi(\theta)\). 
	Set \(x=1-\theta\). 
	In \((x,p)\)-coordinates,
	\[
	\dot x
	=
	x(1-x)\bigl[\Delta p(1-x)-(1-p)x\bigr],
	\qquad
	\dot p
	=
	\kappa p(1-p)\bigl[x^2-(1-\theta_a)^2\bigr].
	\]
	The off-diagonal derivatives satisfy
	\[
	\frac{\partial \dot x}{\partial p}
	=
	x(1-x)\bigl[\Delta(1-x)+x\bigr]\ge0,
	\qquad
	\frac{\partial \dot p}{\partial x}
	=
	2\kappa p(1-p)x\ge0.
	\]
	Thus, the system is cooperative in \((x,p)\). 
	Since the off-diagonal derivatives are nonnegative, the system is cooperative;
	Kamke's comparison theorem for monotone dynamical systems
	\cite{smith1995monotone} gives the comparison principle: larger
	initial skill gap \(x_0\), larger initial delegation \(p_0\), or larger AI skill \(\theta_a\) weakly increase future skill gap and delegation (Lemma \ref{lm:coupling}). 
	It follows that the low-skill basin is upward closed in skill gap and delegation. 
	For each fixed \(\theta\), the slice
	\[
	I_\theta=\{p\in(0,1):(\theta,p)\in B_{\rm low}\}
	\]
	is therefore an upper interval. 
	The threshold point $(\theta,\psi(\theta))$ cannot lie in either open basin;
	therefore it belongs to the common basin boundary, which has already been
	identified with $W^s$.
	
	Defining
	\[
	\psi(\theta)=\inf I_\theta
	\]
	gives the basin boundary as a graph $p = \psi(\theta)$. 
	The upward-closed property implies that
	\(\psi\) is nondecreasing. 
	Differentiability follows because $W^s$ is a smooth embedded curve with no
	vertical tangencies. Away from the saddle, a vertical tangent would require
	$f(\theta,p)=0$, i.e., lying on the strictly decreasing $\theta$-nullcline,
	whereas $\psi$ is nondecreasing. At the saddle, the stable eigenvector has a
	nonzero $\theta$-component.

	\paragraph{Proof overview of Theorem~\ref{thm:monotonicity}.}
	Fix \(\theta_{a,1}<\theta_{a,2}\), and let \(\psi_1,\psi_2\) be the
	corresponding separatrices. 
	Suppose for contradiction that
	\(\psi_2(\theta^*)>\psi_1(\theta^*)\) for some \(\theta^*\). Choose
	\(p_0\) with
	$
	\psi_1(\theta^*)<p_0<\psi_2(\theta^*).
	$
	Then \((\theta^*,p_0)\) lies in the low-skill basin for \(\theta_{a,1}\), but
	in the high-skill basin for \(\theta_{a,2}\). 
	Hence,
	\(\theta^\infty(\theta_{a,1})=0\) while \(\theta^\infty(\theta_{a,2})=1\).
	But Lemma~\ref{lm:coupling} implies that increasing \(\theta_a\) weakly
	decreases \(\theta(t)\) for all \(t\geq 0\), so
	\(\theta^\infty(\theta_{a,2})\le \theta^\infty(\theta_{a,1})\), contradicting \(0<1\). 
	Thus, \(\psi_{\theta_a}(\theta)\) is nonincreasing in
	\(\theta_a\).
	
	Differentiability with respect to \(\theta_a\) follows from the
	parameter-dependent stable manifold theorem \cite{Perko2013}, since the saddle remains
	hyperbolic for every \(\theta_a\in(0,1)\). 
	The sweeping property follows from
	continuity in \(\theta_a\) and the endpoint limits
	\[
	\lim_{\theta_a\downarrow0}\psi_{\theta_a}(\theta)=1,
	\lim_{\theta_a\uparrow1}\psi_{\theta_a}(\theta)=0.
	\]
	These limits are obtained by comparing with the limiting systems
	\(\theta_a=0\) and \(\theta_a=1\). 
	At \(\theta_a=0\), delegation is
	nonincreasing and the high-skill basin fills the interior; at
	\(\theta_a=1\), delegation is nondecreasing and the low-skill basin fills the interior. 
	The intermediate value theorem then gives, for any
	\((\theta_0,p_0)\in(0,1)^2\), a value of \(\theta_a\) exists such that
	\(p_0=\psi_{\theta_a}(\theta_0)\).

	\paragraph{Proof overview of Theorem~\ref{thm:performance}.}
	Expanding
	$
	\ell(\theta,p)=(1-p)(1-\theta)^2+p(1-\theta_a)^2
	$
	gives  $G_\ell(t)=$
	\[
	(1-p(t;\theta_0,p_0))
	\Bigl[
	(1-\theta(t;\theta_0,p_0))^2
	-
	(1-\theta(t;\theta_0,0))^2
	\Bigr] 
	+p(t;\theta_0,p_0)
	\Bigl[
	(1-\theta_a)^2
	-
	(1-\theta(t;\theta_0,0))^2
	\Bigr].
	\]
	Lemma \ref{lm:coupling} gives that for any $t > 0$,
	\[
	\theta(t;\theta_0,p_0)\le \theta(t;\theta_0,0),
	\]
	so the first bracket is nonnegative. 
	The no-AI trajectory satisfies
	\(\dot\theta=\theta(1-\theta)^2\), hence increases monotonically from
	\(\theta_0\) to \(1\). 
	Since \(\theta_0<\theta_a\), there exists a finite
	\(t^*\) with
	\[
	\theta(t^*;\theta_0,0)=\theta_a.
	\]
	For \(t>t^*\), the second bracket is strictly positive, and
	\(p(t;\theta_0,p_0)>0\). 
	Hence, \(G_\ell(t)>0\) for all \(t>t^*\), so
	\(t_c\le t^*<\infty\).
	Moreover, once the no-AI trajectory exceeds
	\(\theta_a\), the performance gap is strictly positive; therefore
	\(\theta(t_c;\theta_0,0)\le\theta_a\).
	
	\section{Alternative specifications and robustness}
	\label{sec:extension}

	The main specification in~\eqref{eq:ODE_simplified} uses scalar skill,
	squared performance loss, and exact perception of fixed AI skill. We summarize
	variants that relax these choices. The main two-dimensional model is analyzed
	theoretically; the variants below are robustness checks, supported by exact
	reductions, local analysis, or simulations as indicated in
	Section~\ref{sec:details_extension}. Across them, the same qualitative phase
	structure is observed: adaptive reliance can create two attracting regimes
	separated by a basin boundary. 
	Thus, the basin mechanism is not tied to the
	particular squared-loss, scalar-skill specification used for the main phase-plane
	analysis. 
	Full derivations, simulations, and intervention variants appear in
	Section~\ref{sec:details_extension} and~\ref{sec:intervention}.
	
	\paragraph{Performance and perception.} 

	\noindent \emph{Jagged AI performance (Section \ref{sec:noisy_AI}).}
	If AI performance varies across tasks, the AI-side loss in the delegation
	update becomes \(\mathbb{E}_{s\sim\mu_a}[(1-s)^2]\). Defining
	$
	(1-\theta_a^{\rm eff})^2=\mathbb{E}_{s\sim\mu_a}[(1-s)^2]
	$
	reduces the expected dynamics to the main specification with
	\(\theta_a\) replaced by \(\theta_a^{\rm eff}\). Thus jaggedness enters through expected task loss while preserving the basin
	structure. 
	
	\noindent \emph{Biased perceived AI skill (Section \ref{sec:biased_AI_skill}).}
	Learners may update reliance using a perceived AI skill
	\(\widetilde{\theta}_a\) rather than the true effective skill \(\theta_a\).
	This replaces \((1-\theta_a)^2\) by
	\((1-\widetilde{\theta}_a)^2\) in the delegation update, while realized
	performance is still evaluated using \(\theta_a\). Overestimation expands the low-skill basin; underestimation has the opposite
	effect.
	
	\noindent \emph{Power-\(z\) loss (Section \ref{sec:model_robust}).}
	The squared loss can be replaced by
	\((1-\theta)^z\) for \(z\ge 1\). This changes both the
	performance comparison in the delegation update and the error-driven learning
	term in the skill update. The nullclines deform with \(z\), but the two-basin
	phase structure is observed to persist.
	
	\paragraph{Richer learner and AI states.} 

	\noindent \emph{Multiple skills (Section \ref{sec:multiple_skill}).}
	Skill can be multidimensional. For example, replacing \(\theta\) by
	\((\theta_1,\theta_2)\) and letting performance depend on
	\(\bar\theta=w_1\theta_1+w_2\theta_2\) yields a system over
	\((\theta_1,\theta_2,p)\). The basin boundary becomes a separating surface
	rather than a curve, but bistability and path dependence persist.
	
	\noindent \emph{Learning from AI during delegation (Section \ref{sec:learner_AI_extension}).}
	Delegation may provide a learning signal when the learner studies
	AI-generated outputs. Adding an AI-learning term to the skill update pulls
	\(\theta\) toward \(\theta_a\). In the main specification, the low-skill
	equilibrium shifts from \((0,1)\) to
	$
	\left(\frac{\beta\theta_a}{\beta+\Delta},1\right)$,
	but a low-skill, high-delegation regime is observed to persist.
	
	\noindent \emph{Improving AI capability (Section \ref{sec:improving_AI}).}
	If AI capability improves over time, the delegation update becomes
	time-dependent through \((1-\theta_a(t))^2\), while the skill update is
	unchanged. As \(\theta_a(t)\) increases, delegation becomes attractive over a larger
	range of skill levels; under a quasi-static interpretation, the basin
	structure shifts continuously as AI improves.
	
	\noindent \emph{Skill-dependent AI performance (Section \ref{sec:skill_dependent_AI}).}
	In some tasks, AI quality depends on the learner's ability to prompt or contextualize the output. 
	We model this by replacing the AI-side
	loss in \(\Phi_d\) with
	$(1-w\theta-(1-w)\theta_a)^2$,
	$w\in[0,1]$,
	while leaving the skill update unchanged. As \(w\) increases, the delegation
	signal depends more strongly on learner skill, reshaping the basin boundary.
	The resulting dynamics interpolate between AI performance that is independent
	of the learner and AI performance that depends on the learner's skill; the
	two-basin structure is observed to persist across this range.

	
	\begin{table}[ht]
		\centering
		\caption{Used notations}
		\begin{tabular}{ccc} 
			\toprule
			\textbf{Symbol} & \textbf{Domain} & \textbf{Description} \\
			\midrule
			$\theta_d$   &   $[0,1]$   & Default skill  \\
			$\theta_a$   &   $[0,1]$   & AI skill  \\
			$\theta_0 $  & $[0,1]$  & Initial skill of learner \\
			$\theta(t) $  & $[0,1]$  & Learner skill at time $t$ \\
			$\theta(t; \theta_0, p_0)$ & $[0,1]$  & Learner skill at time $t$ with initial state $(\theta_0, p_0)$ \\
			$p_0$ & $[0,1]$  & Initial delegation level of learner \\
			$p(t)$ & $[0,1]$  & Delegation level at time $t$ \\
			$p(t; \theta_0, p_0)$ & $[0,1]$  & Delegation level at time $t$ with initial state $(\theta_0, p_0)$ \\
			$\Delta$ & $\R_{\geq 0}$ & Decay rate \\
			$\kappa$ & $\R_{\geq 0}$ & Delegation rate \\
			$d(\theta , \theta' )$ & $[0,1]^2 \rightarrow \R_{\geq 0}$ & Distance between two skills $\theta$ and $\theta'$ \\
			$\ell(\theta , p )$ & $[0,1]^2 \rightarrow \R_{\geq 0}$ & Instantaneous performance loss w.r.t. skill $\theta$ and delegation level $p$ \\
			$\psi(\theta)$ & $[0,1] \rightarrow [0,1]$ & Mapping for the stable manifold from $\theta$ to a delegation level \\ 
			$G_\ell(t)$ & $\R_{\geq 0} \rightarrow \R$ & Performance gap between with and without AI at time $t$ \\
			$t_c$ & $\R_{\geq 0}$ & Maximum crossing time for performance loss \\
			\bottomrule
		\end{tabular}
		\label{tab:notation}
	\end{table}
	
	\section{Details for ODE \eqref{eq:ODE} from Section \ref{sec:model}}
	\label{sec:full}
	
	In this section, we first derive a dynamics from the learning procedure, and then show that the dynamics tracks to ODE \eqref{eq:ODE}.
	Finally, we prove some basic properties of ODE \eqref{eq:ODE}.
	Table \ref{tab:notation} summarizes notations used in this paper.
	
	\subsection{Deriving dynamics under AI assistance}
	\label{sec:derivation}
	
	We derive the dynamics corresponding to the repeated-task model in
	Section~\ref{sec:model}. Time is indexed by short windows $t=0,1,2,\ldots$.
	Within each window $t$, the learner faces $m$ comparable tasks. The learner's
	skill $\theta(t)$ and delegation level $p(t)$ are held fixed during this
	window, while the task-level delegation decisions vary across tasks.
	
	For each task $r\in[m]$, let
	$X_{t,r}\in\{0,1\}$ denote the delegation indicator, where $X_{t,r}=0$ means
	that the learner performs the task independently and $X_{t,r}=1$ means that
	the task is delegated to AI. We model
	\[
	X_{t,r}\sim \mathrm{Bern}(p(t)),
	\qquad r=1,\ldots,m,
	\]
	independently conditional on the current state $(\theta(t),p(t))$. Let
	\[
	\bar X_t := \frac{1}{m}\sum_{r=1}^m X_{t,r}
	\]
	denote the realized fraction of tasks delegated to AI in window $t$.
	
	The learning mechanism within one time window is as follows.
	
	\begin{itemize}
		\item \textbf{(Task completion stage)}
		For each task $r\in[m]$, the learner delegates the task to AI if
		$X_{t,r}=1$ and performs it independently if $X_{t,r}=0$.
		
		\item \textbf{(Evaluation stage)}
		The learner submits the resulting output to the teacher and receives
		task-level feedback. If the task is performed independently, the performance
		loss is $\ell(\theta(t))$; if the task is delegated to AI, the performance
		loss is $\ell(\theta_a)$. Thus the realized task-level performance loss is
		\[
		\ell_{t,r}
		=
		(1-X_{t,r})\ell(\theta(t)) + X_{t,r}\ell(\theta_a).
		\]
		
		\item \textbf{(Skill update)}
		Independent task execution produces an error-driven learning signal, while
		delegation produces non-use decay toward the default skill $\theta_d$. We
		encode these two effects by the task-level skill-update potential
		\[
		\Phi_{s,t,r}
		=
		(1-X_{t,r})\ell(\theta(t))
		+
		\Delta X_{t,r} g(\theta(t),\theta_d).
		\]
		Averaging over the $m$ tasks in the window gives
		\[
		\bar\Phi_s(\theta(t))
		=
		\frac{1}{m}\sum_{r=1}^m \Phi_{s,t,r}.
		\]
		Using the multiplicative-update metric
		$F_\theta(\theta)=1/(\theta(1-\theta))$, the skill update is
		\[
		\theta(t+1)
		=
		\theta(t)
		-
		\eta \theta(t)(1-\theta(t))
		\partial_\theta \bar\Phi_s(\theta(t)).
		\]
		
		\item \textbf{(Delegation update)}
		The delegation level is updated according to the relative performance of
		AI-assisted and independent work. The delegation-update potential is the
		expected performance loss
		\[
		\Phi_d(\theta(t),p(t))
		=
		(1-p(t))\ell(\theta(t)) + p(t)\ell(\theta_a).
		\]
		Using the multiplicative-update metric $F_p(p)=1/(p(1-p))$, the delegation
		update is
		\[
		p(t+1)
		=
		p(t)
		-
		\eta\kappa p(t)(1-p(t))
		\partial_p \Phi_d(\theta(t),p(t)).
		\]
		Equivalently,
		\[
		p(t+1)
		=
		p(t)
		+
		\eta\kappa p(t)(1-p(t))
		\bigl(\ell(\theta(t))-\ell(\theta_a)\bigr).
		\]
	\end{itemize}
	Combining the skill and delegation updates, the discrete-time stochastic
	dynamics are
	\begin{align}
		\label{eq:dynamics_general}
		\begin{aligned}
			X_{t,r} &\sim \mathrm{Bern}(p(t)), \qquad r=1,\ldots,m,\\
			\theta(t+1)
			&=
			\theta(t)
			-
			\eta\theta(t)(1-\theta(t))
			\left[
			(1-\bar X_t)\partial_\theta\ell(\theta(t))
			+
			\Delta \bar X_t \partial_\theta g(\theta(t),\theta_d)
			\right],\\
			p(t+1)
			&=
			p(t)
			+
			\eta\kappa p(t)(1-p(t))
			\bigl(\ell(\theta(t))-\ell(\theta_a)\bigr).
		\end{aligned}
	\end{align}
	For ODE \eqref{eq:ODE_simplified}
	\[
	\ell(\theta)=(1-\theta)^2,
	\qquad
	g(\theta,\theta_d)=(\theta-\theta_d)^2,
	\]
	we have
	\[
	\partial_\theta\ell(\theta)=-2(1-\theta),
	\qquad
	\partial_\theta g(\theta,\theta_d)=2(\theta-\theta_d).
	\]
	Substituting these derivatives into \eqref{eq:dynamics_general} yields
	\begin{align}
		\label{eq:dynamics}
		\begin{aligned}
			X_{t,r} &\sim \mathrm{Bern}(p(t)), \qquad r=1,\ldots,m,\\
			\theta(t+1)
			&=
			\theta(t)
			+
			2\eta\theta(t)(1-\theta(t))
			\left[
			(1-\bar X_t)(1-\theta(t))
			+
			\Delta \bar X_t(\theta_d-\theta(t))
			\right],\\
			p(t+1)
			&=
			p(t)
			+
			\eta\kappa p(t)(1-p(t))
			\left((1-\theta(t))^2-(1-\theta_a)^2\right).
		\end{aligned}
	\end{align}
	
	\subsection{Dynamics \eqref{eq:dynamics} tracks to ODE \eqref{eq:ODE}}
	\label{sec:track}
	
	We show that, as the step size $\eta\to 0$, the stochastic trajectory of
	Dynamics~\eqref{eq:dynamics} tracks the solution of ODE~\eqref{eq:ODE}. The
	key point is that each time window contains many task-level decisions, whose
	average delegation rate is $\bar X_t$.
	
	\begin{lemma}[\bf Expected discrete dynamics]
		\label{lem:expected_dynamics}
		Let $\mathcal{F}_t$ denote the filtration generated by the history up to time
		$t$. For Dynamics~\eqref{eq:dynamics}, conditional on the current state
		$(\theta(t),p(t))$, we have
		\begin{align}
			\label{eq:expected_dynamics}
			\begin{aligned}
				\mathbb{E}[\theta(t+1)\mid\mathcal{F}_t]
				&=
				\theta(t)
				+
				2\eta\theta(t)(1-\theta(t))
				\left[
				(1-p(t))(1-\theta(t))
				+
				\Delta p(t)(\theta_d-\theta(t))
				\right],\\
				p(t+1)
				&=
				p(t)
				+
				\eta\kappa p(t)(1-p(t))
				\left((1-\theta(t))^2-(1-\theta_a)^2\right).
			\end{aligned}
		\end{align}
	\end{lemma}
	
	\begin{proof}
		The update for $p(t+1)$ is deterministic conditional on $\mathcal{F}_t$. The
		skill update depends on the realized fraction of delegated tasks
		\[
		\bar X_t=\frac{1}{m}\sum_{r=1}^m X_{t,r}.
		\]
		Since $X_{t,r}\sim \mathrm{Bern}(p(t))$ independently conditional on
		$\mathcal{F}_t$, we have
		\[
		\mathbb{E}[\bar X_t\mid\mathcal{F}_t]=p(t).
		\]
		The update for $\theta(t+1)$ in \eqref{eq:dynamics} is affine in $\bar X_t$.
		Therefore,
		\begin{align*}
			\mathbb{E}[\theta(t+1)\mid\mathcal{F}_t]
			&=
			\theta(t)
			+
			2\eta\theta(t)(1-\theta(t))
			\left[
			(1-\mathbb{E}[\bar X_t\mid\mathcal{F}_t])(1-\theta(t))
			+
			\Delta\mathbb{E}[\bar X_t\mid\mathcal{F}_t](\theta_d-\theta(t))
			\right]\\
			&=
			\theta(t)
			+
			2\eta\theta(t)(1-\theta(t))
			\left[
			(1-p(t))(1-\theta(t))
			+
			\Delta p(t)(\theta_d-\theta(t))
			\right].
		\end{align*}
		This proves \eqref{eq:expected_dynamics}.
	\end{proof}
	
	The drift field associated with \eqref{eq:expected_dynamics} is
	\[
	\left(
	2\theta(1-\theta)\bigl((1-p)(1-\theta)+\Delta p(\theta_d-\theta)\bigr),
	\;
	\kappa p(1-p)\bigl((1-\theta)^2-(1-\theta_a)^2\bigr)
	\right),
	\]
	which is exactly the right-hand side of ODE~\eqref{eq:ODE}. By standard
	stochastic-approximation arguments \cite{Benaim,Heathcote2000}, as
	$\eta\to 0$, the stochastic dynamics~\eqref{eq:dynamics} track the solution of
	ODE~\eqref{eq:ODE} on finite time intervals.
	
	\subsection{Well-posedness and convergence of ODE~\eqref{eq:ODE}}
	\label{sec:ODE}
	
	We establish basic well-posedness properties of the general ODE~\eqref{eq:ODE}.
	Write
	\begin{align}
		\label{eq:general_ODE_rewrite}
		\begin{aligned}
			\dot{\theta}
			&=
			2\theta(1-\theta)
			\Bigl((1-p)(1-\theta)+\Delta p(\theta_d-\theta)\Bigr),\\
			\dot p
			&=
			\kappa p(1-p)
			\Bigl((1-\theta)^2-(1-\theta_a)^2\Bigr).
		\end{aligned}
	\end{align}
	The factor $2$ is immaterial for the phase portrait, but we keep it here to
	match the derivation from ODE \eqref{eq:ODE_simplified}
	$\ell(\theta)=(1-\theta)^2$ and $g(\theta,\theta_d)=(\theta-\theta_d)^2$.
	
	\begin{theorem}[\bf{Existence, invariance, and convergence}]
		\label{thm:existence_general}
		Assume $\kappa,\Delta>0$ and $\theta_a,\theta_d\in[0,1]$. For every initial
		state $(\theta_0,p_0)\in[0,1]^2$, ODE~\eqref{eq:ODE} admits a unique global
		solution $(\theta(t),p(t))$ defined for all $t\ge 0$. Moreover,
		\[
		(\theta(t),p(t))\in[0,1]^2
		\qquad\text{for all }t\ge 0.
		\]
		If the initial condition lies in the interior $(0,1)^2$, then the trajectory
		remains in $(0,1)^2$ for every finite time. In addition, when the equilibria
		are isolated, every trajectory converges to an equilibrium.
	\end{theorem}
	
	\begin{proof}
		Let $x=(\theta,p)$ and let $F(x)$ denote the vector field in
		\eqref{eq:general_ODE_rewrite}. Each component of $F$ is a polynomial in
		$(\theta,p)$, so $F$ is smooth and hence locally Lipschitz on $\mathbb{R}^2$.
		Therefore, by the Picard-Lindel\"of theorem \cite{hirsch2013differential,Perko2013}, for every initial condition
		$x_0\in\mathbb{R}^2$ there exists a unique maximal solution
		$x(t)$ on some interval $[0,T_{\max})$.
		
		We first show positive invariance of $[0,1]^2$. Each coordinate has the
		logistic-barrier form
		\[
		\dot y(t)=y(t)(1-y(t))h(t),
		\]
		where $h(t)$ is continuous along the trajectory. For such an equation, the
		interval $[0,1]$ is positively invariant: if $y(t)$ reaches $0$ or $1$, the
		right-hand side vanishes at that boundary, and uniqueness prevents the
		trajectory from crossing it. Applying this argument to both $\theta(t)$ and
		$p(t)$ gives
		\[
		\theta(t),p(t)\in[0,1]
		\qquad\text{for all }t<T_{\max}.
		\]
		If the initial condition is in $(0,1)^2$, the same uniqueness argument implies
		that neither coordinate can hit $0$ or $1$ at finite time.
		
		Since $[0,1]^2$ is compact and $F$ is continuous, the vector field is bounded
		on $[0,1]^2$. Hence a solution starting in $[0,1]^2$ cannot blow up in finite
		time. Because $F$ is smooth on all of $\mathbb{R}^2$, there is also no
		singularity of the vector field. Thus the maximal existence time satisfies
		$T_{\max}=+\infty$.
		
		It remains to rule out periodic orbits in the interior. On $(0,1)^2$, define
		the Dulac function
		\[
		D(\theta,p)=\frac{1}{\theta(1-\theta)p(1-p)}.
		\]
		Let
		\[
		A(\theta,p)=(1-p)(1-\theta)+\Delta p(\theta_d-\theta),
		\qquad
		B(\theta)=(1-\theta)^2-(1-\theta_a)^2.
		\]
		Then
		\[
		D\dot{\theta}=\frac{A(\theta,p)}{p(1-p)},
		\qquad
		D\dot p=\frac{\kappa B(\theta)}{\theta(1-\theta)}.
		\]
		Therefore,
		\[
		\frac{\partial}{\partial\theta}(D\dot{\theta})
		+
		\frac{\partial}{\partial p}(D\dot p)
		=
		\frac{\partial_\theta A(\theta,p)}{p(1-p)}
		=
		-\frac{(1-p)+\Delta p}{p(1-p)}
		<0
		\]
		for all $(\theta,p)\in(0,1)^2$. 
		By the Bendixson-Dulac criterion \cite{Perko2013}, there are
		no periodic orbits in the interior. Since trajectories are bounded and the
		equilibria are isolated, the Poincar\'e-Bendixson theorem \cite{Perko2013} implies that every
		interior trajectory converges to an equilibrium. 
		Boundary trajectories reduce
		to one-dimensional dynamics and also converge to boundary equilibria.
	\end{proof}
	\noindent
	We next list the equilibria of the general system. This is useful for relating
	the general ODE~\eqref{eq:ODE} to the normalized system
	\eqref{eq:ODE_simplified} analyzed in the main text.
	
	\begin{theorem}[\bf{Equilibria of the general ODE}]
		\label{thm:convergence_general}
		Assume $\theta_d\in[0,1)$ and $\theta_a\in(\theta_d,1)$. The general
		ODE~\eqref{eq:ODE} has the following equilibria in $[0,1]^2$:
		\[
		(0,0),\qquad (0,1),\qquad (1,0),\qquad (1,1),
		\]
		the low-skill full-delegation equilibrium
		\[
		(\theta_d,1),
		\]
		and the interior equilibrium
		\[
		(\theta^\dagger,p^\dagger)
		=
		\left(
		\theta_a,\,
		\frac{1-\theta_a}
		{1-\theta_a+\Delta(\theta_a-\theta_d)}
		\right).
		\]
		When $\theta_d=0$, the equilibrium $(\theta_d,1)$ coincides with $(0,1)$.
		
		Moreover, $(1,0)$ is the high-skill stable equilibrium. If $\theta_d>0$,
		then $(\theta_d,1)$ is the low-skill stable equilibrium and $(0,1)$ is a
		boundary saddle. When $\theta_d=0$, the low-skill stable equilibrium is
		$(0,1)$. The interior equilibrium $(\theta^\dagger,p^\dagger)$ is a saddle.
	\end{theorem}
	
	\begin{proof}
		Equilibria occur when both coordinates have zero drift. Because of the
		multiplicative factors, all four corners
		\[
		(0,0),\quad (0,1),\quad (1,0),\quad (1,1)
		\]
		are equilibria. On the boundary $p=1$, the $\theta$-equation becomes
		\[
		\dot{\theta}
		=
		2\Delta\theta(1-\theta)(\theta_d-\theta),
		\]
		so $(\theta_d,1)$ is also an equilibrium when $\theta_d\in(0,1)$.
		
		For an interior equilibrium, we must have
		\[
		(1-\theta)^2=(1-\theta_a)^2.
		\]
		Since $\theta,\theta_a\in[0,1]$, this gives $\theta=\theta_a$. Substituting
		$\theta=\theta_a$ into the interior $\theta$-nullcline gives
		\[
		(1-p)(1-\theta_a)+\Delta p(\theta_d-\theta_a)=0.
		\]
		Since $\theta_a>\theta_d$, solving for $p$ yields
		\[
		p^\dagger
		=
		\frac{1-\theta_a}
		{1-\theta_a+\Delta(\theta_a-\theta_d)}.
		\]
		This proves the list of equilibria.
		
		The stability claims follow from the Jacobian. At the interior equilibrium,
		write
		\[
		A(\theta,p)=(1-p)(1-\theta)+\Delta p(\theta_d-\theta),
		\qquad
		B(\theta)=(1-\theta)^2-(1-\theta_a)^2.
		\]
		Since $A=B=0$ at $(\theta^\dagger,p^\dagger)$, the Jacobian has entries
		\[
		\partial_\theta\dot{\theta}
		=
		2\theta_a(1-\theta_a)\partial_\theta A<0,
		\qquad
		\partial_p\dot{\theta}
		=
		2\theta_a(1-\theta_a)\partial_p A<0,
		\]
		\[
		\partial_\theta\dot p
		=
		\kappa p^\dagger(1-p^\dagger)\partial_\theta B<0,
		\qquad
		\partial_p\dot p=0.
		\]
		Thus the determinant is
		\[
		-(\partial_p\dot{\theta})(\partial_\theta\dot p)<0,
		\]
		so the interior equilibrium is a saddle. The boundary stability statements
		follow by restricting the vector field to the corresponding boundary faces.
	\end{proof}
	
	\noindent
	In particular, after the normalization $\theta_d=0$ and the constant time
	rescaling used in the main text, the low-skill equilibrium becomes $(0,1)$ and
	the interior saddle becomes
	\[
	\left(
	\theta_a,\,
	\frac{1-\theta_a}
	{1-\theta_a+\Delta\theta_a}
	\right),
	\]
	which is the stationary structure of the simplified system
	\eqref{eq:ODE_simplified}.

	\subsection{Analysis of stochastic variant of ODE \eqref{eq:ODE_simplified}}
	\label{sec:SDE}
	
	\begin{figure*}[htp]
		\centering
		\begin{subfigure}[b]{0.31\textwidth}
			\centering
			\includegraphics[width=\linewidth]{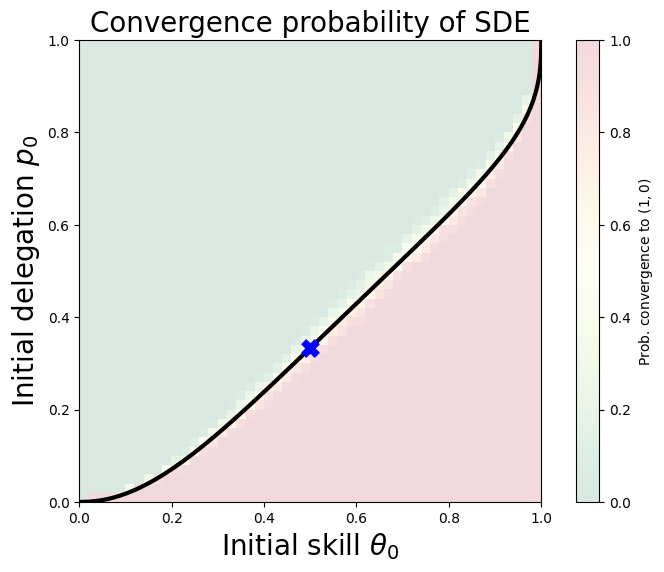}
			\caption{Noisy delegation update \eqref{eq:SDE}}
			\label{fig:SDE_1}
		\end{subfigure}
		\quad \quad
		\begin{subfigure}[b]{0.31\textwidth}
			\centering
			\includegraphics[width=\linewidth]{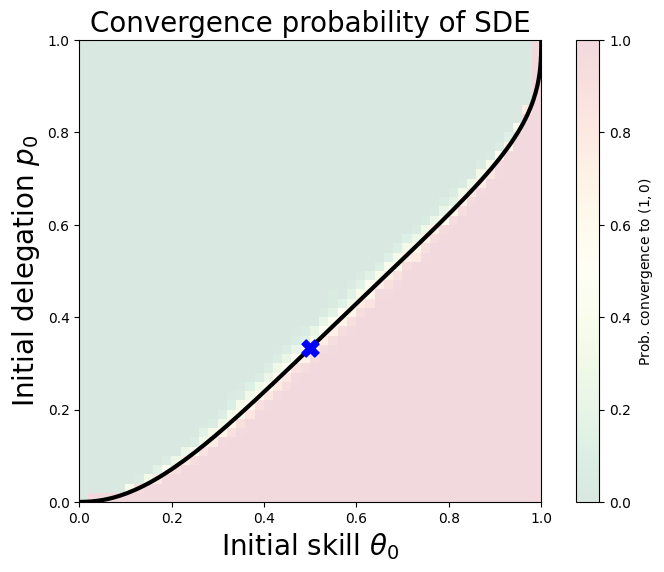}
			\caption{Noisy delegation update \eqref{eq:SDE_2}}
			\label{fig:SDE_2}
		\end{subfigure}
		\caption{Heatmap of the probability of converging to the high-skill equilibrium $(1,0)$ for SDEs \eqref{eq:SDE} and \eqref{eq:SDE_2} as a function of the initial state $(\theta_0,p_0)$, with default settings of $(\theta_a, \kappa, \Delta, \sigma) = (0.5, 3, 2, 0.1)$.
			Green indicates probability $0$, and red indicates probability $1$.
		}
		\label{fig:SDE}
	\end{figure*}
	
	The ODE analysis in Section~\ref{sec:result} describes the averaged dynamics
	of repeated-task adaptation. 
	In practice, however, the empirical feedback used
	to update delegation is noisy: task outcomes vary across instances, the learner
	observes only finitely many AI-assisted and independently completed tasks in
	each time window, and the comparison between the two modes may be aggregated
	imperfectly. 
	To test whether the qualitative conclusions of the deterministic
	ODE are stable under such perturbations, we consider the following Itô
	diffusion perturbation of the delegation update in
	ODE~\eqref{eq:ODE_simplified}:
	\begin{align}
		\label{eq:SDE}
		\begin{aligned}
			d\theta(t)
			&=
			\theta(t)(1-\theta(t))
			\Bigl((1-p(t))(1-\theta(t))-\Delta p(t)\theta(t)\Bigr)\,dt, \\
			dp(t)
			&=
			\kappa p(t)(1-p(t))
			\Bigl(\bigl((1-\theta(t))^2-(1-\theta_a)^2\bigr)\,dt
			+\sigma\,dW_t\Bigr),
		\end{aligned}
	\end{align}
	where $W_t$ is a standard Brownian motion and $\sigma\ge 0$ controls the
	magnitude of persistent noise in the delegation-update signal. When
	$\sigma=0$, the system reduces to ODE~\eqref{eq:ODE_simplified}.
	
	We use this SDE as a robustness model rather than as a full diffusion
	approximation of the finite-sample stochastic process. The noise term captures
	fluctuations in the perceived performance gap that drives the
	delegation-update potential $\Phi_d$, while the multiplicative factor
	$p(t)(1-p(t))$ preserves the state space $p(t)\in[0,1]$. We keep the skill
	equation deterministic to isolate the effect of noisy reliance adaptation; the
	noise is placed in the delegation update because the robustness question is
	whether the basin structure survives imperfect estimation of relative
	performance.
	
	Figure~\ref{fig:SDE_1} reports, for each initial condition
	$(\theta_0,p_0)\in(0,1)^2$, the empirical probability that a simulated
	trajectory reaches the high-skill regime. Specifically, we classify a
	trajectory as high-skill if, at a large terminal time $T$, it satisfies
	$\theta(T)\ge 1-\varepsilon$ and $p(T)\le \varepsilon$. The resulting heatmap
	closely follows the basin boundary of the deterministic ODE shown in
	Figure~\ref{fig:manifold}, with a sharp transition in the probability of
	reaching the high-skill regime. This provides numerical evidence that the ODE
	basin structure is stable under moderate delegation noise.
	
	We also consider a state-dependent noise variant:
	\begin{align}
		\label{eq:SDE_2}
		dp(t)
		=
		\kappa p(t)(1-p(t))
		\Bigl(
		((1-\theta(t))^2-(1-\theta_a)^2)\,dt
		+
		2\sigma(1-\theta(t))\,dW_t
		\Bigr).
	\end{align}
	This captures the possibility that lower-skill learners aggregate or interpret
	performance feedback less reliably. Even when both AI-assisted and independent
	tasks are observed within a time window, evaluating the quality of AI outputs,
	diagnosing mistakes, and comparing them with one's own work may be harder for
	less skilled users. The factor $1-\theta(t)$ therefore makes the delegation
	noise larger when learner skill is lower. Figure~\ref{fig:SDE_2} shows that
	the system still exhibits two dominant long-run regimes and a sharp transition
	between them. Thus, the bistable structure observed in the ODE is not merely an
	artifact of deterministic delegation updating, but is observed to persist under noisy and
	state-dependent reliance adjustment.

	\section{Proofs from Section \ref{sec:result}}
	\label{sec:result_details}
	
	In this section, we provide proofs for the results in Section \ref{sec:result} and show how to derive an approximation of the stable manifold $\psi$.
	
	\subsection{Proof of Theorem \ref{thm:convergence}: equilibrium of ODE \eqref{eq:ODE_simplified}}
	\label{sec:proof_convergence}
	
	\begin{theorem}[\bf{Restatement of Theorem \ref{thm:convergence}}]
		Let $\theta_a, \kappa, \Delta$ be the parameters of ODE~\eqref{eq:ODE_simplified}.
		When \(\theta_a\in(0,1)\), the system has two attracting corner equilibria,
		\((1,0)\) and \((0,1)\); two repelling corner equilibria, \((0,0)\) and
		\((1,1)\); and one interior saddle
		$
		(\theta^\dagger,p^\dagger)
		=
		\left(\theta_a,\frac{1-\theta_a}{1-(1-\Delta)\theta_a}\right)$.
	\end{theorem}
	
	\begin{proof}
		To establish the equilibria and their stability properties, we first identify the fixed points of the system and then analyze the Jacobian matrix of the linearized system around these points.
		
		\paragraph{Existence of equilibria.}
		The equilibria are the solutions to the system of algebraic equations $\dot{\theta} = 0$ and $\dot{p} = 0$:
		\begin{align}
			\theta(1-\theta)\left((1-p)(1-\theta)-\Delta p\,\theta\right) &= 0, \label{eq:null_theta} \\
			\kappa\,p(1-p)\left((1-\theta)^2-(1-\theta_a)^2\right) &= 0. \label{eq:null_p}
		\end{align}
		From the factors $\theta(1-\theta)$ in \eqref{eq:null_theta} and $p(1-p)$ in \eqref{eq:null_p}, we immediately identify the four corner equilibria in the domain $\{0,1\}^2$:
		\[
		(\theta, p) \in \{(0,0), (0,1), (1,0), (1,1)\}.
		\]
		To find the interior equilibrium $(\theta^\dagger, p^\dagger) \in (0,1)^2$, we require the non-trivial factors to vanish:
		\begin{align}
			(1-\theta)^2 - (1-\theta_a)^2 &= 0, \label{eq:interior_p} \\
			(1-p)(1-\theta) - \Delta p\,\theta &= 0. \label{eq:interior_theta}
		\end{align}
		Since $\theta, \theta_a \in [0,1]$, Equation \eqref{eq:interior_p} implies $1-\theta = 1-\theta_a$, yielding the unique skill solution $\theta^\dagger = \theta_a$. Substituting $\theta^\dagger = \theta_a$ into \eqref{eq:interior_theta} and solving for $p$ yields:
		\[
		(1-p)(1-\theta_a) = \Delta p \theta_a \implies 1-\theta_a = p\left(1-\theta_a + \Delta \theta_a\right).
		\]
		Thus, the unique interior equilibrium is
		\[
		(\theta^\dagger, p^\dagger) = \left(\theta_a, \frac{1-\theta_a}{1 - (1-\Delta)\theta_a}\right).
		\]
		
		\paragraph{Stability analysis.}
		We analyze the linear stability via the Jacobian matrix $J$ of the system. Let $f(\theta, p) = \dot{\theta}$ and $g(\theta, p) = \dot{p}$. The Jacobian is given by
		\[
		J(\theta, p) = 
		\begin{bmatrix}
			\partial_\theta f & \partial_p f \\
			\partial_\theta g & \partial_p g
		\end{bmatrix}.
		\]
		
		\noindent \emph{Corner equilibria.}
		At the corner points where $(\theta, p) \in \{0,1\}^2$, the off-diagonal terms vanish (i.e., $\partial_p f = 0$ and $\partial_\theta g = 0$). The Jacobian becomes diagonal:
		\[
		J = \mathrm{diag}(\lambda_\theta, \lambda_p),
		\]
		where the eigenvalues are given by the partial derivatives evaluated at the equilibrium $(\theta^\infty, p^\infty)$:
		\begin{align*}
			\lambda_\theta &= (1 - 2\theta^\infty)\left((1 - p^\infty)(1 - \theta^\infty) - \Delta p^\infty \theta^\infty\right), \\
			\lambda_p &= \kappa (1 - 2p^\infty)\left((1 - \theta^\infty)^2 - (1 - \theta_a)^2\right).
		\end{align*}
		We examine two cases:
		\begin{enumerate}
			\item \textbf{Sources ($\theta^\infty = p^\infty$):}
			\begin{itemize}
				\item At $(0,0)$: $\lambda_\theta = 1 > 0$ and $\lambda_p = \kappa(1 - (1-\theta_a)^2) > 0$. Both eigenvalues are positive.
				\item At $(1,1)$: $\lambda_\theta = \Delta > 0$ and $\lambda_p = \kappa(1-\theta_a)^2 > 0$. Both eigenvalues are positive.
			\end{itemize}
			Thus, $(0,0)$ and $(1,1)$ are unstable nodes.
			
			\item \textbf{Sinks ($\theta^\infty \neq p^\infty$):}
			\begin{itemize}
				\item At $(1,0)$: $\lambda_\theta = 0$ and $\lambda_p = -\kappa(1 - (1-\theta_a)^2) < 0$.
				The negative $\lambda_p$ implies exponential decay in the delegation direction (transverse stability). Along the invariant boundary $p=0$, the skill dynamics reduce to $\dot{\theta} = \theta(1-\theta)^2$, which is positive for $\theta \in (0,1)$. This implies that $\theta$ monotonically increases toward 1. Thus, $(1,0)$ is asymptotically stable.
				
				\item At $(0,1)$: $\lambda_\theta = 0$ and $\lambda_p = -\kappa(1-\theta_a)^2 < 0$.
				The negative $\lambda_p$ implies exponential decay toward full delegation. Along the invariant boundary $p=1$, the skill dynamics reduce to $\dot{\theta} = -\Delta \theta^2(1-\theta)$, which is negative for $\theta \in (0,1)$. This implies that $\theta$ monotonically decays toward 0. Thus, $(0,1)$ is asymptotically stable.
			\end{itemize}
			Thus, the high-skill equilibrium $(1,0)$ and low-skill equilibrium $(0,1)$ are stable nodes.
		\end{enumerate}
		
		\noindent \emph{Interior saddle point.}
		At the interior equilibrium, the nullcline conditions imply
		\[
		\partial_p \dot p(\theta^\dagger,p^\dagger)=0.
		\]
		Moreover,
		\[
		\partial_p \dot{\theta}(\theta^\dagger,p^\dagger)
		=
		-\theta_a(1-\theta_a)\bigl((1-\theta_a)+\Delta\theta_a\bigr)
		=
		-\frac{\theta_a(1-\theta_a)^2}{p^\dagger}<0,
		\]
		and
		\[
		\partial_\theta \dot p(\theta^\dagger,p^\dagger)
		=
		-2\kappa p^\dagger(1-p^\dagger)(1-\theta_a)<0.
		\]
		Also,
		\[
		\partial_\theta \dot{\theta}(\theta^\dagger,p^\dagger)<0.
		\]
		Since \(\partial_p\dot p(\theta^\dagger,p^\dagger)=0\), the determinant is
		\[
		\det J(\theta^\dagger,p^\dagger)
		=
		-\partial_p\dot\theta(\theta^\dagger,p^\dagger)
		\partial_\theta\dot p(\theta^\dagger,p^\dagger)
		<0.
		\]
		Thus the interior equilibrium is a hyperbolic saddle.
		
	\end{proof}
	
	\subsection{Proof of Theorem \ref{thm:saddle}: two basins divided by saddle point}
	\label{sec:proof_saddle}
	
	\begin{theorem}[\bf{Restatement of Theorem \ref{thm:saddle}}]
		Let $\theta_a, \kappa, \Delta$ be the parameters of ODE \eqref{eq:ODE_simplified}.
		Let $(\theta^\dagger,p^\dagger)$ be the interior saddle point guaranteed by Theorem \ref{thm:convergence}.
		There exists a nondecreasing differentiable function $\psi:(0,1)\to(0,1)$
		that extends continuously to $[0,1]$ with $\psi(0)=0$ and $\psi(1)=1$,
		whose graph coincides with the one-dimensional stable manifold of $(\theta^\dagger,p^\dagger)$
		within $(0,1)^2$.
		The curve $p = \psi(\theta)$ partitions the state space $[0,1]^2$ into two distinct basins of attraction:
		\begin{itemize}[leftmargin=*, itemsep=1pt, topsep=2pt, parsep=0pt]
			\item For any initial state with $p_0 > \psi(\theta_0)$ (higher delegation), the learner converges to the low-skill equilibrium $(0,1)$.
			\item For any initial state with $p_0 < \psi(\theta_0)$ (lower delegation), the learner converges to the high-skill equilibrium $(1,0)$.
		\end{itemize}
	\end{theorem}
	
	\noindent
	For preparation, we give the following lemma that characterizes a monotone coupling between skill gap and delegation.
	
	\begin{lemma}[\bf{Monotone coupling between skill gap and delegation}]
		\label{lm:coupling}
		Let $x=1-\theta$ denote the skill gap. For fixed $\kappa,\Delta>0$, the
		system in variables $(x,p)$ is cooperative:
		\[
		\dot x
		= x(1-x)\bigl[\Delta p(1-x)-(1-p)x\bigr],
		\qquad
		\dot p
		= \kappa p(1-p)\bigl[x^2-(1-\theta_a)^2\bigr].
		\]
		Consequently, if $x_0\leq \widetilde x_0$ and
		$p_0\leq \widetilde p_0$, then for all $t\geq 0$,
		\[
		x(t;x_0,p_0)\leq x(t;\widetilde x_0,\widetilde p_0),
		\qquad
		p(t;x_0,p_0)\leq p(t;\widetilde x_0,\widetilde p_0).
		\]
		Moreover, for fixed initial condition, increasing $\theta_a$ weakly
		increases $p(t)$ and $x(t)$ for all $t\geq 0$.
	\end{lemma}
	
	\noindent
	The proof can be found in Section \ref{sec:proof_coupling}.
	Now we are ready to prove Theorem \ref{thm:saddle}.
	
	\begin{proof}
		The proof proceeds in four steps: first, obtain the local stable manifold
		$W^s$ from the stable manifold theorem; second, continue its two
		branches globally and show they reach the boundary; third, identify the global stable manifold $W^s$ as the separatrix; fourth, use the monotone comparison lemma to represent $W^s$ as a nondecreasing graph $p=\psi(\theta)$.
		
		\paragraph{Local stable manifold exists and is differentiable.}
		From Theorem~\ref{thm:convergence}, the interior equilibrium $(\theta^\dagger, p^\dagger)$ is a hyperbolic saddle point. 
		By the Stable Manifold Theorem for planar system \cite{Perko2013}, there exists a unique one-dimensional stable manifold $W^s(\theta^\dagger, p^\dagger)$ passing through $(\theta^\dagger, p^\dagger)$ tangent to the stable eigenvector.
		Since $W^s$ is a one-dimensional stable manifold of a smooth planar vector field, it is a differentiable embedded curve.

		\paragraph{Each backward branch reaches the boundary.}
		Let $\Omega=(0,1)^2$, and write the vector field of~\eqref{eq:ODE_simplified} as
		\[
		F(\theta,p)=(f(\theta,p),h(\theta,p)).
		\]
		On $\Omega$, take the Dulac function
		\[
		D(\theta,p)=\frac{1}{\theta(1-\theta)p(1-p)}.
		\]
		A direct computation gives
		\[
		\partial_\theta(Df)+\partial_p(Dh)
		=
		-\frac{(1-p)+\Delta p}{p(1-p)}<0
		\qquad\text{on }\Omega .
		\]
		Since $\Omega$ is simply connected, the Bendixson-Dulac criterion rules out
		closed invariant curves in $\Omega$; in particular, there are no periodic
		or homoclinic orbits contained in the interior.
		
		We now show that each branch of the stable manifold reaches the boundary.
		Let $\Gamma$ be one branch of the local stable manifold
		$W^s(\theta^\dagger,p^\dagger)\setminus\{(\theta^\dagger,p^\dagger)\}$,
		and define its backward continuation by
		\[
		\Gamma^-:=\{\phi_{-t}(z):z\in\Gamma,\ t\geq 0\},
		\]
		where $\phi_t$ denotes the flow. Suppose, for contradiction, that
		\[
		\overline{\Gamma^-}\subset \Omega .
		\]
		Then the alpha-limit set
		\[
		\alpha(\Gamma)
		:=
		\bigcap_{T\geq 0}
		\overline{\{\phi_{-t}(z):z\in\Gamma,\ t\geq T\}}
		\]
		is a nonempty compact invariant subset of $\Omega$. 
		By the Poincare-Bendixson theorem \cite{Perko2013} and the absence of periodic or homoclinic orbits in $\Omega$, the only possible interior limit set is an equilibrium.
		By Theorem~\ref{thm:convergence}, the only interior equilibrium is
		$(\theta^\dagger,p^\dagger)$. 
		Hence,
		\[
		\alpha(\Gamma)=\{(\theta^\dagger,p^\dagger)\}.
		\]
		But every point on $\Gamma$ also converges to
		$(\theta^\dagger,p^\dagger)$ in forward time, because
		$\Gamma\subset W^s(\theta^\dagger,p^\dagger)$. 
		Therefore, $\Gamma^-$ is a
		homoclinic branch to the saddle, contradicting the Bendixson-Dulac
		exclusion above. 
		Thus,
		\[
		\overline{\Gamma^-}\cap\partial[0,1]^2\neq\emptyset .
		\]
		Applying the same argument to the other local branch shows that the global stable manifold has two branches accumulating on the boundary of the square.
		
		\paragraph{The two backward branches limit to $(0,0)$ and $(1,1)$.}
		The two branches of
		$W^s(\theta^\dagger,p^\dagger)$ have alpha-limit sets as time goes to
		$-\infty$. 
		By compactness and invariance, each alpha-limit set is a compact
		invariant subset of the boundary or an equilibrium in the interior. 
		The Dulac argument rules out interior periodic or homoclinic limit sets, and the only interior equilibrium is the saddle itself. 
		Therefore, each backward branch must limit to a boundary equilibrium.
		
		The attracting corner equilibria $(1,0)$ and $(0,1)$ cannot be alpha-limits of nontrivial forward trajectories on the stable manifold, because they have no forward unstable direction into the interior. 
		Hence, the two alpha-limits must be the repelling corners $(0,0)$ and $(1,1)$. 
		The local stable eigenvector at the saddle has positive slope, so one branch extends toward the lower-left corner $(0,0)$ and the other toward the upper-right corner $(1,1)$.
		Therefore, the graph connects $(0,0)$ to $(1,1)$, which gives
		\[
		\lim_{\theta\downarrow 0}\psi(\theta)=0,
		\qquad
		\lim_{\theta\uparrow 1}\psi(\theta)=1.
		\]

		\paragraph{Partitioning the basins of attraction.}
		Let $W^s$ denote the global stable manifold of the interior saddle
		$(\theta^\dagger,p^\dagger)$. 
		Since $W^s$ is invariant and connects the
		boundary of $[0,1]^2$, it partitions the interior of the state space into two
		connected components, denoted by $\Omega_{\mathrm{high}}$ and
		$\Omega_{\mathrm{low}}$.
		
		For any initial condition on $W^s$, the trajectory converges to the saddle.
		For an initial condition off $W^s$, uniqueness of ODE solutions implies that
		the trajectory cannot cross $W^s$ (Theorem \ref{thm:existence_general}). 
		Since every trajectory converges to an
		equilibrium and the only attracting equilibria in the interior dynamics are
		$(1,0)$ and $(0,1)$, each connected component must belong to one of these two
		basins. 
		The component adjacent to $(1,0)$ is the high-skill basin, and the
		component adjacent to $(0,1)$ is the low-skill basin. 

		Let $B_{\rm low}$ and $B_{\rm high}$ be the basins of $(0,1)$ and $(1,0)$,
		respectively. We first identify the basin boundary in the interior. We claim that
		\[
		\partial B_{\rm low}\cap(0,1)^2
		=
		\partial B_{\rm high}\cap(0,1)^2
		=
		W^s(\theta^\dagger,p^\dagger)\cap(0,1)^2 .
		\]
		
		Indeed, basin boundaries are invariant under the flow. Let
		$z\in \partial B_{\rm low}\cap(0,1)^2$. By Theorem~\ref{thm:existence_general},
		the trajectory from $z$ converges to an equilibrium. It cannot converge to $(0,1)$ or $(1,0)$: if it did, then since these equilibria are attracting,
		$z$ would lie in the interior of the corresponding basin, contradicting
		$z\in\partial B_{\rm low}$. 
		It also cannot converge to the repelling corners
		$(0,0)$ or $(1,1)$ from an interior initial condition. 
		Hence, the trajectory must converge to the unique interior equilibrium
		$(\theta^\dagger,p^\dagger)$, and therefore
		\[
		z\in W^s(\theta^\dagger,p^\dagger).
		\]
		This proves
		\[
		\partial B_{\rm low}\cap(0,1)^2
		\subseteq
		W^s(\theta^\dagger,p^\dagger)\cap(0,1)^2 .
		\]
		The same argument gives the corresponding inclusion for
		$\partial B_{\rm high}$.
		
		Conversely, let $z\in W^s(\theta^\dagger,p^\dagger)\cap(0,1)^2$.
		Then the trajectory from $z$ converges to the saddle, so $z$ is not in
		either attracting basin. 
		It remains to show that $z$ lies on their common
		boundary. 
		Near the saddle, the local stable manifold separates the local
		phase portrait into two sides: points on one side leave a neighborhood of the saddle toward the high-skill basin, while points on the other side leave toward the low-skill basin. 
		Since the flow is a diffeomorphism on every
		finite time interval and the basins are invariant, this two-sided separation property propagates along the entire global stable manifold.
		Hence, every neighborhood of $z$ contains points from both $B_{\rm high}$ and $B_{\rm low}$, so
		\[
		z\in \partial B_{\rm high}\cap \partial B_{\rm low}.
		\]
		Therefore,
		\[
		\partial B_{\rm low}\cap(0,1)^2
		=
		\partial B_{\rm high}\cap(0,1)^2
		=
		W^s(\theta^\dagger,p^\dagger)\cap(0,1)^2 .
		\]
		Hence, $W^s$ is the separatrix between the two basins.
		
		\paragraph{Graph representation and monotonicity of the stable manifold.}
		Work in the order $(\theta,p)\preceq(\theta',p')$ if
		$\theta\geq \theta'$ and $p\leq p'$, equivalently if the first state has
		weakly smaller skill gap and weakly smaller delegation. 
		By Lemma \ref{lm:coupling},
		$B_{\rm low}$ is upward closed in skill gap and delegation: if
		$(\theta,p)\in B_{\rm low}$, $\theta'\leq \theta$, and $p'\geq p$, then
		$(\theta',p')\in B_{\rm low}$.
		
		For each fixed $\theta\in(0,1)$, define
		\[
		I_\theta=\{p\in(0,1):(\theta,p)\in B_{\rm low}\}.
		\]
		The upward-closed property implies that $I_\theta$ is an upper interval.
		Moreover, $I_\theta$ is nonempty and not all of $(0,1)$, because the
		vertical line $\{\theta\}\times(0,1)$ must cross the basin boundary
		separating $B_{\rm low}$ from $B_{\rm high}$. Define
		\[
		\psi(\theta):=\inf I_\theta .
		\]
		For every fixed $\theta\in(0,1)$, the slice is nonempty because the
		boundary point $(\theta,1)$ evolves along $p=1$ toward $(0,1)$, and by
		stability of $(0,1)$, all sufficiently nearby interior points also lie in $B_{\rm low}$. 
		Similarly, the slice is not all of $(0,1)$ because
		$(\theta,0)$ evolves along $p=0$ toward $(1,0)$, and sufficiently nearby
		interior points lie in $B_{\rm high}$.
		
		Then
		\[
		(\theta,p)\in B_{\rm low}\quad\text{if }p>\psi(\theta),
		\qquad
		(\theta,p)\in B_{\rm high}\quad\text{if }p<\psi(\theta),
		\]
		and
		\[
		(\theta,\psi(\theta))\in
		\partial B_{\rm low}\cap (0,1)^2
		=
		W^s(\theta^\dagger,p^\dagger)\cap (0,1)^2 .
		\]
		Thus, the stable manifold is the graph of $\psi$.
		
		If $\theta_1<\theta_2$, then a learner at skill $\theta_1$ has weakly
		larger skill gap than a learner at skill $\theta_2$. 
		Since $B_{\rm low}$ is
		upward closed in skill gap and delegation, the threshold needed to enter
		$B_{\rm low}$ cannot be lower at the smaller skill gap. Hence
		\[
		\psi(\theta_1)\leq \psi(\theta_2),
		\]
		so $\psi$ is nondecreasing.
		
		\paragraph{Smoothness of the graph.}
		The stable manifold theorem gives that
		$W^s(\theta^\dagger,p^\dagger)$ is a differentiable embedded curve in the interior.
		Since we have shown that this curve is the graph of a monotone function
		$\psi$, it remains only to rule out vertical tangencies.
		
		The $\theta$-nullcline in the interior is
		\[
		p=h(\theta):=
		\frac{1-\theta}{1-(1-\Delta)\theta}.
		\]
		It is strictly decreasing because
		\[
		h'(\theta)=
		-\frac{\Delta}{\bigl(1-(1-\Delta)\theta\bigr)^2}<0.
		\]
		Since $\psi$ is nondecreasing, the graph $p=\psi(\theta)$ can intersect the strictly decreasing nullcline $p=h(\theta)$ at most once. 
		It intersects it
		at the saddle $(\theta^\dagger,p^\dagger)$. 
		Therefore, $\dot\theta\neq 0$ on the stable manifold away from the saddle, and the embedded curve has no vertical tangent away from the saddle.
		
		At the saddle, the stable eigenvector also has nonzero $\theta$ component.
		Indeed, for the Jacobian
		\[
		J=
		\begin{pmatrix}
			a & b\\
			c & 0
		\end{pmatrix},
		\qquad b<0,\quad c<0,
		\]
		the stable eigenvalue $\lambda_s<0$ satisfies
		$c v_\theta=\lambda_s v_p$.
		Hence, if $v_\theta=0$, then
		$v_p=0$, impossible for an eigenvector. 
		Thus, the local stable manifold is
		also a differentiable graph near the saddle. 
		Consequently, $\psi$ is differentiable on
		$(0,1)$.

		Overall, we complete the proof.
	\end{proof}
	
	\subsection{Proof of Theorem \ref{thm:monotonicity}: effects of $\theta_a$ on $\psi$}
	\label{sec:proof_AI_skill}
	
	\begin{theorem}[\textbf{Restatement of Theorem \ref{thm:monotonicity}}]
		For every $\theta\in(0,1)$, the boundary value
		$\psi_{\theta_a}(\theta)$ is nonincreasing and continuously differentiable in $\theta_a$. 
		Moreover, the family $\{\psi_{\theta_a}\}_{\theta_a\in(0,1)}$ sweeps the state space in the following sense: for every $(\theta_0,p_0)\in(0,1)^2$, there exists
		$\theta_a\in(0,1)$ such that
		$
		p_0=\psi_{\theta_a}(\theta_0).
		$
	\end{theorem}
	
	\begin{proof}
		We prove the properties of the separatrix $\psi_{\theta_a}$ sequentially.
		
		\paragraph{Monotonicity with respect to $\theta_a$.}
		We show that for any fixed $\theta \in (0,1)$, $\psi_{\theta_a}(\theta)$ is decreasing in $\theta_a$.
		Let $\theta_{a,1} < \theta_{a,2}$ be two distinct AI skill levels. 
		Let $\psi_1$ and $\psi_2$ denote the corresponding separatrices.
		Suppose, for the sake of contradiction, that there exists some $\theta^* \in (0,1)$ such that $\psi_2(\theta^*) > \psi_1(\theta^*)$.
		Since the functions are continuous and connect $(0,0)$ to $(1,1)$ (Theorem~\ref{thm:saddle}), we can select a test point $(\theta^*, p_0)$ in the state space such that:
		\[
		\psi_1(\theta^*) < p_0 < \psi_2(\theta^*).
		\]
		Now, consider the long-run convergence of a learner starting at $(\theta^*, p_0)$ under the two different AI regimes:
		\begin{itemize}
			\item \textbf{Regime 1 ($\theta_{a,1}$):} Since $p_0 > \psi_1(\theta^*)$, the initial state lies in the \emph{Low-Skill Basin} (above the separatrix). By Theorem~\ref{thm:saddle}, the learner converges to the low-skill equilibrium: $\theta^\infty_1 = 0$.
			\item \textbf{Regime 2 ($\theta_{a,2}$):} Since $p_0 < \psi_2(\theta^*)$, the initial state lies in the \emph{High-Skill Basin} (below the separatrix). By Theorem~\ref{thm:saddle}, the learner converges to the high-skill equilibrium: $\theta^\infty_2 = 1$.
		\end{itemize}
		Comparing the outcomes, we have $\theta^\infty_2 > \theta^\infty_1$.
		However, Lemma~\ref{lm:coupling} states that for any time $t$, the skill $\theta(t)$ is non-increasing in $\theta_a$. 
		Taking the limit $t \to \infty$, this implies that increasing the AI skill from $\theta_{a,1}$ to $\theta_{a,2}$ must result in $\theta^\infty_2 \le \theta^\infty_1$.
		The result $\theta^\infty_2 = 1 > 0 = \theta^\infty_1$ is a direct contradiction.
		Therefore, the assumption must be false, and it must hold that $\psi_2(\theta) \leq \psi_1(\theta)$ for all $\theta \in (0,1)$.
		
		\paragraph{Differentiability.}
		The vector field defined by ODE~\eqref{eq:ODE_simplified} is $C^\infty$-smooth with respect to the parameter $\theta_a$.
		Since the saddle point $(\theta^\dagger, p^\dagger)$ is hyperbolic for all $\theta_a \in (0,1)$, the Stable Manifold Theorem with parameters \cite{Perko2013} guarantees that the local stable manifold varies smoothly (is differentiable) with respect to $\theta_a$.
		Since the global separatrix $\psi_{\theta_a}$ is obtained by backward integration of the local manifold—and the flow of the ODE is smooth with respect to parameters—the global function $\psi_{\theta_a}(\theta)$ is continuously differentiable with respect to $\theta_a$.
		
		\paragraph{Sweeping Property.}
		We show that for any point $(\theta_0, p_0)$, there exists a $\theta_a$ such that the separatrix passes through it.
		We analyze the geometric limits of the saddle point $(\theta^\dagger, p^\dagger)$ as $\theta_a$ approaches the boundaries:
		\begin{itemize}
			\item \textbf{Limit $\theta_a \to 1$:} The saddle point $(\theta^\dagger, p^\dagger)$ converges to the high-skill corner $(1,0)$. For every fixed $\theta\in(0,1)$,
			\[
			\lim_{\theta_a\uparrow 1}\psi_{\theta_a}(\theta)=0.
			\]
			Equivalently, the basin boundary moves downward toward the $\theta$-axis
			as AI skill approaches one.
			Thus, $\lim_{\theta_a \to 1} \psi_{\theta_a}(\theta_0) = 0$.
			\item \textbf{Limit $\theta_a \to 0$:} The saddle point converges to the low-skill corner $(0,1)$ (since $p^\dagger \to 1$). 
			The stable manifold connects the source $(0,0)$ to the saddle $(0,1)$ and the threshold converges pointwise to $1$ as $\theta_a\downarrow 0$, pushing the boundary toward $p=1$. 
			Thus, $\lim_{\theta_a \to 0} \psi_{\theta_a}(\theta_0) = 1$.
		\end{itemize}
		Since $\psi_{\theta_a}(\theta_0)$ is a continuous function of $\theta_a$ and its range includes the interval $(0,1)$ as $\theta_a$ varies, the Intermediate Value Theorem implies that for any $p_0 \in (0,1)$, there exists a value $\theta_a \in (0,1)$ such that $\psi_{\theta_a}(\theta_0) = p_0$.
		
		\paragraph{Continuity in $\theta_a$.}
		Fix $\bar\theta_a\in(0,1)$ and $\bar\theta\in(0,1)$. 
		The vector field is $C^\infty$ in $(\theta,p,\theta_a)$, and the interior equilibrium
		\[
		(\theta^\dagger(\theta_a),p^\dagger(\theta_a))
		=
		\left(
		\theta_a,
		\frac{1-\theta_a}{1-(1-\Delta)\theta_a}
		\right)
		\]
		is hyperbolic for every $\theta_a\in(0,1)$. 
		By the parameter-dependent stable manifold theorem \cite{Perko2013}, the local stable manifold varies continuously, in
		fact differentiable, with respect to $\theta_a$. 
		Since the global branch is obtained by flowing the local manifold backward, and the flow depends continuously on initial condition and parameter on every finite time interval, the intersection height of the global stable manifold with the vertical section $\{\bar\theta\}\times(0,1)$ varies continuously in $\theta_a$. 
		Therefore, $\psi_{\theta_a}(\bar\theta)$ is continuous in $\theta_a$.
		
		To prove that for any $(\theta_0, p_0)\in (0,1)^2$, there exists $\theta_a\in (0,1)$ such that $p_0 = \psi_{\theta_a}(\theta_0)$, it suffices to prove the following lemma.
		
		\begin{lemma}[\bf{Endpoint limits of the separatrix}]
			For every fixed $\theta_0\in(0,1)$,
			\[
			\lim_{\theta_a\uparrow 1}\psi_{\theta_a}(\theta_0)=0,
			\qquad
			\lim_{\theta_a\downarrow 0}\psi_{\theta_a}(\theta_0)=1.
			\]
		\end{lemma}
		
		\begin{proof}
			We prove the two limits by comparison with the limiting systems
			$\theta_a=1$ and $\theta_a=0$.
			
			First consider $\theta_a=1$. The delegation equation becomes
			\[
			\dot p=\kappa p(1-p)(1-\theta)^2\geq 0.
			\]
			For every interior initial condition with $p_0>0$, delegation is
			nondecreasing, and the only possible attracting regime in the interior is
			the low-skill full-delegation regime. More precisely, every trajectory with
			$(\theta_0,p_0)\in(0,1)^2$ and $p_0>0$ converges to $(0,1)$ under the
			limiting system $\theta_a=1$.
			
			Fix $\theta_0\in(0,1)$ and any $\varepsilon>0$. Under the limiting system
			$\theta_a=1$, the trajectory starting from $(\theta_0,\varepsilon)$
			converges to $(0,1)$. Hence there exists a finite time $T$ at which this
			trajectory enters a compact subset of the local basin of $(0,1)$. By
			continuous dependence of ODE solutions on parameters, for all
			$\theta_a<1$ sufficiently close to $1$, the trajectory starting from
			$(\theta_0,\varepsilon)$ also enters the same local basin. Therefore
			$(\theta_0,\varepsilon)$ lies in the low-skill basin for all sufficiently
			large $\theta_a$. Since $\varepsilon>0$ was arbitrary, the basin threshold
			satisfies
			\[
			\limsup_{\theta_a\uparrow 1}\psi_{\theta_a}(\theta_0)\leq \varepsilon.
			\]
			Letting $\varepsilon\downarrow 0$ gives
			\[
			\lim_{\theta_a\uparrow 1}\psi_{\theta_a}(\theta_0)=0.
			\]
			Next consider $\theta_a=0$. The delegation equation becomes
			\[
			\dot p=\kappa p(1-p)\bigl((1-\theta)^2-1\bigr)
			=
			-\kappa p(1-p)\theta(2-\theta)\leq 0.
			\]
			For every interior initial condition with $p_0<1$, delegation is
			nonincreasing, and the limiting system converges to the high-skill
			low-delegation regime $(1,0)$. Fix $\theta_0\in(0,1)$ and any
			$\varepsilon>0$. Under the limiting system $\theta_a=0$, the trajectory
			starting from $(\theta_0,1-\varepsilon)$ converges to $(1,0)$. Hence it
			enters a compact subset of the local basin of $(1,0)$ in finite time.
			Again by continuous dependence on parameters, for all $\theta_a>0$
			sufficiently close to $0$, the trajectory starting from
			$(\theta_0,1-\varepsilon)$ remains in the high-skill basin. Therefore
			\[
			\psi_{\theta_a}(\theta_0)\geq 1-\varepsilon
			\]
			for all sufficiently small $\theta_a$. Letting $\varepsilon\downarrow 0$
			gives
			\[
			\lim_{\theta_a\downarrow 0}\psi_{\theta_a}(\theta_0)=1.
			\]
		\end{proof}
		\noindent
		Fix $(\theta_0,p_0)\in(0,1)^2$. By the endpoint lemma,
		\[
		\lim_{\theta_a\downarrow 0}\psi_{\theta_a}(\theta_0)=1>p_0,
		\qquad
		\lim_{\theta_a\uparrow 1}\psi_{\theta_a}(\theta_0)=0<p_0.
		\]
		Since $\psi_{\theta_a}(\theta_0)$ is continuous in $\theta_a$, the
		intermediate value theorem implies that there exists
		$\theta_a\in(0,1)$ such that
		\[
		\psi_{\theta_a}(\theta_0)=p_0.
		\]
		This proves the sweeping property.
	\end{proof}

	\subsection{Proof of Theorem \ref{thm:performance}: short-term gain, long-term losses}
	\label{sec:proof_performance}
	
	{
		\begin{theorem}[\bf{Restatement of Theorem \ref{thm:performance}}]
			Let \(\theta_a,\kappa,\Delta\) be the parameters of
			ODE~\eqref{eq:ODE_simplified}. Fix an initial state
			\((\theta_0,p_0)\in(0,1)^2\) with \(\theta_0<\theta_a\). Let
			\[
			t_c := \inf\left\{t\ge 0:\ \forall t'>t,\ G_\ell(t')>0\right\}
			\]
			denote the final crossing time, after which adaptive delegation incurs higher
			loss than the no-AI baseline. Then \(t_c<\infty\) and
			\(\theta(t_c;\theta_0,0)\le \theta_a\).
	\end{theorem}}
	
	\begin{proof}
		We first express the performance gap $G_\ell(t)$ by expanding the loss term $\ell(\theta, p) = (1-p)(1-\theta)^2 + p(1-\theta_a)^2$:
		\begin{align*}
			G_\ell(t) &= \left[ (1 - p(t; \theta_0, p_0))(1 - \theta(t; \theta_0, p_0))^2 + p(t; \theta_0, p_0)(1 - \theta_a)^2 \right] - (1 - \theta(t; \theta_0, 0))^2 \\
			&= (1 - p(t; \theta_0, p_0)) \left[ (1 - \theta(t; \theta_0, p_0))^2 - (1 - \theta(t; \theta_0, 0))^2 \right] \\
			& \quad + p(t; \theta_0, p_0) \left[ (1 - \theta_a)^2 - (1 - \theta(t; \theta_0, 0))^2 \right].
		\end{align*}
		Since $p_0 > 0$, Lemma~\ref{lm:coupling} implies that for all $t > 0$, the delegation level remains positive, $p(t; \theta_0, p_0) > 0$, and the AI-assisted skill is no more than the baseline skill:
		\[
		\theta(t; \theta_0, p_0) \leq \theta(t; \theta_0, 0).
		\]
		Consequently, the first bracketed term (skill loss difference) is non-negative for all $t > 0$:
		\begin{equation}
			\label{eq:skill_diff}
			(1 - \theta(t; \theta_0, p_0))^2 - (1 - \theta(t; \theta_0, 0))^2 \geq 0.
		\end{equation}
		
		\noindent
		The baseline skill $\theta(t; \theta_0, 0)$ evolves according to $\dot{\theta} = \theta(1-\theta)^2$, which strictly increases from $\theta_0$ to 1.
		Since $\theta_0 < \theta_a < 1$, by the Intermediate Value Theorem, there exists a unique finite time $t^* > 0$ such that $\theta(t^*; \theta_0, 0) = \theta_a$.
		For any time $t > t^*$, we have $\theta(t; \theta_0, 0) > \theta_a$, which implies:
		\[
		(1 - \theta_a)^2 - (1 - \theta(t; \theta_0, 0))^2 > 0.
		\]
		Combining this with \eqref{eq:skill_diff} and the fact that $p(t; \theta_0, p_0) > 0$, we conclude that $G_\ell(t)$ is a sum of strictly positive terms for all $t > t^*$.
		The crossing time is defined as $t_c := \inf\{t \ge 0 : \forall t' > t, G_\ell(t') > 0\}$.
		Since $G_\ell(t) > 0$ holds for all $t > t^*$, it follows that $t_c \le t^* < \infty$.
		Finally, because $G_\ell(t)$ becomes strictly positive whenever $\theta(t; \theta_0, 0) > \theta_a$, any period where the AI provides a benefit ($G_\ell(t) \le 0$) must occur while the baseline skill is still below the AI skill.
		Therefore, at the crossing time, the baseline skill satisfies $\theta(t_c; \theta_0, 0) \le \theta_a$.
	\end{proof}
	
	\subsection{Proof of Lemma \ref{lm:coupling}: monotone coupling between skill gap and delegation}
	\label{sec:proof_coupling}
	
	\begin{lemma}[\bf{Restatement of Lemma \ref{lm:coupling}}]
		Let $x=1-\theta$ denote the skill gap. For fixed $\kappa,\Delta>0$, the
		system in variables $(x,p)$ is cooperative:
		\[
		\dot x
		= x(1-x)\bigl[\Delta p(1-x)-(1-p)x\bigr],
		\qquad
		\dot p
		= \kappa p(1-p)\bigl[x^2-(1-\theta_a)^2\bigr].
		\]
		Consequently, if $x_0\leq \widetilde x_0$ and
		$p_0\leq \widetilde p_0$, then for all $t\geq 0$,
		\[
		x(t;x_0,p_0)\leq x(t;\widetilde x_0,\widetilde p_0),
		\qquad
		p(t;x_0,p_0)\leq p(t;\widetilde x_0,\widetilde p_0).
		\]
		Moreover, for fixed initial condition, increasing $\theta_a$ weakly
		increases $p(t)$ and $x(t)$ for all $t\geq 0$.
	\end{lemma}
	
	\begin{proof}
		In the variables $x=1-\theta$ and $p$, write the vector field as
		\[
		F_1(x,p)=x(1-x)\bigl[\Delta p(1-x)-(1-p)x\bigr],
		\]
		\[
		F_2(x,p;\theta_a)
		=\kappa p(1-p)\bigl[x^2-(1-\theta_a)^2\bigr].
		\]
		On the interior $(0,1)^2$,
		\[
		\frac{\partial F_1}{\partial p}
		=x(1-x)\bigl[\Delta(1-x)+x\bigr]\geq 0,
		\]
		and
		\[
		\frac{\partial F_2}{\partial x}
		=2\kappa p(1-p)x\geq 0.
		\]
		Thus the system is cooperative in \((x,p)\). By Kamke's comparison theorem
		for monotone dynamical systems~\cite{smith1995monotone}, larger initial skill gap
		\(x_0\), larger initial delegation \(p_0\), or larger AI skill \(\theta_a\)
		weakly increase future skill gap and delegation
		(Lemma~\ref{lm:coupling}).
		
		For the parameter $\theta_a$, observe that
		\[
		\frac{\partial F_1}{\partial \theta_a}=0,
		\qquad
		\frac{\partial F_2}{\partial \theta_a}
		=2\kappa p(1-p)(1-\theta_a)\geq 0.
		\]
		Therefore, increasing $\theta_a$ increases the vector field in the $p$
		coordinate and does not decrease it in the $x$ coordinate. 
		Applying the
		same comparison theorem to the parameterized systems gives the claimed
		monotonicity in AI skill.
	\end{proof}
	
	\subsection{Approximation of the stable manifold of the saddle point}
	\label{sec:basin_approximation}
	
	\begin{figure}[t]
		\centering
		\begin{subfigure}[b]{0.3\textwidth}
			\centering
			\includegraphics[width=\linewidth]{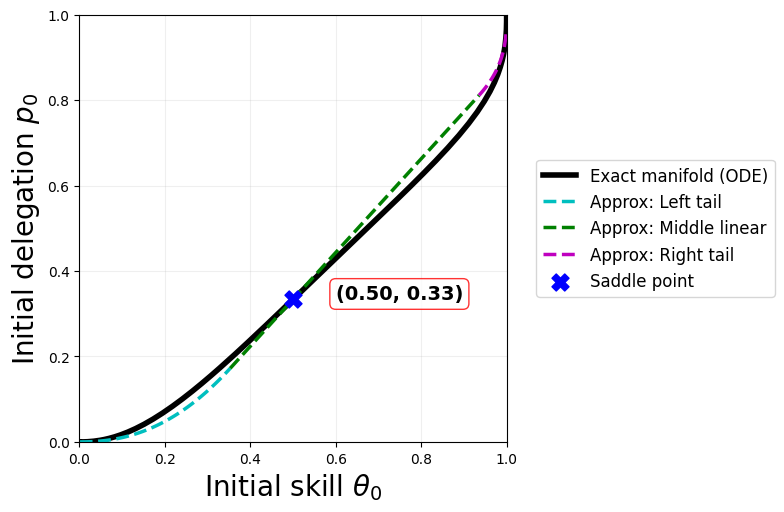}
			\caption{Default setting.}
			\label{fig:approximation}   
		\end{subfigure}
		\hfill
		\begin{subfigure}[b]{0.3\textwidth}
			\centering
			\includegraphics[width=\linewidth]{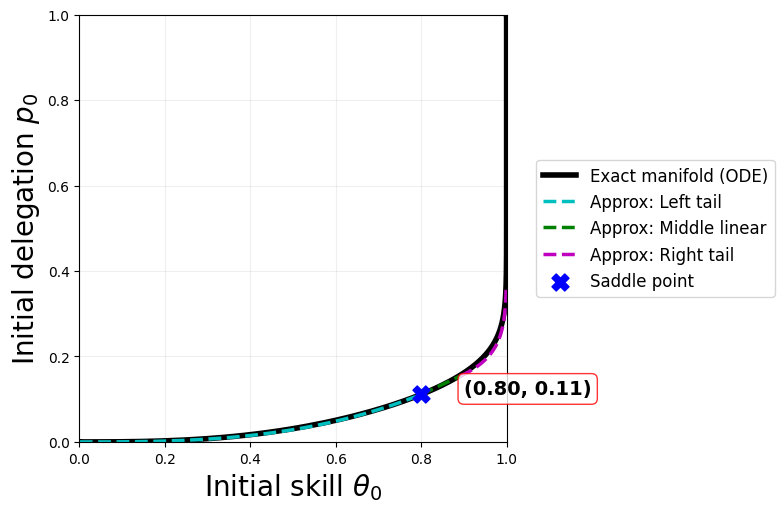}
			\caption{Varying $\theta_a = 0.8$.}
			\label{fig:approximation_AI}
		\end{subfigure}
		\hfill
		\begin{subfigure}[b]{0.3\textwidth}
			\centering
			\includegraphics[width=\linewidth]{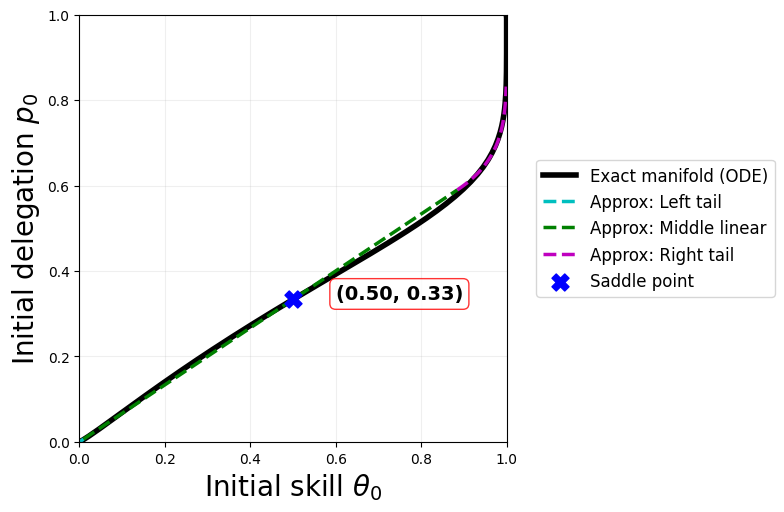}
			\caption{Varying $\kappa = 1.5$.}
			\label{fig:approximation_kappa}
		\end{subfigure}
		\caption{Plots illustrating the closeness between the stable manifold $\psi(\cdot)$ and its approximation $\widetilde{\psi}$, with default parameter settings $(\theta_a, \kappa, \Delta) = (0.5, 3, 2)$.}
		\label{fig:stable_other}
	\end{figure}
	
	In this section, we construct an explicit approximation $\widetilde{\psi}$ to the stable manifold $\psi$ (Eq.~\eqref{eq:basin}) to enable actionable predictions.
	Figure~\ref{fig:approximation} plots $\widetilde{\psi}$ under different parameter choices and shows that it closely tracks $\psi$.

	We construct $\widetilde{\psi}(\theta)$ by approximating the global stable manifold as a composite of three locally valid functions: a linear segment near the saddle point derived from the stable eigenvector, and two power-law segments near the corners $(0,0)$ and $(1,1)$ derived from the local eigenvalues.
	We determine the transition points $\theta_l$ and $\theta_r$ by enforcing continuity of both the function and its first derivative (differentiable continuity).
	
	\paragraph{1. Linear approximation at the saddle.}
	We first linearize the system around the saddle point $(\theta^\dagger, p^\dagger)$. The Jacobian matrix $J$ evaluated at the saddle is given by:
	\[
	J = 
	\begin{bmatrix}
		J_{11} & J_{12} \\
		J_{21} & J_{22}
	\end{bmatrix}
	=
	\begin{bmatrix}
		\partial_\theta \dot{\theta} & \partial_p \dot{\theta} \\
		\partial_\theta \dot{p} & \partial_p \dot{p}
	\end{bmatrix}.
	\]
	At the saddle, we have $\partial_p \dot{p} = 0$, so $J_{22} = 0$. The other entries are:
	\begin{align*}
		J_{11} = & \ (1-2\theta^\dagger)\left[(1-p^\dagger)(1-\theta^\dagger) - \Delta p^\dagger \theta^\dagger\right] \\
		& \ + \theta^\dagger(1-\theta^\dagger)[-(1-p^\dagger) - \Delta p^\dagger]= -\theta^\dagger(1-\theta^\dagger)(1-p^\dagger + \Delta p^\dagger) \\
		J_{12} = & \ \theta^\dagger(1-\theta^\dagger)[-(1-\theta^\dagger) - \Delta \theta^\dagger] = -\theta^\dagger(1-\theta^\dagger)(1 - \theta^\dagger(1-\Delta)) \\
		J_{21} = & \ -2\kappa p^\dagger(1-p^\dagger)(1-\theta^\dagger) 
	\end{align*}
	The characteristic equation for eigenvalues $\lambda$ is $\lambda^2 - J_{11}\lambda - J_{12}J_{21} = 0$. 
	The stable eigenvalue $\lambda_s$ is the negative root. 
	The associated eigenvector $v_s = [1, m^\dagger]^T$ satisfies  $(J_{11} - \lambda_s) + J_{12}m^\dagger = 0$.
	Solving for the slope $m^\dagger$ yields the expression:
	\[
	m^\dagger = \frac{\lambda_s - J_{11}}{J_{12}} = \frac{-J_{11} - \sqrt{J_{11}^2 + 4 J_{12} J_{21}}}{2 J_{12}}.
	\]
	Thus, near the saddle, the manifold is approximated by the line:
	\[
	\psi_{\text{mid}}(\theta) = p^\dagger + m^\dagger(\theta - \theta^\dagger).
	\]
	
	\paragraph{2. Power-law approximation at $(0,0)$.}
	Near the origin $(0,0)$, we approximate the dynamics by keeping only the lowest-order linear terms.
	For $\theta \approx 0, p \approx 0$:
	\begin{align*}
		\dot{\theta} &\approx \theta(1)(1 - 0) = \theta. \\
		\dot{p} &\approx \kappa p (1) (1 - (1-\theta_a)^2) = \kappa p (1 - (1-\theta^\dagger)^2).
	\end{align*}
	We define the ratio of the growth rates as $\beta_l := \frac{\dot{p}/\rho}{\dot{\theta}/\theta} = \kappa (1 - (1-\theta^\dagger)^2)$.
	The differential equation describing the trajectory shape is $\frac{dp}{d\theta} = \frac{\dot{p}}{\dot{\theta}} = \beta_l \frac{p}{\theta}$.
	Integrating this yields the power law form:
	\[
	\psi_{\text{left}}(\theta) = C_l \cdot \theta^{\beta_l}.
	\]
	
	\paragraph{3. Power-law approximation at $(1,1)$.}
	Near the corner $(1,1)$, we transform variables to $x = 1-\theta$ and $y = 1-p$, where $x,y \approx 0$.
	The dynamics become:
	\begin{align*}
		\dot{x} = -\dot{\theta} &\approx -\theta(1-\theta)(-\Delta p \theta) \approx \Delta x. \quad (\text{since } \theta \to 1, p \to 1). \\
		\dot{y} = -\dot{p} &\approx -\kappa p(1-p)(-(1-\theta_a)^2) \approx \kappa y (1-\theta^\dagger)^2.
	\end{align*}
	We define the ratio of decay rates as $\beta_r := \frac{\dot{y}}{\dot{x}} = \frac{\kappa (1-\theta^\dagger)^2}{\Delta}$.
	The trajectory shape satisfies $\frac{dy}{dx} = \beta_r \frac{y}{x}$. Integrating yields $y = C_r x^{\beta_r}$.
	Transforming back to original coordinates, we get:
	\[
	1 - \psi_{\text{right}}(\theta) = C_r (1-\theta)^{\beta_r} \implies \psi_{\text{right}}(\theta) = 1 - C_r (1-\theta)^{\beta_r}.
	\]
	
	\paragraph{4. Smooth pasting at breakpoints.}
	We determine the breakpoints $\theta_l$ and $\theta_r$ by matching the linear segment $\psi_{\text{mid}}$ to the power laws $\psi_{\text{left}}$ and $\psi_{\text{right}}$.
	\begin{itemize}
		\item \textbf{Left Breakpoint $\theta_l$:} We require $\psi_{\text{left}}(\theta_l) = \psi_{\text{mid}}(\theta_l)$ and $\psi'_{\text{left}}(\theta_l) = \psi'_{\text{mid}}(\theta_l)$.
		The derivative condition gives:
		\[
		C_l \beta_l \theta_l^{\beta_l - 1} = m^\dagger \implies \frac{\beta_l}{\theta_l} (C_l \theta_l^{\beta_l}) = m^\dagger \implies \frac{\beta_l}{\theta_l} \psi_{\text{mid}}(\theta_l) = m^\dagger.
		\]
		Substituting the linear form $\psi_{\text{mid}}(\theta_l) = p^\dagger + m^\dagger(\theta_l - \theta^\dagger)$:
		\[
		\beta_l (p^\dagger + m^\dagger \theta_l - m^\dagger \theta^\dagger) = m^\dagger \theta_l.
		\]
		Rearranging terms to solve for $\theta_l$:
		\[
		\beta_l (p^\dagger - m^\dagger \theta^\dagger) = m^\dagger \theta_l (1 - \beta_l) \implies \theta_l = \frac{\beta_l (p^\dagger - m^\dagger \theta^\dagger)}{m^\dagger (1 - \beta_l)}.
		\]
		
		\item \textbf{Right Breakpoint $\theta_r$:} Similarly, we match the function and derivative at $\theta_r$.
		The derivative condition for the right segment is:
		\[
		\psi'_{\text{right}}(\theta) = -C_r \beta_r (1-\theta)^{\beta_r - 1} (-1) = \frac{\beta_r}{1-\theta} C_r (1-\theta)^{\beta_r} = \frac{\beta_r}{1-\theta} (1 - \psi_{\text{right}}(\theta)).
		\]
		Setting this equal to $m^\dagger$ at $\theta_r$:
		\[
		\frac{\beta_r}{1-\theta_r} (1 - \psi_{\text{mid}}(\theta_r)) = m^\dagger.
		\]
		Substituting $\psi_{\text{mid}}(\theta_r) = p^\dagger + m^\dagger(\theta_r - \theta^\dagger)$:
		\[
		\beta_r (1 - p^\dagger - m^\dagger \theta_r + m^\dagger \theta^\dagger) = m^\dagger (1-\theta_r).
		\]
		Rearranging terms to solve for $\theta_r$:
		\[
		\beta_r (1 - p^\dagger + m^\dagger \theta^\dagger) - \beta_r m^\dagger \theta_r = m^\dagger - m^\dagger \theta_r
		\]
		\[
		m^\dagger \theta_r (1 - \beta_r) = m^\dagger - \beta_r(1 - p^\dagger + m^\dagger \theta^\dagger)
		\]
		\[
		\theta_r = \frac{m^\dagger - \beta_r(1 - p^\dagger + m^\dagger \theta^\dagger)}{m^\dagger (1 - \beta_r)}.
		\]
	\end{itemize}
	To ensure $\theta_l \leq \theta_a \leq \theta_r$, we let
	\[
	\theta_l := \min\{\theta^\dagger, \frac{\beta_l (p^\dagger - m^\dagger \theta^\dagger)}{m^\dagger (1 - \beta_l)}\}, \text{ and } \theta_r := \min\{1, \frac{m^\dagger - \beta_r(1-p^\dagger + m^\dagger \theta^\dagger)}{(1-\beta_r) m^\dagger} \}.
	\]
	In summary, the approximation $\widetilde{\psi}$ is:

	\begin{align}
		\label{eq:basin}
		\widetilde{\psi}(\theta) := 
		\begin{cases}
			(p^\dagger +  m^\dagger\cdot (\theta_l - \theta^\dagger)) \cdot (\frac{\theta}{\theta_l})^{\beta_l} & \text{ if } 0\leq \theta \leq \theta_l \\
			p^\dagger +  m^\dagger\cdot (\theta - \theta^\dagger)  & \text{ if } \theta_l < \theta < \theta_r \\
			1 - (1 - \psi_{\mathrm{mid}}(\theta_r)) \cdot (\frac{1-\theta}{1 - \theta_r})^{\beta_r}   & \text{ if } \theta_r \leq \theta \leq 1.
		\end{cases}
	\end{align}
	
	\subsection{Effects of parameters on the stable manifold}
	\label{sec:variation_other}

	\begin{figure}[t]
		\centering
		\begin{subfigure}[b]{0.3\textwidth}
			\centering
			\includegraphics[width=\linewidth]{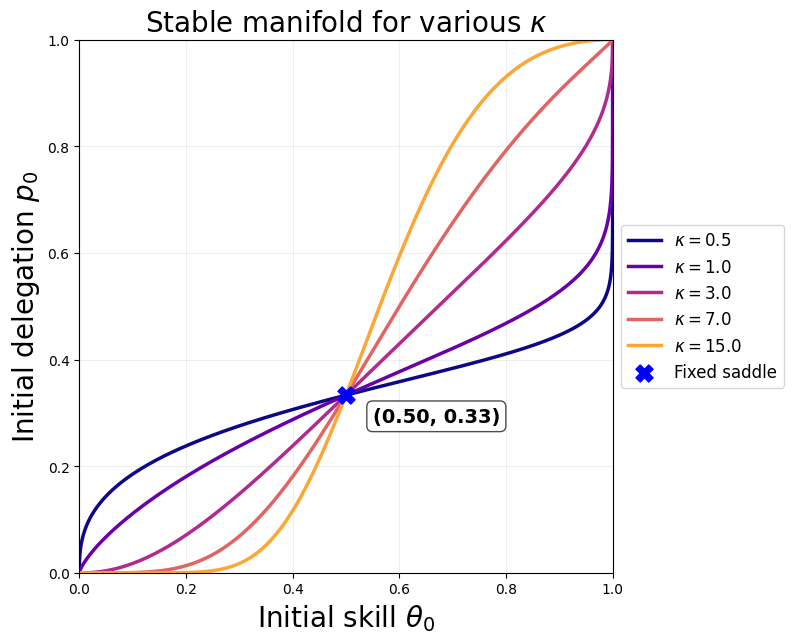}
			\caption{Effects of $\kappa$}
			\label{fig:stable_kappa}
		\end{subfigure}
		\qquad
		\begin{subfigure}[b]{0.3\textwidth}
			\centering
			\includegraphics[width=\linewidth]{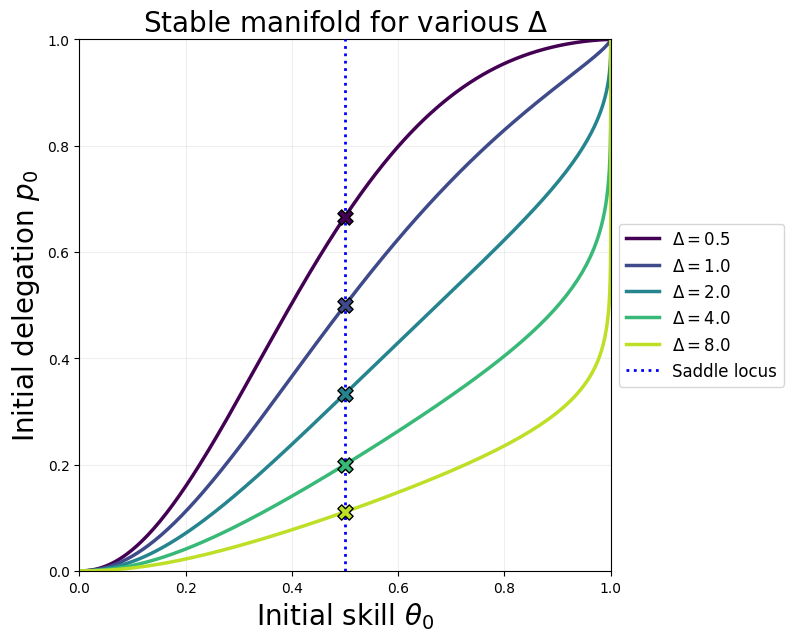}
			\caption{Effects of $\Delta$}
			\label{fig:stable_Delta}
		\end{subfigure}
		\caption{Plots illustrating the relationship between the basin boundary $\psi(\cdot)$ and model parameters $\theta_a, \kappa, \Delta$, with default setting $(\theta_a, \kappa, \Delta) = (0.5, 3,2)$.}
		\label{fig:effects} 
	\end{figure}

	Beyond Theorem~\ref{thm:monotonicity}, we also study how $\kappa$ and $\Delta$ affect the stable manifold.
	Figure~\ref{fig:effects} shows the resulting variation in the basin boundary as $\kappa$ and $\Delta$ change.
	We observe that changing $\kappa$ does not affect the saddle location, whereas changing $\Delta$ shifts the saddle along the line $\theta=\theta_a$.

	\begin{lemma}[\bf{Effects of $\kappa$ on stable manifold}]
		\label{lm:variation_kappa}
		Fix $\theta_a$.
		For $\theta \in (0,\theta_a)$, $\psi(\theta)$ is monotonically decreasing in $\kappa$ and continuously differentiable with respect to $\kappa$.
		For $\theta \in (\theta_a, 1)$, $\psi(\theta)$ is nondecreasing in $\kappa$ and continuously differentiable with respect to $\kappa$.
	\end{lemma}
	
	\begin{proof}
		The lemma establishes the behavior of the stable manifold $\psi(\theta)$ (separatrix) as the delegation speed parameter $\kappa$ varies. We treat differentiability and monotonicity separately.
		
		\paragraph{1. Differentiability.}
		The vector field defined by ODE~\eqref{eq:ODE_simplified} is $C^\infty$-smooth with respect to all state variables and the parameter $\kappa$. Since the interior equilibrium $(\theta^\dagger, p^\dagger)$ is a hyperbolic saddle point for all $\kappa > 0$ (Theorem~\ref{thm:convergence}), the Stable Manifold Theorem with parameters guarantees that the local stable manifold varies smoothly (differentiable) with respect to $\kappa$.
		Since the global separatrix $\psi(\theta)$ is the backward extension of this local manifold, and the flow is smooth, $\psi(\theta)$ is continuously differentiable with respect to $\kappa$.
		
		\paragraph{2. Monotonicity via vector field comparison.}
		We prove the monotonicity by comparing the slopes of the trajectories under two different parameters $\kappa_1 < \kappa_2$.
		Let $F_\kappa(\theta, p) = (\dot{\theta}, \dot{p}_\kappa)$ denote the vector field, where:
		\[
		\dot{\theta} = f(\theta, p), \quad \dot{p}_\kappa = \kappa \cdot g(\theta, p),
		\]
		with $g(\theta, p) = p(1-p)[(1-\theta)^2 - (1-\theta_a)^2]$. Note that $f$ is independent of $\kappa$, and $g$ captures the sign of the delegation drift.
		The slope of the vector field at any point $(\theta, p)$ where $f \neq 0$ is given by:
		\[
		m(\theta, p; \kappa) = \frac{d p}{d \theta} = \frac{\dot{p}_\kappa}{\dot{\theta}} = \kappa \frac{g(\theta, p)}{f(\theta, p)}.
		\]
		
		\paragraph{Case 1: Region $\theta \in (0, \theta_a)$.}
		In this region, the skill is below the AI skill, so $(1-\theta)^2 > (1-\theta_a)^2$, implying $g(\theta, p) > 0$.
		The separatrix $\psi$ connects the source $(0,0)$ to the saddle $(\theta_a, p^\dagger)$. Along this curve, $\theta$ is increasing, so $\dot{\theta} = f(\theta, p) > 0$.
		Since $g > 0$ and $f > 0$, the slope $m$ is positive.
		Comparing $\kappa_1 < \kappa_2$:
		\[
		m(\theta, p; \kappa_2) = \frac{\kappa_2}{\kappa_1} m(\theta, p; \kappa_1) > m(\theta, p; \kappa_1).
		\]
		The vector field for $\kappa_2$ is strictly ``steeper'' (points more upward) than for $\kappa_1$ everywhere in this region.
		Consider a point $z = (\theta_0, p_0)$ lying exactly on the separatrix $\psi_1$ for $\kappa_1$. Under the flow of $\kappa_1$, the trajectory $\gamma_1(t)$ starting at $z$ converges to the saddle.
		Under the flow of $\kappa_2$, the trajectory $\gamma_2(t)$ starting at $z$ has a strictly larger slope at every point than $\gamma_1$. Geometrically, this implies $\gamma_2$ must rise \emph{above} $\gamma_1$.
		Since $\gamma_1$ eventually hits the saddle, the steeper trajectory $\gamma_2$ must pass ``above'' the saddle (i.e., to the left of the stable manifold of the saddle in the local linearization), entering the basin of attraction of the low-skill equilibrium $(0,1)$.
		For $z$ to be in the low-skill basin of $\kappa_2$, it must lie \emph{above} the separatrix $\psi_2$.
		Thus, $p_0 > \psi_2(\theta_0)$. Since $p_0 = \psi_1(\theta_0)$, we have $\psi_1(\theta_0) > \psi_2(\theta_0)$.
		Therefore, $\psi(\theta)$ is monotonically decreasing in $\kappa$ on $(0, \theta_a)$.
		
		\paragraph{Case 2: Region $\theta \in (\theta_a, 1)$.}
		In this region, $(1-\theta)^2 < (1-\theta_a)^2$, implying $g(\theta, p) < 0$.
		The separatrix connects the saddle to the source $(1,1)$ (in backward time), or equivalently, trajectories flow from the saddle towards the high-skill sink $(1,0)$. 
		However, technically $\psi$ is defined as the stable manifold. 
		In this region, $\psi$ separates flow to $(0,1)$ and $(1,0)$.

		Consider the slope again. 
		Note that $g(\theta, p)$ is negative, which means delegation is being suppressed.
		Increasing $\kappa$ makes $\dot{p}$ \emph{more negative}.
		Consider a point $z$ on the separatrix $\psi_1$ (for $\kappa_1$). 
		Under $\kappa_1$, it flows to the saddle.
		Under $\kappa_2$, the downward push is stronger. 
		The trajectory starting at $z$ will drop \emph{below} the trajectory of $\kappa_1$.
		Since the $\kappa_1$ trajectory hits the saddle, the $\kappa_2$ trajectory (being lower) will pass ``below'' the saddle, entering the basin of attraction of the high-skill equilibrium $(1,0)$.
		For $z$ to be in the high-skill basin of $\kappa_2$, it must lie \emph{below} the separatrix $\psi_2$.
		Thus, $p_0 < \psi_2(\theta_0)$. Since $p_0 = \psi_1(\theta_0)$, we have $\psi_1(\theta_0) < \psi_2(\theta_0)$.
		Therefore, $\psi(\theta)$ is nondecreasing in $\kappa$ on $(\theta_a, 1)$.
	\end{proof}
	
	\begin{lemma}[\bf{Effects of $\Delta$ on stable manifold}]
		\label{lm:variation_Delta}
		For $\theta \in (0,1)$, $\psi(\theta)$ is monotonically decreasing in $\Delta$ and continuously differentiable with respect to $\Delta$.
	\end{lemma}
	
	\begin{proof}
		We analyze the dependence of the stable manifold $\psi(\theta)$ on the drift parameter $\Delta$.
		
		\paragraph{Differentiability.}
		The vector field of ODE~\eqref{eq:ODE_simplified} is smooth ($C^\infty$) with respect to the state variables $(\theta, p)$ and the parameter $\Delta$.
		The interior equilibrium $(\theta^\dagger, p^\dagger)$ depends smoothly on $\Delta$ (specifically, $p^\dagger(\Delta) = \frac{1-\theta_a}{1-\theta_a + \Delta \theta_a}$).
		Since the equilibrium remains a hyperbolic saddle for all $\Delta > 0$, the Stable Manifold Theorem with parameters \cite{Perko2013} guarantees that the local stable manifold varies smoothly (differentiable) with respect to $\Delta$.
		By the smoothness of the flow extending the local manifold to the global separatrix, $\psi(\theta)$ is continuously differentiable with respect to $\Delta$.
		
		\paragraph{Monotonicity via trajectory comparison.}
		We prove that $\psi(\theta)$ decreases as $\Delta$ increases by establishing a monotonicity property for the trajectories and using a basin-of-attraction argument.

		Let $\Delta_1 < \Delta_2$. Consider two systems starting from the same initial state $(\theta_0, p_0)$.
		Comparing the drift functions:
		\begin{itemize}
			\item The delegation update $\dot{p} = g(\theta, p)$ is independent of $\Delta$.
			\item The skill update $\dot{\theta} = f(\theta, p; \Delta)$ satisfies:
			\[
			\frac{\partial f}{\partial \Delta} = -\theta^2(1-\theta)p < 0 \quad \text{for } \theta, p \in (0,1).
			\]
		\end{itemize}
		Since a higher $\Delta$ strictly reduces the skill growth rate, the system forms a monotone dynamical system with respect to this parameter.
		Specifically, similar to Lemma~\ref{lm:coupling}, strictly lower skill growth leads to lower future skill $\theta(t)$, which in turn induces higher delegation $p(t)$.
		Thus, for any finite $t > 0$, the solutions satisfy:
		\[
		\theta(t; \Delta_2) \le \theta(t; \Delta_1) \quad \text{and} \quad p(t; \Delta_2) \ge p(t; \Delta_1).
		\]
		Taking the limit $t \to \infty$, the equilibria reached must satisfy the same ordering: $\theta^\infty(\Delta_2) \le \theta^\infty(\Delta_1)$.

		Let $\psi_1$ and $\psi_2$ be the separatrices for $\Delta_1$ and $\Delta_2$ respectively.
		We claim $\psi_2(\theta) < \psi_1(\theta)$ for all $\theta \in (0,1)$.
		Suppose, for contradiction, that there exists some $\theta^*$ such that $\psi_2(\theta^*) \ge \psi_1(\theta^*)$.
		We can choose an initial delegation $p_0$ such that $\psi_1(\theta^*) < p_0 < \psi_2(\theta^*)$.
		Consider the asymptotic outcome of a learner starting at $(\theta^*, p_0)$ under both regimes:
		\begin{itemize}
			\item \textbf{Under $\Delta_1$:} Since $p_0 > \psi_1(\theta^*)$, the state is in the low-skill basin (above the separatrix). The learner converges to the low-skill equilibrium $(\theta^\infty = 0)$.
			\item \textbf{Under $\Delta_2$:} Since $p_0 < \psi_2(\theta^*)$, the state is in the high-skill basin (below the separatrix). The learner converges to the high-skill equilibrium $(\theta^\infty = 1)$.
		\end{itemize}
		This implies $\theta^\infty(\Delta_2) = 1 > 0 = \theta^\infty(\Delta_1)$.
		However, this contradicts the trajectory monotonicity, which requires $\theta^\infty(\Delta_2) \le \theta^\infty(\Delta_1)$.
		Therefore, the assumption is false, and it must hold that $\psi_2(\theta) < \psi_1(\theta)$.
		Thus, $\psi(\theta)$ is monotonically decreasing in $\Delta$.
	\end{proof}
	
	\section{Interventions}
	\label{sec:intervention}
	
	The basin characterization gives a geometric way to reason about interventions:
	an intervention is effective if it moves the separatrix $p=\psi(\theta)$ so
	that more initial conditions converge to the high-skill equilibrium. We study
	three representative mechanisms. 
	Section~\ref{sec:conservative_reliance}
	considers a learner-side intervention that makes reliance updates more
	conservative after evidence against AI. 
	Section~\ref{sec:forward_looking} considers another learner-side
	forward-looking cost that makes users internalize the long-run value of
	preserving skill.
	Section~\ref{sec:delegation_penalty} considers a teacher-side intervention that penalizes detected AI-assisted outputs. 
	
	\begin{figure*}[t]
		\centering
		\begin{subfigure}[b]{0.32\textwidth}
			\centering
			\includegraphics[width=\linewidth]{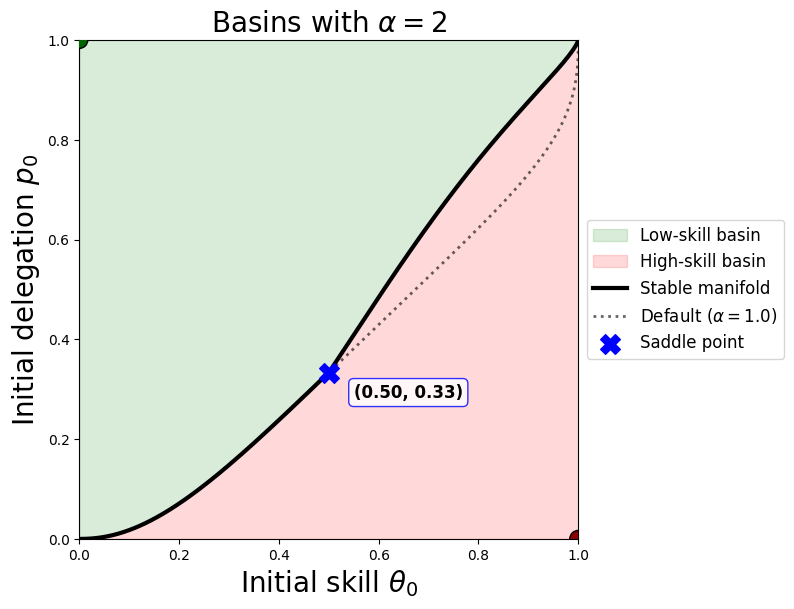}
			\caption{\footnotesize Conservative reliance updates \eqref{eq:asymmetric}}
			\label{fig:asymmetric_delegation}
		\end{subfigure}
		\hfill
		\begin{subfigure}[b]{0.32\textwidth}
			\centering
			\includegraphics[width=\linewidth]{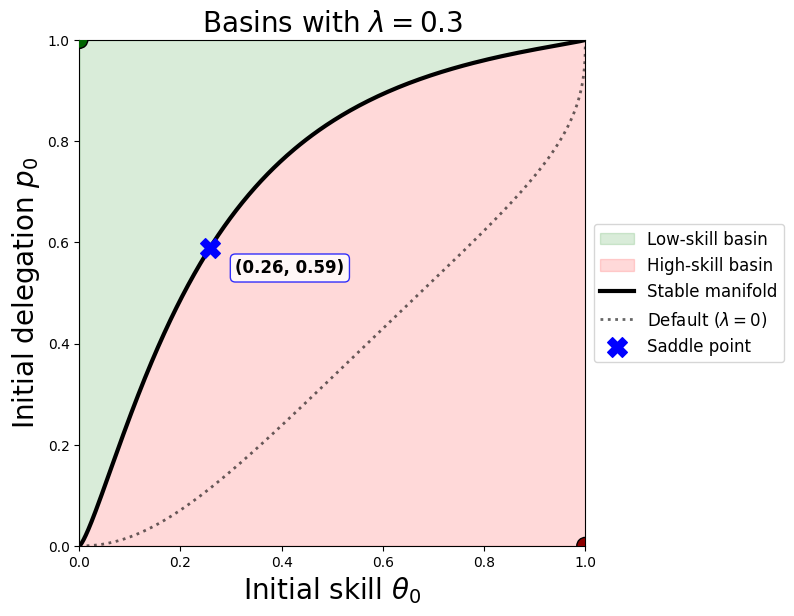}
			\caption{Forward-looking delegation \eqref{eq:forward_looking}}
			\label{fig:forward_looking_I}
		\end{subfigure}
		\hfill
		\begin{subfigure}[b]{0.32\textwidth}
			\centering
			\includegraphics[width=\linewidth]{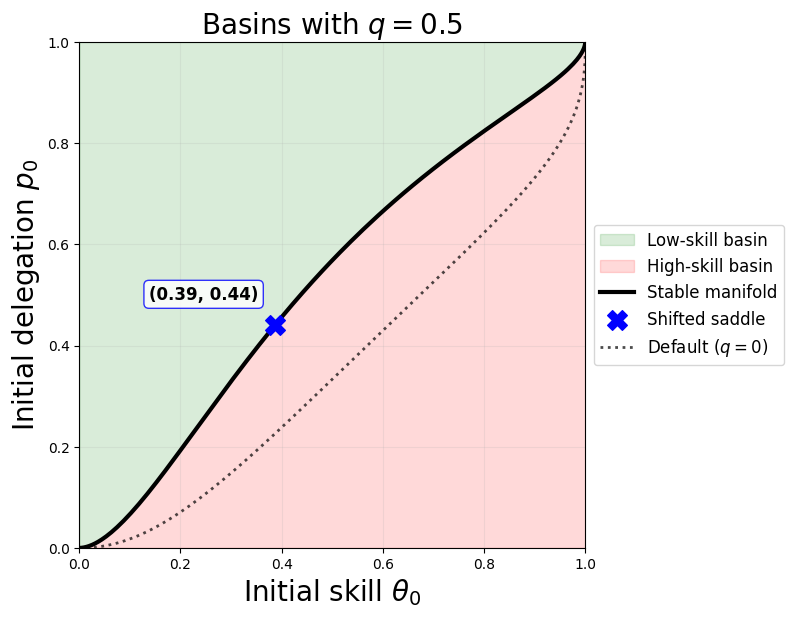}
			\caption{\footnotesize Delegation penalties \eqref{eq:alternative}}
			\label{fig:alternative_loss}
		\end{subfigure}
		
		\caption{Plots illustrating how the basins of attraction change under learner- and teacher-side interventions, with default parameter settings $(\theta_a, \kappa, \Delta) = (0.5, 3, 2)$.
		}
		\label{fig:extension}
	\end{figure*}

	\subsection{Conservative reliance updates}
	\label{sec:conservative_reliance}
	
	A natural learner-side intervention is to make users more conservative after
	evidence that AI performs worse than independent work. This does not change the
	performance loss, the skill-update potential $\Phi_s$, or the
	delegation-update potential $\Phi_d$; it only changes how the learner maps the
	delegation-update signal into changes in reliance.
	
	Let
	\[
	D(\theta):=(1-\theta)^2-(1-\theta_a)^2
	\]
	denote the baseline performance gap driving delegation. Instead of using the
	symmetric signal $D(\theta)$, we use
	\[
	D_\alpha(\theta)
	=
	[D(\theta)]_+
	-
	\alpha[-D(\theta)]_+,
	\]
	where $\alpha>0$ controls the relative sensitivity to negative evidence about
	AI. When $\alpha>1$, the learner reduces reliance faster whenever independent
	work outperforms AI.
	
	Applying the same multiplicative-update geometry gives
	\begin{align}
		\label{eq:asymmetric}
		\begin{aligned}
			\dot{\theta}
			&=
			\theta(1-\theta)\left((1-p)(1-\theta)-\Delta p\theta\right),\\
			\dot p
			&=
			\kappa p(1-p)
			\left(
			\left[(1-\theta)^2-(1-\theta_a)^2\right]_+
			-
			\alpha\left[(1-\theta_a)^2-(1-\theta)^2\right]_+
			\right).
		\end{aligned}
	\end{align}
	When $\alpha=1$, this reduces to the baseline ODE~\eqref{eq:ODE_simplified}.
	
	Figure~\ref{fig:asymmetric_delegation} shows the basin structure for
	$\alpha=2$. Compared with the symmetric baseline, conservative reliance updates
	expand the high-skill basin by making learners more responsive to evidence
	against AI. Thus increasing $\alpha$ can be interpreted as a learner-side
	intervention that encourages vigilance against over-reliance without
	restricting AI access.
	
	\subsection{Forward-looking delegation costs}
	\label{sec:forward_looking}
	
	A second intervention is to make learners internalize the long-run value of
	preserving their own skill. This can represent an explicit self-control cost,
	a training objective that rewards practice, or an institutional design that
	adds friction to unnecessary delegation. Mathematically, the intervention leaves
	the skill-update potential $\Phi_s$ unchanged and modifies the
	delegation-update potential $\Phi_d$ by adding a perceived cost
	$\lambda\ge 0$ to AI-assisted task completion:
	\[
	\Phi_d^\lambda(\theta,p)
	=
	(1-p)(1-\theta)^2
	+
	p\bigl((1-\theta_a)^2+\lambda\bigr).
	\]
	Applying the same multiplicative-update geometry gives
	\begin{align}
		\label{eq:forward_looking}
		\begin{aligned}
			\dot{\theta}
			&=
			\theta(1-\theta)\left((1-p)(1-\theta)-\Delta p\theta\right),\\
			\dot p
			&=
			\kappa p(1-p)
			\left(
			(1-\theta)^2-(1-\theta_a)^2-\lambda
			\right).
		\end{aligned}
	\end{align}
	When $\lambda=0$, this reduces to ODE~\eqref{eq:ODE_simplified}.
	
	Equivalently, the intervention makes AI behave as if it had a higher effective
	loss. When $(1-\theta_a)^2+\lambda\le 1$, define
	$\widehat{\theta}_a$ by
	\[
	(1-\widehat{\theta}_a)^2
	=
	(1-\theta_a)^2+\lambda .
	\]
	Then $\widehat{\theta}_a\le \theta_a$, so the intervention is equivalent to
	reducing the effective attractiveness of AI in the delegation update. 
	By Theorem \ref{thm:monotonicity}, this shifts the basin boundary in favor of the high-skill
	regime. 
	Figure~\ref{fig:forward_looking_I} illustrates this effect: larger
	$\lambda$ weakens over-delegation and enlarges the set of initial conditions
	that converge to sustained skill.

	\subsection{Delegation penalties}
	\label{sec:delegation_penalty}
	
	We next consider a penalty on AI-generated outputs. This specification modifies
	the performance loss and delegation-update potential, but leaves the
	skill-update potential unchanged. Let $q\in[0,1]$ be the probability that an
	AI-generated output is detected. If detected, the AI output incurs the harsher
	loss $1-\theta_a$ instead of $(1-\theta_a)^2$. 
	Therefore, the effective
	AI-side loss becomes
	\[
	L_a^{\mathrm{pen}}
	:=
	(1-q)(1-\theta_a)^2+q(1-\theta_a).
	\]
	The delegation-update potential is
	\[
	\Phi_d^{\mathrm{pen}}(\theta,p)
	=
	(1-p)(1-\theta)^2+pL_a^{\mathrm{pen}}.
	\]
	The normalized dynamics are then
	\begin{align}
		\label{eq:alternative}
		\begin{aligned}
			\dot{\theta}
			&=
			\theta(1-\theta)\left((1-p)(1-\theta)-\Delta p\theta\right),\\
			\dot p
			&=
			\kappa p(1-p)
			\left(
			(1-\theta)^2
			-
			\bigl((1-q)(1-\theta_a)^2+q(1-\theta_a)\bigr)
			\right).
		\end{aligned}
	\end{align}
	When $q=0$, this reduces to the baseline ODE~\eqref{eq:ODE_simplified}.
	Increasing $q$ raises the effective loss of AI delegation, thereby weakening
	the incentive to rely on AI and shifting the basin boundary toward greater
	independent practice. 
	Figure~\ref{fig:alternative_loss} illustrates this
	effect. 
	This suggests that detection-based penalties can serve as an external intervention for reducing excessive AI reliance
	\cite{Veselovsky2023ArtificialAA}.
	

	\section{Conclusion, limitations, and future work}
	\label{sec:conclusion}
	
	We developed a dynamical framework for repeated AI-assisted learning in which
	human skill and AI reliance co-evolve. In the main specification, adaptive
	delegation changes the global geometry of learning: fixed delegation yields a
	one-dimensional learning-decay process with a single attractor, whereas
	adaptive delegation can produce two attracting regimes separated by the stable
	manifold of an interior saddle. This boundary determines which initial
	skill-reliance states lead to sustained learning and which lead to persistent
	delegation and skill loss. Increasing AI capability reshapes the boundary,
	expanding the low-skill basin even as assistance appears beneficial for longer.
	
	The resulting risk is not AI delegation itself, but the feedback loop between
	short-run performance gains and reduced opportunities for practice. In this
	sense, the model formalizes a mechanism for cognitive debt: assistance can
	improve immediate outcomes while shifting the learner toward a region of the
	state space from which independent skill is harder to recover. Interventions
	that reduce early reliance, preserve opportunities for practice, or modify the
	feedback signal associated with delegation can be understood as attempts to
	reshape the basin boundary toward sustained skill development.
	
	The model is intentionally stylized. It isolates repeated adaptive delegation
	rather than the full space of human-AI workflows, and it is meant to
	characterize qualitative dynamics rather than serve as a calibrated structural
	model of human behavior. Across alternative specifications, the phase
	structure is observed to persist under jagged or misperceived AI performance, alternative
	losses, multiple skills, learning from AI, skill-dependent AI performance, and
	changing AI capability. 
	Section~\ref{sec:usability} outlines how the model parameters can be
	estimated from repeated-task data, illustrates the calculation on a synthetic
	example, and states empirical predictions that would support or challenge the
	basin mechanism.
	Future work should connect these dynamics to richer workflows that separate
	generation, review, repair, verification, and effort allocation.
	
	\section*{Acknowledgments} We thank 
	Mariia Eremeeva for useful discussions. This work was funded by NSF Award CCF-2112665.

	\bibliography{references}
	\bibliographystyle{plain}
	
	\appendix
	
	\section{Details from Section \ref{sec:extension}: alternative  specifications and robustness}
	\label{sec:details_extension}
	
	\begin{table}[h]
		\centering
		\label{tab:robustness_status}
		\begin{tabular}{ll}
			\toprule
			Variant & Status \\
			\midrule
			Jagged AI performance & Exact reduction via effective loss \\
			Biased perceived AI skill & Analytically characterized \\
			Power-\(z\) loss & Analytic nullcline deformation and simulations \\
			Stochastic delegation updates & Simulation-supported \\
			Multiple skills & Simulation-supported \\
			Learning from AI during delegation & Analytic equilibrium shift and simulations \\
			Improving AI capability & Quasi-static analysis and simulations \\
			Skill-dependent AI performance & Simulation-supported \\
			\bottomrule
		\end{tabular}
		\caption{Status of robustness variants.}
	\end{table}
	Recall that ODE~\eqref{eq:ODE_simplified} relies on several
	baseline modeling choices: AI performance is represented by a fixed and
	correctly perceived effective skill $\theta_a$; performance is measured by
	squared loss; learner ability is represented by a single scalar skill
	$\theta$; and delegation provides no direct learning signal. These assumptions
	make the phase-plane analysis tractable. 
	In this section, we relax them and
	study whether the qualitative conclusions depend on these particular
	specifications.
	
	The variants below modify one of the main components of the framework: the
	performance loss $\ell$, the skill-update potential $\Phi_s$, the
	delegation-update potential $\Phi_d$, or the learner/AI state evolution. We
	first consider uncertainty and misperception in AI performance: jagged AI
	performance in Section~\ref{sec:noisy_AI} and biased perceived AI skill in
	Section~\ref{sec:biased_AI_skill}. We then study alternative performance-loss
	specifications in Section~\ref{sec:model_robust}. Finally, we consider richer
	learner and AI states, including multiple skills
	(Section~\ref{sec:multiple_skill}), learning from AI during delegation
	(Section~\ref{sec:learner_AI_extension}), improving AI capability over time
	(Section~\ref{sec:improving_AI}), and skill-dependent AI performance
	(Section~\ref{sec:skill_dependent_AI}).
	
	For each specification, we state the modified dynamical system, identify which
	model component changes, and provide theoretical analysis or simulations when
	available. Across these variants, the same qualitative phase structure
	is observed to persist: adaptive reliance can generate two attracting regimes separated by a
	basin boundary, so small differences in initial skill or reliance can lead to
	different long-run outcomes.
	
	\subsection{Jagged AI performance}
	\label{sec:noisy_AI}
	
	With a jagged AI, the performance loss of AI-generated outputs is random.
	We model this by letting $s_{t,r}\sim \mu_a$ denote the realized AI skill on task $r$ in time window $t$. 
	Define the effective AI loss
	\[
	L_a := \mathbb{E}_{s\sim\mu_a}\bigl[(1-s)^2\bigr].
	\]
	Compared with the baseline model, the jagged-AI extension changes only the AI-side performance loss from $\ell(a) = (1-\theta_a)^2$ to $L_a$.
	Equivalently, define the loss-equivalent effective AI skill
	$\theta_a^{\mathrm{eff}}$ by
	\[
	(1-\theta_a^{\mathrm{eff}})^2=L_a.
	\]
	For notational simplicity, we write $\theta_a$ for this effective skill below.
	
	As in Section~\ref{sec:derivation}, each time window $t$ contains $m$
	comparable tasks. During this window, $\theta(t)$ and $p(t)$ are held fixed,
	while task-level delegation decisions vary across tasks. For each
	$r\in[m]$, let
	\[
	X_{t,r}\sim \mathrm{Bern}(p(t)),
	\qquad
	s_{t,r}\sim \mu_a,
	\]
	independently conditional on the current state $(\theta(t),p(t))$. Let
	\[
	\bar X_t := \frac{1}{m}\sum_{r=1}^m X_{t,r}
	\]
	be the realized fraction of delegated tasks.
	
	The learning mechanism within one time window is as follows.
	
	\begin{itemize}
		\item \textbf{(Task completion stage)}
		For each task $r\in[m]$, the learner delegates to AI if $X_{t,r}=1$ and
		performs the task independently if $X_{t,r}=0$.
		
		\item \textbf{(Evaluation stage)}
		The teacher provides task-level feedback. If the task is performed
		independently, the loss is $(1-\theta(t))^2$. If the task is delegated to AI,
		the loss is $(1-s_{t,r})^2$. Thus
		\[
		\ell_{t,r}
		=
		(1-X_{t,r})(1-\theta(t))^2
		+
		X_{t,r}(1-s_{t,r})^2 .
		\]
		
		\item \textbf{(Skill update)}
		The skill update is unchanged from the baseline mechanism: independent work
		generates an error-driven learning signal, while delegation induces non-use
		decay. Averaging over the $m$ tasks in the window gives
		\[
		\theta(t+1)
		=
		\theta(t)
		+
		2\eta\theta(t)(1-\theta(t))
		\left[
		(1-\bar X_t)(1-\theta(t))
		+
		\Delta\bar X_t(\theta_d-\theta(t))
		\right].
		\]
		
		\item \textbf{(Delegation update)}
		The delegation update uses the realized feedback from delegated tasks and the
		learner's belief $L_a$ about the expected AI loss on non-delegated tasks. The
		realized performance gap driving the update is
		\[
		G_t
		:=
		\frac{1}{m}\sum_{r=1}^m
		\left[
		(1-X_{t,r})\bigl((1-\theta(t))^2-L_a\bigr)
		+
		X_{t,r}\bigl((1-\theta(t))^2-(1-s_{t,r})^2\bigr)
		\right].
		\]
		The delegation level updates as
		\[
		p(t+1)
		=
		p(t)
		+
		\eta\kappa p(t)(1-p(t))G_t .
		\]
	\end{itemize}
	
	\noindent
	In summary, the stochastic dynamics under jagged AI are
	\begin{align*}
		\begin{aligned}
			X_{t,r} &\sim \mathrm{Bern}(p(t)),\qquad r=1,\ldots,m,\\
			s_{t,r} &\sim \mu_a,\qquad r=1,\ldots,m,\\
			\theta(t+1)
			&=
			\theta(t)
			+
			2\eta\theta(t)(1-\theta(t))
			\left[
			(1-\bar X_t)(1-\theta(t))
			+
			\Delta\bar X_t(\theta_d-\theta(t))
			\right],\\
			p(t+1)
			&=
			p(t)
			+
			\eta\kappa p(t)(1-p(t)) \\
			& \quad \times 
			\frac{1}{m}\sum_{r=1}^m
			\left[
			(1-X_{t,r})\bigl((1-\theta(t))^2-L_a\bigr)
			+
			X_{t,r}\bigl((1-\theta(t))^2-(1-s_{t,r})^2\bigr)
			\right].
		\end{aligned}
	\end{align*}
	The key difference from Dynamics~\eqref{eq:dynamics} is the introduction of
	the random AI skill $s_{t,r}$ and the resulting stochasticity in the
	delegation update.
	
	We now track this stochastic dynamics to its expectation. Since
	$X_{t,r}$ and $s_{t,r}$ are independent conditional on
	$(\theta(t),p(t))$, and $\mathbb{E}[s_{t,r}]=\mathbb{E}_{s\sim\mu_a}[s]$ with
	$\mathbb{E}[(1-s_{t,r})^2]=L_a$, we have
	\[
	\mathbb{E}[\bar X_t\mid \mathcal{F}_t]=p(t)
	\]
	and
	\begin{align*}
		\mathbb{E}[G_t\mid \mathcal{F}_t]
		&=
		(1-p(t))\bigl((1-\theta(t))^2-L_a\bigr)
		+
		p(t)\bigl((1-\theta(t))^2-L_a\bigr)\\
		&=
		(1-\theta(t))^2-L_a .
	\end{align*}
	Therefore,
	\begin{align*}
		\begin{aligned}
			\mathbb{E}[\theta(t+1)\mid\mathcal{F}_t]
			&=
			\theta(t)
			+
			2\eta\theta(t)(1-\theta(t))
			\left[
			(1-p(t))(1-\theta(t))
			+
			\Delta p(t)(\theta_d-\theta(t))
			\right],\\
			\mathbb{E}[p(t+1)\mid\mathcal{F}_t]
			&=
			p(t)
			+
			\eta\kappa p(t)(1-p(t))
			\left((1-\theta(t))^2-L_a\right).
		\end{aligned}
	\end{align*}
	
	After normalizing $\theta_d=0$ and applying the same time rescaling as in
	ODE~\eqref{eq:ODE_simplified}, the jagged-AI dynamics become
	\begin{align*}
		\begin{aligned}
			\dot{\theta}
			&=
			\theta(1-\theta)
			\left[
			(1-p)(1-\theta)
			+
			\Delta p(\theta_d-\theta)
			\right],\\
			\dot p
			&=
			\kappa p(1-p)
			\left((1-\theta)^2-\mathbb{E}_{s\sim\mu_a}[(1-s)^2]\right).
		\end{aligned}
	\end{align*}
	Using the fact that $(1-\theta_a)^2=\mathbb{E}_{s\sim\mu_a}[(1-s)^2]$, this dynamics has the same form as the baseline ODE with the
	effective AI skill $\theta_a$.
	
	\subsection{Biased perceived AI skill}
	\label{sec:biased_AI_skill}
	
	ODE~\eqref{eq:ODE_simplified} assumes that the learner updates delegation
	using the true effective AI skill $\theta_a$. We now relax this assumption by
	allowing the learner's perceived AI skill to differ from the true one. Let
	$\widetilde{\theta}_a\in[0,1]$ denote the AI skill perceived by the learner.
	Realized task outcomes are still evaluated using the true AI skill $\theta_a$;
	only the learner's reliance update uses $\widetilde{\theta}_a$.
	
	This extension leaves the true performance loss and the skill-update potential
	$\Phi_s$ unchanged. 
	The only change is in the
	delegation-update potential used by the learner:
	\[
	\widetilde{\Phi}_d(\theta,p)
	=
	(1-p)(1-\theta)^2
	+
	p(1-\widetilde{\theta}_a)^2 .
	\]
	Applying the same multiplicative-update geometry gives
	\begin{align}
		\label{eq:ODE_biased_delegation}
		\begin{aligned}
			\dot{\theta}
			&=
			\theta(1-\theta)\left((1-p)(1-\theta)-\Delta p\theta\right),\\
			\dot p
			&=
			\kappa p(1-p)
			\left((1-\theta)^2-(1-\widetilde{\theta}_a)^2\right).
		\end{aligned}
	\end{align}
	When $\widetilde{\theta}_a=\theta_a$, this reduces to the baseline
	ODE~\eqref{eq:ODE_simplified}.
	
	This specification separates actual AI performance from perceived AI performance.
	If $\widetilde{\theta}_a>\theta_a$, the learner overestimates AI quality, so
	the delegation nullcline shifts from $\theta=\theta_a$ to
	$\theta=\widetilde{\theta}_a$. As a result, delegation increases over a larger
	range of human skill levels, enlarging the high-reliance region. Conversely,
	if $\widetilde{\theta}_a<\theta_a$, the learner underestimates AI quality,
	which shrinks the region in which delegation increases and favors independent
	practice.
	
	\subsection{Robustness to $z$-th power loss}
	\label{sec:model_robust}
	
	We consider a $z$-th power loss, which changes the sensitivity of feedback
	to errors:
	\[
	\ell_z(\theta)=(1-\theta)^z,\qquad z\ge 1.
	\]
	To keep the skill-decay term consistent with the same geometry, we also use
	the power distance $g_z(\theta,0)=\theta^z$ in the normalized setting
	$\theta_d=0$. Thus the skill-update and delegation-update potentials become
	\[
	\Phi_s^z(\theta,p)
	=
	(1-p)(1-\theta)^z+\Delta p\theta^z,
	\qquad
	\Phi_d^z(\theta,p)
	=
	(1-p)(1-\theta)^z+p(1-\theta_a)^z.
	\]
	Applying the same multiplicative-update geometry as in the baseline model,
	and absorbing the common factor $z$ into the time scale and the free
	relative-speed parameter $\kappa$, gives
	\begin{align}
		\label{eq:power_z_ODE}
		\begin{aligned}
			\dot{\theta}
			&=
			\theta(1-\theta)
			\left(
			(1-p)(1-\theta)^{z-1}
			-
			\Delta p\theta^{z-1}
			\right),\\
			\dot p
			&=
			\kappa p(1-p)
			\left(
			(1-\theta)^z-(1-\theta_a)^z
			\right).
		\end{aligned}
	\end{align}
	The case $z=2$ recovers the baseline squared-loss dynamics.
	
	The qualitative phase structure is observed to persist. The $p$-nullcline remains
	$\theta=\theta_a$, while the interior $\theta$-nullcline is
	\[
	p=\psi_z(\theta)
	:=
	\frac{(1-\theta)^{z-1}}
	{(1-\theta)^{z-1}+\Delta\theta^{z-1}}.
	\]
	Hence the interior equilibrium is
	\[
	(\theta^\dagger,p^\dagger)
	=
	\left(
	\theta_a,\,
	\frac{(1-\theta_a)^{z-1}}
	{(1-\theta_a)^{z-1}+\Delta\theta_a^{z-1}}
	\right).
	\]
	Changing $z$ deforms the nullclines and changes basin sizes, but preserves the
	two-regime structure. Thus the path-dependent dynamics are not an artifact of
	the quadratic loss. Figure~\ref{fig:alternative_learning} illustrates this
	extension.

	\subsection{Multiple-skill extension}
	\label{sec:multiple_skill}
	
	\begin{figure*}[t]
		\centering
		\begin{subfigure}[b]{0.27\textwidth}
			\centering
			\includegraphics[width=\linewidth]{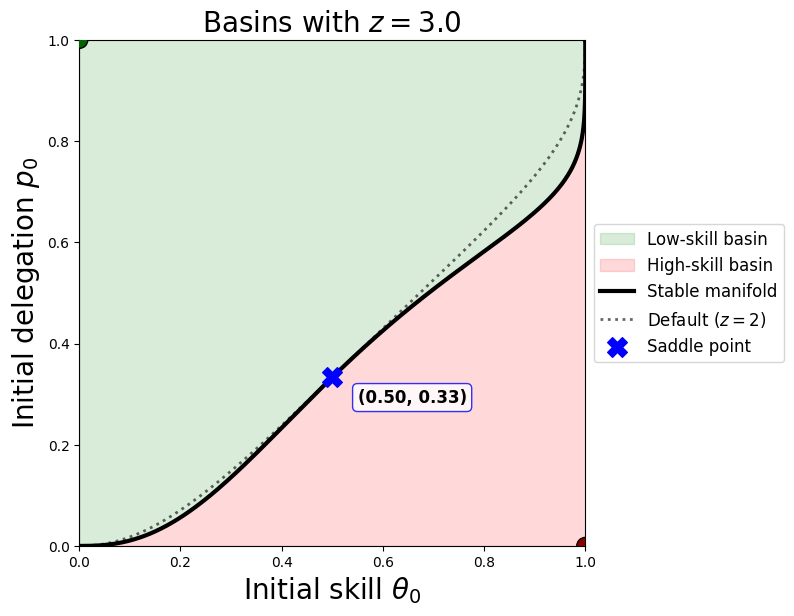}
			\caption{Basins of $z$-th power loss \eqref{eq:power_z_ODE}}
			\label{fig:alternative_learning}
		\end{subfigure}
		\hfill
		\begin{subfigure}[b]{0.35\textwidth}
			\centering
			\includegraphics[width=\linewidth]{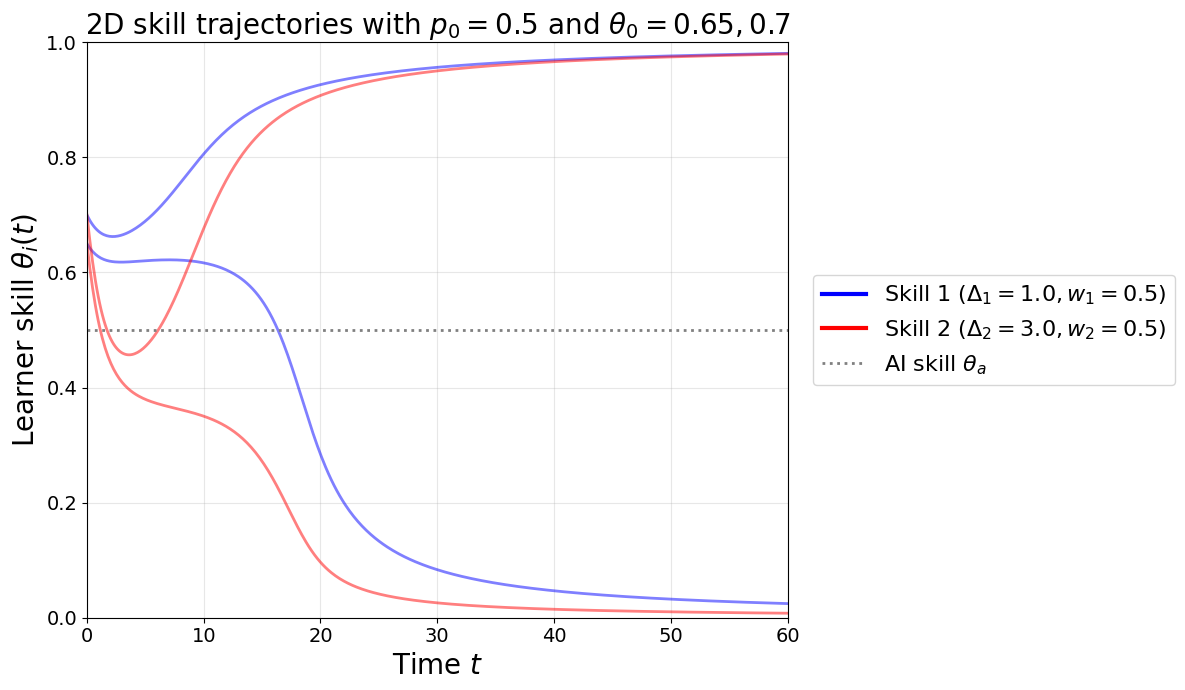}
			\caption{Trajectories of two skills with the same initial $\theta_0$}
			\label{fig:2D_trajectory}
		\end{subfigure}
		\quad \quad
		\begin{subfigure}[b]{0.27\textwidth}
			\centering
			\includegraphics[width=\linewidth]{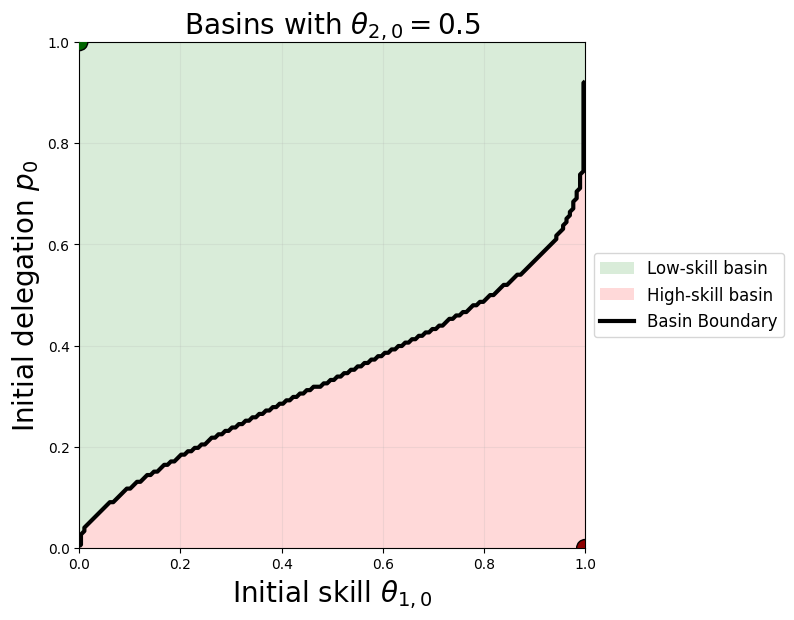}
			\caption{Basins of $(\theta_{1,0}, p_0)$ for fixed $\theta_{2,0}$}
			\label{fig:2D_basin}
		\end{subfigure}
		\caption{Plots illustrating $z$-th power loss with default setings $(\theta_a, \kappa, \Delta) = (0.5, 3,2)$; and the two-skill extension \eqref{eq:two_skill_ODE}, with default settings $(\theta_a,\kappa,\Delta_1,\Delta_2,w_1,w_2)=(0.5,3,1,3,0.5,0.5)$.
			Skill trajectories can still diverge from nearby initial conditions, while $\theta_1(t)>\theta_2(t)$ holds for all $t\ge 0$ because the decay rate satisfies $\Delta_1<\Delta_2$.
		}
		\label{fig:two_skill}
	\end{figure*}
	
	ODE~\eqref{eq:ODE_simplified} represents learner ability by a single scalar
	skill parameter $\theta$. This scalar representation is useful because it
	allows a complete phase-space characterization: the equilibria can be solved
	explicitly, the interior equilibrium is a saddle, and its stable manifold
	forms a curve that partitions the state space into high- and low-skill basins.
	We now show that the same qualitative mechanism extends beyond the
	one-dimensional skill case.
	
	Consider a two-skill model with skill vector
	$(\theta_1,\theta_2)\in[0,1]^2$. Let
	\[
	\bar{\theta}=w_1\theta_1+w_2\theta_2,
	\qquad
	w_1,w_2>0,\quad w_1+w_2=1,
	\]
	denote the task-relevant aggregate skill. The performance loss depends on the
	aggregate skill:
	\[
	\ell(\theta_1,\theta_2)=(1-\bar{\theta})^2.
	\]
	Thus the delegation-update potential becomes
	\[
	\Phi_d(\theta_1,\theta_2,p)
	=
	(1-p)(1-\bar{\theta})^2
	+
	p(1-\theta_a)^2.
	\]
	The skill-update potential is also modified to reflect coordinate-wise skill
	decay. In the normalized setting $\theta_d=0$, we use
	\[
	\Phi_s(\theta_1,\theta_2,p)
	=
	(1-p)(1-\bar{\theta})^2
	+
	\frac{p}{2}\sum_{i=1}^2 \Delta_i\theta_i^2,
	\]
	where $\Delta_i>0$ is the non-use decay rate of skill $i$. The factor $1/2$
	is a normalization convention and can be absorbed into $\Delta_i$. Applying
	the multiplicative-update geometry coordinate-wise gives
	\begin{align}
		\label{eq:two_skill_ODE}
		\begin{aligned}
			\dot{\theta}_i
			&=
			\theta_i(1-\theta_i)
			\Bigl(
			2w_i(1-p)(1-\bar{\theta})
			-
			\Delta_i p\,\theta_i
			\Bigr),
			\qquad i\in\{1,2\},\\
			\dot p
			&=
			\kappa p(1-p)
			\Bigl(
			(1-\bar{\theta})^2-(1-\theta_a)^2
			\Bigr).
		\end{aligned}
	\end{align}
	The delegation update is therefore still driven by the comparison between
	human and AI performance, but human performance now depends on the aggregate
	skill $\bar{\theta}$, while skill decay can differ across coordinates.
	
	This model contains the one-dimensional system as a special case. In
	particular, when $w_1=w_2=1/2$, $\Delta_1=\Delta_2=\Delta$, and
	$\theta_1(0)=\theta_2(0)$, the diagonal subspace
	$\theta_1(t)=\theta_2(t)$ is invariant. On this diagonal,
	$\bar{\theta}(t)=\theta_1(t)=\theta_2(t)$, and
	\eqref{eq:two_skill_ODE} reduces to ODE~\eqref{eq:ODE_simplified}.
	
	\paragraph{Equilibrium and basin structure.}
	We next describe the stationary structure of the two-skill system. The two
	equilibria
	\[
	(\theta_1,\theta_2,p)=(1,1,0),
	\qquad
	(\theta_1,\theta_2,p)=(0,0,1)
	\]
	correspond respectively to the high-skill and low-skill regimes. In addition,
	the interior $p$-nullcline is
	\[
	\bar{\theta}=\theta_a.
	\]
	Combining this condition with the interior $\theta_i$-nullclines yields the
	interior equilibrium
	\begin{align*}
		(\theta_1^\dagger,\theta_2^\dagger,p^\dagger)
		=
		\left(
		\frac{w_1\Delta_2\theta_a}
		{w_1^2\Delta_2+w_2^2\Delta_1},
		\;
		\frac{w_2\Delta_1\theta_a}
		{w_2^2\Delta_1+w_1^2\Delta_2},
		\;
		\frac{
			2(1-\theta_a)(w_1^2\Delta_2+w_2^2\Delta_1)
		}{
			2(1-\theta_a)(w_1^2\Delta_2+w_2^2\Delta_1)
			+
			\Delta_1\Delta_2\theta_a
		}
		\right),
	\end{align*}
	whenever this point lies in $(0,1)^3$. This equilibrium plays the same role
	as the interior saddle in the one-dimensional model: its stable directions
	form part of the boundary separating trajectories that converge to the
	high-skill regime from those that converge to the low-skill regime.
	
	There may also be boundary saddle-type equilibria on faces where one skill is
	fixed at $0$ or $1$. More generally, for $i\neq j$ and $b\in\{0,1\}$, a
	candidate equilibrium on the face $\theta_j=b$ is given by
	\[
	\theta_j=b,\qquad
	\theta_i=\frac{\theta_a-w_j b}{w_i},
	\qquad
	p=
	\frac{
		2w_i^2(1-\theta_a)
	}{
		2w_i^2(1-\theta_a)+\Delta_i(\theta_a-w_j b)
	},
	\]
	provided that $\theta_i\in(0,1)$ and $p\in(0,1)$. Thus, depending on the
	parameters, some of these boundary equilibria may lie outside the feasible
	state space and are absent. Together with the interior equilibrium
	$(\theta_1^\dagger,\theta_2^\dagger,p^\dagger)$, these saddle-type stationary points help organize the separating surface.
	
	Therefore, the geometric picture from the one-dimensional model is observed to persist, but
	with one important change: the basin boundary is no longer a curve in the
	$(\theta,p)$-plane, but a surface in the three-dimensional state space
	$(\theta_1,\theta_2,p)$. Adaptive delegation can still generate bistability
	and path dependence. Learners with similar initial delegation levels but
	slightly different skill vectors may lie on different sides of this separating
	surface and therefore converge to different long-run outcomes.
	
	\paragraph{Simulations.}
	Figure~\ref{fig:two_skill} illustrates the convergence behavior of the
	two-skill dynamics~\eqref{eq:two_skill_ODE}. In
	Figure~\ref{fig:2D_trajectory}, trajectories are initialized with the same
	delegation level but nearby skill vectors. The trajectories can diverge toward
	different equilibria, reflecting sensitivity to initial conditions induced by
	the separating surface. Figure~\ref{fig:2D_basin} fixes one coordinate, for
	example $\theta_{2,0}$, and plots the resulting basins in the
	$(\theta_{1,0},p_0)$-plane. The reduced slice exhibits the same qualitative
	two-basin structure as the one-dimensional model.
	
	The two-skill model also captures heterogeneous decay across skills. For
	example, when $w_1=w_2$, $\theta_1(0)=\theta_2(0)$, and
	$\Delta_1<\Delta_2$, the first skill decays more slowly under delegation.
	Consequently, trajectories satisfy $\theta_1(t)>\theta_2(t)$ for $t>0$
	whenever delegation is positive, producing transient differentiation across
	skills even when the two coordinates ultimately converge to the same long-run
	regime. This shows that the scalar model is not essential for the emergence of
	path dependence; it is the minimal setting in which the full dynamics can be
	characterized explicitly.
	
	\begin{figure*}[t]
		\centering
		\begin{subfigure}[b]{0.32\textwidth}
			\centering
			\includegraphics[width=\linewidth]{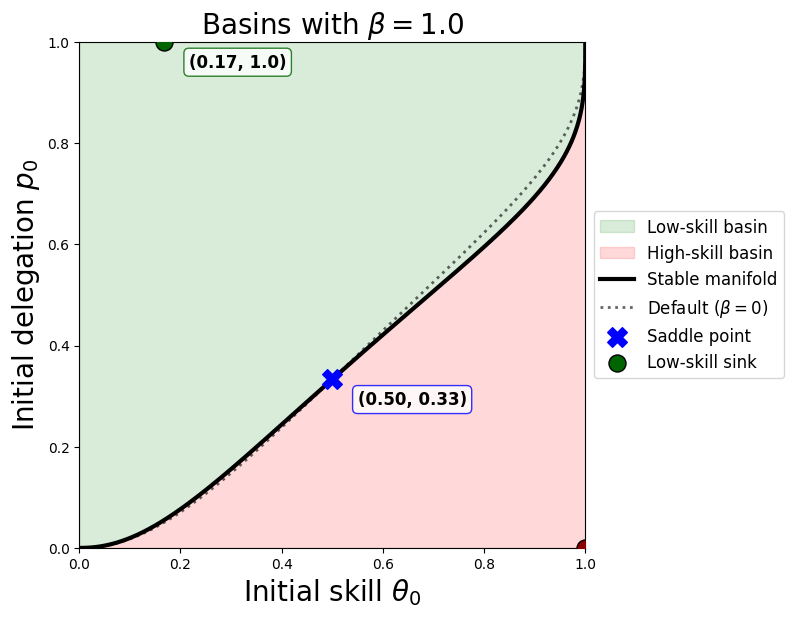}
			\caption{Learning drift from AI \eqref{eq:learning_from_AI}}
			\label{fig:AI_learning}
		\end{subfigure}
		\hfill
		\begin{subfigure}[b]{0.305\textwidth}
			\centering
			\includegraphics[width=\linewidth]{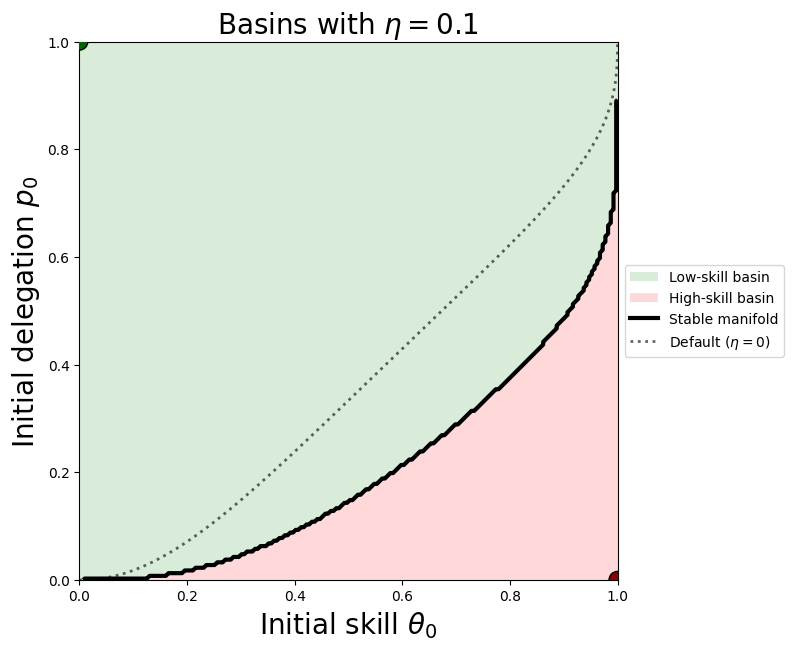}
			\caption{Improved AI skill \eqref{eq:improved_AI}}
			\label{fig:improved_AI_0.1}
		\end{subfigure}
		\hfill
		\begin{subfigure}[b]{0.305\textwidth}
			\centering
			\includegraphics[width=\linewidth]{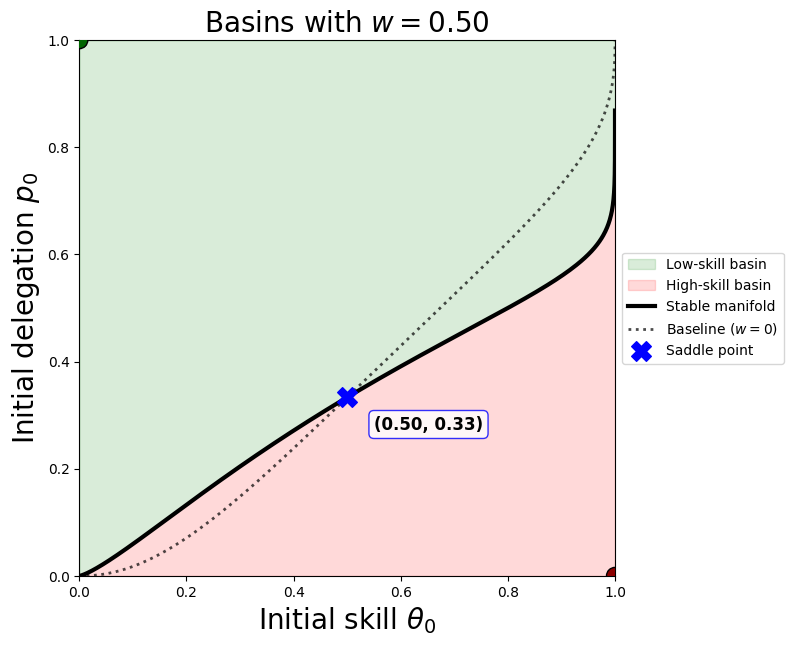}
			\caption{Skill-dependent AI loss \eqref{eq:skill_dependent_AI}}
			\label{fig:skill_dependent_AI}
		\end{subfigure}
		\caption{Plots illustrating how the basins of attraction vary under learner- and AI-side extensions, with default settings $(\theta_a,\kappa,\Delta)=(0.5,3,2)$.
		}
		\label{fig:extension_learner}
	\end{figure*}
	
	\subsection{Learning from AI during delegation}
	\label{sec:learner_AI_extension}
	
	ODE~\eqref{eq:ODE_simplified} assumes that delegation provides no learning
	signal. We now relax this assumption by allowing AI-generated outputs to pull
	the learner's skill toward the AI skill level. This modifies the
	skill-update potential $\Phi_s$, while the performance loss and
	delegation-update potential $\Phi_d$ remain unchanged.
	
	In the normalized setting $\theta_d=0$, we replace the delegation-side decay
	term in $\Phi_s$ by a combination of non-use decay and learning from AI:
	\[
	\Phi_s^\beta(\theta,p)
	=
	(1-p)(1-\theta)^2
	+
	p\left(\Delta\theta^2+\beta(\theta-\theta_a)^2\right),
	\]
	where $\beta\ge 0$ controls the strength of learning from AI. Applying the
	same multiplicative-update geometry and the same time rescaling as in
	ODE~\eqref{eq:ODE_simplified} gives
	\begin{align}
		\label{eq:learning_from_AI}
		\begin{aligned}
			\dot{\theta}
			&=
			\theta(1-\theta)
			\left(
			(1-p)(1-\theta)
			+
			p\bigl(\beta(\theta_a-\theta)-\Delta\theta\bigr)
			\right),\\
			\dot p
			&=
			\kappa p(1-p)\left((1-\theta)^2-(1-\theta_a)^2\right).
		\end{aligned}
	\end{align}
	The case $\beta=0$ recovers ODE~\eqref{eq:ODE_simplified}.
	
	The $p$-nullcline remains $\theta=\theta_a$, since $\Phi_d$ is unchanged. The
	interior $\theta$-nullcline becomes
	\[
	p=\psi_\beta(\theta)
	:=
	\frac{1-\theta}
	{(1-\theta)+(\beta+\Delta)\theta-\beta\theta_a},
	\]
	on its feasible range, and remains monotone decreasing in $\theta$. Thus the
	phase portrait is deformed but retains the same qualitative basin structure.
	
	The main change is that the low-skill regime is lifted. Under full delegation
	$p=1$, the stable skill level becomes
	\[
	\theta_{\mathrm{low}}
	=
	\frac{\beta\theta_a}{\beta+\Delta}.
	\]
	Hence, learning from AI raises the floor of the high-delegation regime but
	does not eliminate it when $\Delta>0$. For learners with $\theta>\theta_a$,
	the term $\beta(\theta_a-\theta)$ also pulls skill downward toward the AI
	level. Figure~\ref{fig:AI_learning} shows that the basin boundary shifts with
	$\beta$, while the two-regime, path-dependent structure is observed to persist.

	\subsection{Improving AI capability}
	\label{sec:improving_AI}
	
	ODE~\eqref{eq:ODE_simplified} assumes a fixed AI skill level $\theta_a$. We now
	relax this assumption by allowing AI capability to improve over time:
	\[
	\dot{\theta}_a
	=
	\rho\theta_a(1-\theta_a),
	\qquad \rho>0 .
	\]
	This specification changes the performance loss and delegation-update potential
	through the time-varying AI loss $(1-\theta_a(t))^2$, while the skill-update
	potential $\Phi_s$ remains unchanged.
	
	The resulting dynamics are
	\begin{align}
		\label{eq:improved_AI}
		\begin{aligned}
			\dot{\theta}
			&=
			\theta(1-\theta)\left((1-p)(1-\theta)-\Delta p\theta\right),\\
			\dot p
			&=
			\kappa p(1-p)\left((1-\theta)^2-(1-\theta_a(t))^2\right),\\
			\dot{\theta}_a
			&=
			\rho\theta_a(1-\theta_a).
		\end{aligned}
	\end{align}
	Equivalently, the delegation-update potential becomes time-dependent:
	\[
	\Phi_d(\theta,p;t)
	=
	(1-p)(1-\theta)^2
	+
	p(1-\theta_a(t))^2 .
	\]
	
	This produces a non-autonomous extension of the baseline $(\theta,p)$ system:
	as $\theta_a(t)$ increases, the performance gap driving delegation changes. 
	A useful interpretation is quasi-static. At each time $t$, the learner faces the phase portrait corresponding to the current AI skill $\theta_a(t)$. As AI capability improves, delegation becomes more attractive and the basin boundary moves accordingly. 
	Thus trajectories that would remain in the high-skill basin
	under a weaker fixed AI may cross the moving boundary and converge instead to the lower-skill, high-delegation regime.
	
	Figure~\ref{fig:improved_AI_0.1} illustrates this effect. Improving AI
	capability does not remove path dependence; instead, it can strengthen the
	delegation feedback loop by progressively increasing the short-run advantage
	of AI assistance.
	
	\subsection{Skill-dependent AI performance}
	\label{sec:skill_dependent_AI}
	
	The baseline model assumes that AI performance depends only on the AI's
	effective skill $\theta_a$. This is appropriate for settings where AI can
	complete a delegated task with little user input, such as routine homework
	problems or fixing isolated bugs. In other settings, AI output quality also
	depends on the user's skill: for example, in writing a paper, a more skilled
	user can provide better prompts, steer the interaction, and revise AI outputs
	more effectively.
	
	We model this by letting the AI-side loss depend on both the learner skill and
	the AI skill:
	\[
	\ell_a(\theta,\theta_a)
	=
	\bigl(1-w\theta-(1-w)\theta_a\bigr)^2,
	\qquad w\in[0,1].
	\]
	Here $w$ measures how much AI-assisted performance depends on the learner's own
	skill. The case $w=0$ recovers the baseline model, where AI loss is
	$(1-\theta_a)^2$. Larger $w$ corresponds to settings where effective AI output
	quality depends more strongly on the user's ability to guide and evaluate the
	AI.
	
	The skill-update potential $\Phi_s$ is unchanged, since independent work still
	generates practice-driven learning and delegation still reduces direct
	practice. The delegation-update potential becomes
	\[
	\Phi_d^{\mathrm{int}}(\theta,p)
	=
	(1-p)(1-\theta)^2
	+
	p\bigl(1-w\theta-(1-w)\theta_a\bigr)^2 .
	\]
	Applying the same multiplicative-update geometry gives
	\begin{align}
		\label{eq:skill_dependent_AI}
		\begin{aligned}
			\dot{\theta}
			&=
			\theta(1-\theta)\left((1-p)(1-\theta)-\Delta p\theta\right),\\
			\dot p
			&=
			\kappa p(1-p)
			\left[
			(1-\theta)^2
			-
			\bigl(1-w\theta-(1-w)\theta_a\bigr)^2
			\right].
		\end{aligned}
	\end{align}
	
	The delegation drift can be written explicitly as
	\[
	(1-\theta)^2
	-
	\bigl(1-w\theta-(1-w)\theta_a\bigr)^2
	=
	(1-w)(\theta_a-\theta)
	\bigl(2-(1+w)\theta-(1-w)\theta_a\bigr).
	\]
	Thus, for $w<1$, the interior $p$-nullcline remains $\theta=\theta_a$, because
	the second factor is positive in the interior of $[0,1]^2$. Since the
	$\theta$-nullcline is unchanged, the interior equilibrium remains
	\[
	(\theta^\dagger,p^\dagger)
	=
	\left(
	\theta_a,\,
	\frac{1-\theta_a}{1-\theta_a+\Delta\theta_a}
	\right).
	\]
	At this point,
	\[
	\partial_\theta \dot p
	=
	-2\kappa(1-w)p^\dagger(1-p^\dagger)(1-\theta_a)<0,
	\]
	for $w<1$ and $\theta_a<1$. Since $\partial_p\dot{\theta}<0$ as in the
	baseline model, the determinant of the Jacobian remains negative, and the
	interior equilibrium remains a saddle. When $w=1$, AI-assisted loss equals the
	learner's own loss, so the delegation drift vanishes identically and the
	adaptive-delegation mechanism degenerates.
	
	Therefore, $w$ does not change the saddle location for $w<1$, but it changes
	the strength of the delegation response. Relative to the baseline signal
	\[
	(1-\theta)^2-(1-\theta_a)^2,
	\]
	the skill-dependent signal is state-dependent:
	\[
	(1-\theta)^2
	-
	\bigl(1-w\theta-(1-w)\theta_a\bigr)^2
	=
	(1-w)
	\frac{2-(1+w)\theta-(1-w)\theta_a}{2-\theta-\theta_a}
	\left((1-\theta)^2-(1-\theta_a)^2\right).
	\]
	Hence larger $w$ weakens the performance gap that drives delegation updates.
	Intuitively, when AI output quality depends more on the learner's own skill,
	AI is less of a purely external substitute for human skill.
	
	Figure~\ref{fig:skill_dependent_AI} illustrates this extension with $w=0.5$. 
	Observe that the boundary shifts upward for low initial skill, reflecting that low-skill learners obtain less effective AI assistance and therefore have a weaker incentive to delegate. 
	For higher initial skill, the boundary can shift
	downward, reflecting that skilled users can extract more value from AI and may maintain reliance over a larger region. 
	Thus, skill-dependent AI performance
	does not uniformly expand or shrink the low-skill basin; instead, it reshapes the basin boundary while preserving the two-regime structure.

	\section{Empirical calibration and falsifiable predictions}
	\label{sec:usability}
	
	We outline how the parameters of the main specification can be estimated from repeated-task data and how the resulting phase portrait could be tested. The
	goal is not to provide a fully calibrated psychometric procedure, but to make
	clear which observables correspond to the model variables and which empirical
	patterns would support or challenge the basin mechanism.
	
	Consider an experimental setup with time windows \(t=1,\ldots,T\). In each
	window \(t\), the learner faces \(m_t\) comparable tasks indexed by
	\(r=1,\ldots,m_t\). For each task, we observe a delegation indicator
	\(X_{t,r}\in\{0,1\}\), where \(X_{t,r}=1\) means that the task is delegated to
	AI, and an evaluation loss \(\ell_{t,r}\in[0,1]\) assigned by the evaluator to
	the realized output. We also assume access to an AI-only benchmark loss
	\(\ell^a_{t,r}\in[0,1]\), obtained either by evaluating the AI on the same
	task or from a held-out benchmark of comparable tasks. All losses are
	normalized to lie in \([0,1]\).
	
	The empirical delegation level in window \(t\) is estimated by the fraction of
	tasks delegated to AI:
	\[
	\widehat p_t
	:=
	\frac{1}{m_t}\sum_{r=1}^{m_t} X_{t,r}.
	\]
	Alternatively, if the experiment elicits the learner's willingness to delegate
	before each window, this stated willingness can be used as a direct measurement
	of \(p(t)\).
	
	\paragraph{Estimating AI effective skill.}
	Under the squared-loss specification, the AI effective skill is defined
	through its expected loss. We estimate it by
	\[
	\widehat L_a
	:=
	\frac{1}{\sum_{t=1}^T m_t}
	\sum_{t=1}^T\sum_{r=1}^{m_t} \ell^a_{t,r},
	\qquad
	\widehat\theta_a
	:=
	1-\sqrt{\widehat L_a}.
	\]
	If AI performance is jagged, this estimate should be interpreted as the
	loss-equivalent effective skill satisfying
	\((1-\widehat\theta_a)^2=\widehat L_a\).
	
	\paragraph{Estimating learner skill.}
	For each window \(t\), let
	\[
	\widehat L^H_t
	:=
	\frac{1}{|\mathcal H_t|}
	\sum_{r\in\mathcal H_t}\ell_{t,r},
	\qquad
	\mathcal H_t:=\{r:X_{t,r}=0\},
	\]
	be the average loss on independently completed tasks. When
	\(|\mathcal H_t|>0\), we estimate the learner's skill by
	\[
	\widehat\theta_t
	:=
	1-\sqrt{\widehat L^H_t}.
	\]
	If a window contains no independently completed tasks, or if independently
	completed tasks are selected non-randomly by the learner, \(\theta_t\) is not
	identified from submitted outputs alone. In practice, \(\theta_t\) should be
	estimated using randomized no-AI probe tasks sampled from the same task
	distribution, possibly combined with nearby windows or a smoothing procedure.
	
	\paragraph{Estimating the decay parameter \(\Delta\).}
	Using the normalized ODE~\eqref{eq:ODE_simplified}, and taking each time window
	as one time unit, the skill dynamics imply the Euler approximation
	\[
	\widehat\theta_{t+1}-\widehat\theta_t
	\approx
	\widehat\theta_t(1-\widehat\theta_t)
	\left(
	(1-\widehat p_t)(1-\widehat\theta_t)
	-
	\Delta \widehat p_t\widehat\theta_t
	\right).
	\]
	For windows with reliable estimates of both \(\widehat\theta_t\) and
	\(\widehat\theta_{t+1}\), define
	\[
	a_t
	:=
	\widehat p_t\widehat\theta_t^2(1-\widehat\theta_t),
	\qquad
	b_t
	:=
	\widehat\theta_t(1-\widehat\theta_t)(1-\widehat p_t)(1-\widehat\theta_t)
	-
	(\widehat\theta_{t+1}-\widehat\theta_t).
	\]
	Then \(b_t\approx \Delta a_t\). A least-squares estimate is
	\[
	\widehat\Delta
	=
	\frac{\sum_t a_t b_t}{\sum_t a_t^2},
	\]
	where the sum is taken over windows with \(a_t>0\) and reliable skill
	estimates.
	
	\paragraph{Estimating the delegation adaptation rate \(\kappa\).}
	The delegation dynamics imply
	\[
	\widehat p_{t+1}-\widehat p_t
	\approx
	\kappa \widehat p_t(1-\widehat p_t)
	\left(
	(1-\widehat\theta_t)^2-(1-\widehat\theta_a)^2
	\right).
	\]
	Let
	\[
	c_t
	:=
	\widehat p_t(1-\widehat p_t)
	\left(
	(1-\widehat\theta_t)^2-(1-\widehat\theta_a)^2
	\right).
	\]
	Then \(\widehat p_{t+1}-\widehat p_t\approx \kappa c_t\), and we estimate
	\[
	\widehat\kappa
	=
	\frac{\sum_t c_t(\widehat p_{t+1}-\widehat p_t)}
	{\sum_t c_t^2},
	\]
	using windows where \(c_t\) is not close to zero.
	
	\paragraph{Using the fitted model.}
	After estimating \((\widehat\theta_a,\widehat\kappa,\widehat\Delta)\) and the
	current state \((\widehat\theta_t,\widehat p_t)\), the fitted ODE can be used
	to predict whether the learner lies in the high-skill or low-skill basin. In
	particular, one can compute the saddle point and an approximation to its stable
	manifold, then compare the learner's current delegation level with the critical
	delegation threshold. This translates observed repeated-task data into the
	phase-space predictions of the model.
	
	\paragraph{Illustrative calibration example.}
	To illustrate the calibration steps, consider a fabricated repeated-task dataset
	with \(T=4\) windows and \(m_t=50\) comparable tasks per window. Suppose the
	observed delegation fractions and average independent-task losses are as shown
	in Table~\ref{tab:synthetic_calibration}, and suppose the average AI-only
	benchmark loss across the same task family is \(\widehat L_a=0.04\). Then
	\[
	\widehat\theta_a = 1-\sqrt{0.04}=0.8.
	\]
	Using the independent-task losses gives the skill estimates
	\[
	\widehat\theta_1=0.70,\qquad
	\widehat\theta_2=0.68,\qquad
	\widehat\theta_3=0.64,\qquad
	\widehat\theta_4=0.60.
	\]
	Substituting these values into the least-squares formulas above yields the
	illustrative estimates
	\[
	\widehat\Delta \approx 2.38,
	\qquad
	\widehat\kappa \approx 4.59.
	\]
	The corresponding estimated saddle is
	\[
	(\widehat\theta^\dagger,\widehat p^\dagger)
	=
	\left(
	\widehat\theta_a,\,
	\frac{1-\widehat\theta_a}{1-(1-\widehat\Delta)\widehat\theta_a}
	\right)
	\approx
	(0.80,\,0.096).
	\]
	Since \(\widehat{\psi}\) is nondecreasing and
	\(\widehat{\psi}(0.80)=\widehat p^\dagger\approx0.096\), we have
	\[
	\widehat{\psi}(0.70)\le 0.096<0.20.
	\]
	Thus the initial state
	\((\widehat\theta_1,\widehat p_1)=(0.70,0.20)\) lies above the estimated
	boundary and is classified as being in the low-skill basin.

	\begin{table}[h]
		\centering
		\caption{Synthetic repeated-task data used only to illustrate the calibration
			procedure.}
		\label{tab:synthetic_calibration}
		\begin{tabular}{cccc}
			\toprule
			Window \(t\) & Delegation fraction \(\widehat p_t\) & Avg. independent loss \(\widehat L^H_t\) & Skill estimate \(\widehat\theta_t\)\\
			\midrule
			1 & 0.20 & 0.09   & 0.70 \\
			2 & 0.24 & 0.1024 & 0.68 \\
			3 & 0.30 & 0.1296 & 0.64 \\
			4 & 0.38 & 0.16   & 0.60 \\
			\bottomrule
		\end{tabular}
	\end{table}
	
	\paragraph{Empirical predictions and falsification.}
	The fitted model makes three testable predictions. First, learners whose
	estimated states lie above the fitted separatrix should move toward higher
	delegation and lower independent skill, while learners below it should move
	toward lower delegation and sustained skill growth. Second, learners near the
	estimated boundary should be sensitive to small changes in early reliance:
	delayed access, required independent practice, or conservative delegation
	feedback should shift some trajectories from the predicted low-skill basin to
	the high-skill basin. Third, increasing AI capability should lower the fitted
	boundary, so the same initial state is more likely to fall in the low-skill
	basin under a stronger AI system.
	
	The basin mechanism would be challenged if estimated basin membership had no
	predictive relationship with long-run independent performance, if early
	reliance had no effect on later skill after conditioning on initial skill, or
	if changes in AI capability did not move the empirical threshold in the
	predicted direction.
	
\end{document}